\documentclass[11pt,pdfa,letterpaper]{article}
\newif\ifcomments

\commentsfalse

\usepackage[in]{fullpage}

\usepackage{iftex}
\ifPDFTeX
  \usepackage[utf8]{inputenc}
  \usepackage[noTeX]{mmap}
  \usepackage[T1]{fontenc}
\fi
\ifLuaTeX
  \usepackage{luatex85}
  \usepackage[noTeX]{mmap}
\fi

\usepackage{comment}
\usepackage{amsfonts}
\usepackage{amsmath}
\usepackage{mathtools,amsthm,amssymb} %
\usepackage{xcolor,xstring,xparse}
\usepackage{graphicx}
\usepackage{comment}
\usepackage{qtree}
\usepackage{tree-dvips}
\usepackage{float}
\definecolor{linkblue}{HTML}{001487}
\usepackage[colorlinks=true,allcolors=linkblue]{hyperref}
\usepackage[nameinlink,capitalize,noabbrev]{cleveref}
\usepackage{braket}
\usepackage{mathrsfs}
\usepackage{tikz}
\usepackage{qcircuit}
\usepackage{xspace}
\usepackage{dsfont}
\usepackage{enumitem} %
\usepackage{longfbox}
\usepackage{csquotes}
\usepackage{authblk}
\usepackage{longfbox}

\usepackage{adjustbox}
\usepackage{gensymb}
\usepackage{siunitx}
\usepackage{stmaryrd}
\usepackage{float}
\usepackage{amsthm}
\usepackage{tikz-cd}
\usepackage{tikz}
\usetikzlibrary{backgrounds}
\usepackage{hyperref}

\hypersetup{
    colorlinks = true,
    linkcolor = violet,
    linkbordercolor = {white}
}

\usepackage{booktabs}    
\usepackage{tabularx}    

\usepackage{geometry}
\usepackage[style=alphabetic,minalphanames=3,maxalphanames=4,maxnames=99,backref=true]{biblatex}

\usepackage{url}
\usepackage{bbm}
\usepackage{hyperref} 
\hypersetup{breaklinks=true}

\usetikzlibrary{fadings}
\usetikzlibrary{patterns}
\usetikzlibrary{shadows.blur}
\usetikzlibrary{shapes}

\usepackage{framed}

\usepackage{caption}
\ifcomments
  \newcommand{\andrey}[1]{{\color{green}\bf (Andrey: #1)}} 
  \newcommand{\akshar}[1]{{\color{red}\bf (Akshar: #1)}} 
  \newcommand{\jon}[1]{{\color{blue}\bf (Jonathan: #1)}}
  \newcommand{\alex}[1]{{\color{purple}\bf (Alex: #1)}}
  \newcommand{\vinod}[1]{{\color{red}\bf (Vinod: #1)}}
\else
  \newcommand{\andrey}[1]{}
  \newcommand{\akshar}[1]{}
  \newcommand{\jon}[1]{}
  \newcommand{\alex}[1]{}
  \newcommand{\vinod}[1]{}
\fi

\theoremstyle{plain}
\newtheorem{theorem}{Theorem}[section]
\newtheorem*{theorem*}{Theorem}
\newtheorem{lemma}[theorem]{Lemma}
\newtheorem{claim}[theorem]{Claim}
\Crefname{claim}{Claim}{Claims}

\newtheorem*{lemma*}{Lemma}

\newtheorem{corollary}[theorem]{Corollary}
\newtheorem*{corollary*}{Corollary}
\newtheorem{proposition}[theorem]{Proposition}
\newtheorem{conjecture}[theorem]{Conjecture}

\newtheorem{question}[theorem]{Question}

\theoremstyle{definition}
\newtheorem{definition}[theorem]{Definition}
\newtheorem{problem}[theorem]{Problem}

\theoremstyle{remark}
\newtheorem{remark}[theorem]{Remark}
\numberwithin{equation}{section}

\newcommand{\F}{\mathbb{Z}}
\newcommand{\R}{\mathbb{R}}

\newcommand{\Z}{\mathbb{Z}}
\newcommand{\C}{\mathbb{C}}

\newcommand{\Mod}[1]{\ (\mathrm{mod}\ #1)}

\DeclareMathOperator{\poly}{poly}
\DeclareMathOperator{\negl}{negl}

\newcommand{\ketbra}[2]{|#1\rangle\langle#2|}

\renewcommand{\vec}[1]{\mathbf{#1}}
\renewcommand{\Mod}[1]{\ (\mathrm{mod}\ #1)}

\newcommand{\mb}[1]{\mathbf{#1}}

\newcommand{\br}[1]{\left( #1 \right)}

\renewcommand{\exp}[1]{\text{exp}\br{#1}} 

\newcommand{\bounds}[2]{\bigg\rvert_{#1}^{#2}}
\newcommand{\boudns}[2]{\bounds} 

\renewcommand{\a}{\alpha}
\renewcommand{\b}{\beta}

\renewcommand{\d}{\delta}

\newcommand{\e}{\epsilon}

\newcommand{\m}{\mu}

\renewcommand{\r}{\rho}

\renewcommand{\t}{\tau}

\newcommand{\w}{\omega}
\newcommand{\W}{\Omega}

\newcommand{\CA}{\mathcal{A}}
\newcommand{\CB}{\mathcal{B}}
\newcommand{\CC}{\mathcal{C}}
\newcommand{\CD}{\mathcal{D}}
\newcommand{\CE}{\mathcal{E}}

\newcommand{\CO}{\mathcal{O}}
\newcommand{\CP}{\mathcal{P}}
\newcommand{\CR}{\mathcal{R}}
\newcommand{\CS}{\mathcal{S}}
\newcommand{\CT}{\mathcal{T}}
\newcommand{\CV}{\mathcal{V}}

\DeclareMathOperator{\Unif}{Unif}

\DeclareMathOperator{\Log}{Log}
\DeclareMathOperator{\Stab}{Stab}
\DeclareMathOperator{\wt}{wt}

\DeclareMathOperator{\TV}{\mathsf{TV}}

\newcommand{\Ber}{\mathsf{Ber}}
\newcommand{\Bin}{\mathsf{Bin}}

\newcommand{\abs}[1]{\left| #1 \right|}

\NewDocumentCommand{\lsn}{ O{} O{} O{} }{%
  \ensuremath{%
    \mathsf{LSN}%
    \ifblank{#2}{}{^{#2}}%
    \ifblank{#3}{}{_{#3}}  
    \ifblank{#1}{}{\!\left(#1\right)}%
  }%
}

\NewDocumentCommand{\slsn}{ O{} O{} O{} }{%
  \ensuremath{%
    \mathsf{stateLSN}%
    \ifblank{#2}{}{^{#2}}%
    \ifblank{#3}{}{_{#3}}  
    \ifblank{#1}{}{\!\left(#1\right)}%
  }%
}

\NewDocumentCommand{\lpn}{ O{} }{%
  \ensuremath{%
    \mathsf{LPN}%
    \ifblank{#1}{}{\!\left(#1\right)}%
  }%
}

\NewDocumentCommand{\symplpn}{ O{} }{%
  \ensuremath{%
    \mathsf{sympLPN}%
    \ifblank{#1}{}{\!\left(#1\right)}%
  }%
}

\NewDocumentCommand{\qncp}{ O{} }{%
  \ensuremath{%
    \mathsf{QNCP}%
    \ifblank{#1}{}{\!\left(#1\right)}%
  }%
}

\NewDocumentCommand{\errqncp}{ O{} }{%
  \ensuremath{%
    \mathsf{errQNCP}%
    \ifblank{#1}{}{\!\left(#1\right)}%
  }%
}

\NewDocumentCommand{\recqncp}{ O{} }{%
  \ensuremath{%
    \mathsf{recQNCP}%
    \ifblank{#1}{}{\!\left(#1\right)}%
  }%
}

\NewDocumentCommand{\ncp}{ O{} }{%
  \ensuremath{%
    \mathsf{NCP}%
    \ifblank{#1}{}{\!\left(#1\right)}%
  }%
}

\NewDocumentCommand{\qsdp}{ O{} }{%
  \ensuremath{%
    \mathsf{QSDP}%
    \ifblank{#1}{}{\!\left(#1\right)}%
  }%
}

\DeclareMathOperator{\Symp}{Symp}

\makeatletter
\renewcommand{\paragraph}{%
  \@startsection{paragraph}{4}%
  {\z@}{2.25ex \@plus 1ex \@minus .2ex}{-1em}%
  {\normalfont\normalsize\bfseries}%
}
\makeatother
 \title{Average-Case Complexity of Quantum Stabilizer Decoding}

\author[1]{\quad\quad Andrey Boris Khesin}
\author[2]{Jonathan Z. Lu\footnote{\url{lujz@mit.edu}}}
\author[3]{Alexander Poremba\footnote{\url{poremba@bu.edu}}}
\author[4]{\newline Akshar Ramkumar\footnote{\url{aramkuma@caltech.edu}}}
\author[5]{Vinod Vaikuntanathan}
\affil[1]{Department of Computer Science, University of Oxford}
\affil[2]{Department of Mathematics, Massachusetts Institute of Technology}
\affil[3]{Department of Computer Science \& Department of Physics, Boston University}
\affil[4]{Department of Mathematics, California Institute of Technology}
\affil[5]{EECS \& CSAIL, Massachusetts Institute of Technology}

\begin{document}

\pagenumbering{gobble}

\clearpage\maketitle
\thispagestyle{empty}

\begin{abstract}
Random classical linear codes are widely believed to be hard to decode. 
While slightly sub-exponential time algorithms exist when the coding rate vanishes sufficiently rapidly, all known algorithms at constant rate require exponential time.
By contrast, the complexity of decoding a random quantum stabilizer code has remained an open question for quite some time.
This work closes the gap in our understanding of the algorithmic hardness of decoding random quantum versus random classical codes.
We prove that decoding a random stabilizer code with even a \emph{single} logical qubit is at least as hard as decoding a random classical code at constant rate---the maximally hard regime.
This result suggests that the easiest random quantum decoding problem is at least as hard as the hardest random classical decoding problem, and shows that any sub-exponential algorithm decoding a typical stabilizer code, at \emph{any} rate, would immediately imply a breakthrough in cryptography.

More generally, we also characterize many other complexity-theoretic properties of stabilizer codes.
While classical decoding admits a random self-reduction, we prove significant barriers for the existence of random self-reductions in the quantum case.
This result follows from new bounds on Clifford entropies and Pauli mixing times, which may be of independent interest.
As a complementary result, we demonstrate various other self-reductions which are in fact achievable, such as between search and decision.
We also demonstrate several ways in which quantum phenomena, such as quantum degeneracy, force several reasonable definitions of stabilizer decoding---all of which are classically identical---to have distinct or non-trivially equivalent complexity. 
\end{abstract}

\newpage
\tableofcontents

\newpage
\pagenumbering{arabic} 
\setcounter{page}{1}   

\section{Introduction}

In the classical realm, coding theory has made a remarkable impact across many different areas of mathematics and computer science, far beyond its original use in telecommunications and storage. 
Presently, the span of coding-theoretic applications range across computational complexity theory (e.g. the PCP theorem and the hardness of approximation~\cite{10.1145/278298.278306,DBLP:journals/jacm/FeigeGLSS96,10.1145/1236457.1236459,}), cryptography (e.g. as a basis of hardness for cryptographic protocols~\cite{Sha79,Stern93,FS96,McEliece1978,Alekhnovich03}), combinatorics and extremal set theory (e.g. packing/covering problems~\cite{hamming1950error, plotkin1960binary} and list-decodable codes~\cite{10.1145/1132516.1132518}), algorithms (e.g. sketching and streaming algorithms~\cite{10.1145/1147954.1147955,10.1109/FOCS.2008.82}), and more. 

One of the key motivations for the development of coding theory which has remained relevant due to its presence in most of the above applications is the search for \emph{asymptotically good} codes---families of error correcting codes which have both constant rate and linear distance asymptotically.
This search naturally led to the rigorous study of \emph{random codes}; in particular, random \emph{linear} codes. 
A classic result of coding theory known as the Gilbert-Varshamov bound is that a random linear code is asymptotically good with high probability~\cite{Gilbert1952ACO, Varshamov}. 
Yet, despite their extraordinary error correcting properties, random linear codes are notoriously hard to decode in practice, and the fastest known algorithms still run in exponential or nearly exponential time~\cite{BKW03,10.1007/11538462_32}.
This longstanding barrier against efficient algorithms has led to the formulation of the \emph{Learning Parity with Noise} ($\lpn$)~\cite{BFKL93} problem---the task of decoding random linear codes in the presence of noise---which has since become a popular hardness assumption upon which much cryptography~\cite{10.5555/647097.717000,Alekhnovich03,10.1007/11535218_18,10.1007/978-3-642-03356-8_35,10.1007/978-3-642-34961-4_40,DBLP:conf/cans/DavidDN14,applebaum_et_al:LIPIcs.ITCS.2017.7,BLVW} and learning theory~\cite{BFKL93,10.1109/FOCS.2006.51} has been built.
Formally, given
\begin{align}\label{eq:lpn-sample}
    (\vec A \sim \Z_2^{n\times k},\;\vec A \cdot \vec x +  \vec e \Mod{2}) 
\end{align}
for a Bernoulli error $\vec e \sim \Ber(p)^{\otimes n}$ of noise rate $p\in (0,\frac{1}{2})$, the $\lpn(k, n, p)$ assumption\footnote{This is the so-called \emph{search} variant of $\lpn$ problem; the \emph{decision} variant of $\lpn(k, n, p)$ states that it is computationally difficult to distinguish a sample as in~\eqref{eq:lpn-sample} from an unstructured random sample $(\vec A \sim \Z_2^{n\times k},\vec u \sim \Z_2^n)$.} states that it is hard to recover the randomly chosen message $\vec x \sim \Z_2^k$; here, $\vec A$ serves as a generator matrix for a linear code. In terms of parameters, a typical choice is one in which $n = \poly(k)$ and $p$ is constant.

The hardness of the $\lpn$ problem is well-supported and enjoys a firm theoretical foundation.
For example,~\cite{BLVW} gave a random self-reduction into $\lpn$ by proving that $\lpn$ with a very high noise rate $p$ is at least as hard as a its worst case variant---known as the \emph{nearest codeword problem} ($\mathsf{NCP}$)---under some structural assumptions about the worst-case instance.
In addition,~\cite{10.1007/s00145-010-9061-2} gave a search-to-decision reduction which renders both the search and the decision variants of the problem equivalent. 
Despite several decades of study, the $\lpn$ problem has seen virtually no asymptotic algorithmic improvements since the nearly-exponential time approaches of~\cite{BKW03,10.1007/11538462_32}.
It is widely conjectured that $\lpn$ is computationally intractable even for quantum computers~\cite{ghosal2025post}.

In the quantum realm, the development of quantum coding theory has focused heavily on the goal of constructing asymptotically good and practical 
quantum error-correcting codes~\cite{PhysRevA.52.R2493,gottesman2009introductionquantumerrorcorrection}, primarily as a means of achieving fault-tolerance in a race to achieve practical logical quantum computation~\cite{doi:10.1137/S0097539799359385,doi:10.1126/science.279.5349.342}. 
Just as classical coding theory focuses primarily on linear codes, so too does quantum coding theory focus primarily on \emph{stabilizer codes}, the quantum analog of linear codes~\cite{gottesman1997stabilizercodesquantumerror}. 
As in the classical case, random stabilizer codes are also \emph{asymptotically good} with high probability \cite{Graeme_thesis,gottesman_book} and in general offer valuable insights into the fundamental limitations of quantum codes, such as achievable rates, code distances, or thresholds. 
A stabilizer code can be specified by a Clifford operator $\vec{C} \in \CC_n$ the same way a classical linear is specified by a generator matrix.
A $\llbracket n, k \rrbracket$ stabilizer code specified by $\vec C \in \CC_n$ has code states of the form $\vec C (\ket{0^{n-k}}\otimes \ket{\psi})$, where $\ket{\psi}$ is a $k$-qubit state.
To generate a random $\llbracket n, k \rrbracket$ stabilizer code, we simply choose a uniformly random Clifford $\vec C \sim \CC_n$ ~\cite{Graeme_thesis,poremba2025learningstabilizersnoiseproblem}; this is always possible to do efficiently, since $\CC_n$ is a finite group of order $2^{O(n^2)}$.

In sharp contrast with classical coding theory, little is known about the hardness of decoding random quantum stabilizer codes. 
This is somewhat surprising because random stabilizer codes are ubiquitous in quantum information science: they appear 
in the theory of quantum communication~\cite{Graeme_thesis}, quantum authentication and the verification of quantum computations~\cite{aharonov2017interactiveproofsquantumcomputations}, quantum cryptography~\cite{dulek2018quantumciphertextauthenticationkey} and quantum gravity~\cite{Hayden_Preskill_2007,Harlow_2013}.
Given the historical impact of $\lpn$ across numerous areas of computer science, one may reasonably expect that the hardness of random quantum stabilizer decoding---if properly understood---could occupy a similar foundational role in quantum information science, more broadly. 
We are therefore motivated to study the fundamental question:
\begin{question} \label{question:main_question}
    How hard is it to decode a typical quantum stabilizer code?
\end{question}
This question has remained open for quite some time~\cite{HsiehLeGall11,iyer2015hardness,kuo2020hardnesses}, in part because characterizing the complexity of certain quantum coding tasks has proven non-trivial even in the worst case.
In fact, all prior techniques to determine the hardness of stabilizer decoding have been limited to {\em worst-case} syndrome decoding problems~\cite{HsiehLeGall11,iyer2015hardness}. 
An answer to the above question would not only sharpen our understanding of quantum stabilizer decoding in \emph{typical} instances, but may also provide new insights into the complexity of quantum information processing tasks~\cite{ bostanci2023unitarycomplexityuhlmanntransformation}, or whether it is possible to base cryptography on hard problems in \emph{quantum}, rather than classical, error-correction~\cite{10.1007/978-3-642-27660-6_9}.

\begin{figure}[t]
\centering

\renewcommand{\arraystretch}
{1.2} 
\setlength{\tabcolsep}{6pt}       

{\begin{tabular}{@{} l @{\hspace{8pt}} l @{\hspace{20pt}} l @{}}
\toprule
 & Learning Stabilizers with Noise & Learning Parity with Noise \\
\midrule
\textbf{Input:} &
\(\big(\vec C \sim \mathcal{C}_n,\; \vec E\,\vec C \ket{0^{n-k}, \vec x}\big)\) &
\(\big(\vec A \sim \mathbb{Z}_2^{n\times k},\; \vec A\vec x + \vec e \Mod{2}\big)\) \\
\textbf{Noise:} &
Depolarizing noise \(\vec E \sim \mathcal{D}_p^{\otimes n}\) &
Bernoulli noise \(\vec e \sim \text{Ber}_p^{\otimes n}\) \\
\textbf{Task:} &
Find \(\vec x \sim \mathbb{Z}_2^k\) &
Find \(\vec x \sim \mathbb{Z}_2^k\) \\
\bottomrule
\end{tabular}}
\caption{{Comparison between $\lsn$~\cite{poremba2025learningstabilizersnoiseproblem} and $\lpn$~\cite{BFKL93}. In both cases, the input features a classical description of a \emph{random} code; in the case of $\lsn$, it is specified by a random \emph{Clifford} encoding of a quantum stabilizer code, whereas in the case of $\lpn$ it consists of a random \emph{generator matrix} of a classical linear code. The problem parameters are characterized by the logical (qu)bits $k$, physical (qu)bits $n$, and the error probability per physical (qu)bit $p$.}}
\label{tab:lpn-lsn-comparison}
\end{figure}

In a first attempt to formally explore Question~\ref{question:main_question}, \cite{poremba2025learningstabilizersnoiseproblem} recently defined the \emph{Learning Stabilizers with Noise} ($\lsn$) problem as a natural quantum analog of $\lpn$; this is the task of decoding a random quantum stabilizer code in the presence of noise.
More precisely, the task is to find $\vec x \sim \Z_2^k$ given
\begin{align}
    \left(\vec C \sim  \CC_n, \, \vec E \, \vec C \ket{0^{n-k},\vec x}\right)  
\end{align}
where ${\vec C} \sim  \CC_n$ is a random $n$-qubit Clifford (say, described by a circuit) and $\vec E \, \vec C \ket{0^{n-k},\vec x}$ is a noisy quantum \emph{stabilizer state}, for some Pauli error ${\vec E} \sim \mathcal{D}_p^{\otimes n}$  from a local depolarizing channel.
\cite{poremba2025learningstabilizersnoiseproblem} showed that $\lsn$ can be solved in exponential time using only a single sample by way of the \emph{pretty good measurement}~\cite{barnum2000reversingquantumdynamicsnearoptimal,montanaro2019pretty}.
In addition, they showed that $\lsn$ is contained in the unitary complexity class $\mathsf{unitaryBQP}^{\mathsf{NP}}$~\cite{ bostanci2023unitarycomplexityuhlmanntransformation}. 
While these results provide an initial assessment of the hardness of the $\lsn$ problem through the lens of both algorithms and complexity, the central issue of how to characterize (or lower bound) the average-case hardness itself, as raised in Question~\ref{question:main_question}, remains unresolved. 

More formally, we may break down Question~\ref{question:main_question} into more fine-grained questions about the hardness of random stabilizer decoding.
The first is a question of lower-bounding quantum decoding by more familiar problems.
\begin{question} \label{question:LPN_to_LSN?}
    Is random stabilizer decoding at least as hard as other (e.g. classical) average-case tasks which are widely believed to be hard?
\end{question}
A positive answer would provide rigorous theoretical foundations towards the complexity of decoding stabilizer codes in \emph{typical} instances.  
Closely related to the question of hardness lower bounds is the existence of non-trivial decoding algorithms; in particular:
\begin{question} \label{question:attack_algorithms}
    Do random quantum stabilizer codes admit any sub-exponential time decoders, similar to random linear codes?
\end{question}
While sub-exponential algorithms for certain parameter regimes of $\lpn$ are known~\cite{BKW03,10.1007/11538462_32}, no such algorithm has ever been discovered for decoding a random quantum stabilizer code in \emph{any} (non-trivial) parameter regime.
This is particularly surprising, given that stabilizer decoding in quantum error correction has been studied intensely from many angles for the past three decades.

Another important structural component in the characterization of the average-case complexity of a problem is the existence of \emph{self reductions}: either between decision and search variants, or between worst-case and average-case variants.
Such reductions certify the robustness of the problem to certain useful deformations (e.g. distinguish two cases vs. find a solution) in how the problem is defined.
Both of these, for example, have been studied extensively in the case of $\lpn$~\cite{brakerski2019worst,10.1007/s00145-010-9061-2}. This naturally raises the question:
\begin{question} \label{question:self_reducibility}
    Is stabilizer decoding self-reducible? That is, does stabilizer decoding admit search-to-decision, decision-to-search, and random self-reductions?
\end{question}

\subsection{Contributions and Scope}

Our limited understanding of random stabilizer codes reflects broader gaps in the foundations of quantum error correction.
These gaps have important implications for many different areas of research, including coding theory, algorithms, complexity theory, and cryptography. 
The overarching goal of this work is to expose these limitations and to develop the necessary mathematical techniques to overcome them.
Concretely, we seek to
\begin{itemize}
    \item bring the study of random quantum stabilizer codes to the same theoretical footing as those of random classical linear codes, and
    
    \item to lay the theoretical foundations for how to view problems in quantum error correction, more generally, through the lens of computational complexity and cryptography.
\end{itemize}

To that end, we show the following results about decoding random quantum stabilizer codes, which collectively answer Question~\ref{question:main_question} to roughly the same degree as our understanding of random classical linear codes.
First, to Question~\ref{question:LPN_to_LSN?}, we prove that decoding a random stabilizer code with a \emph{single} logical qubit is at least as hard as decoding a random classical linear code with constant rate $R = \frac{k}n = \Theta(1)$.
More precisely, we reduce $\lpn$ with $n = ck$, where $c > 1$ is a chosen constant to $\lsn$ with \emph{any} rate, e.g. with a single logical qubit $k = 1$, such that the probability parameters $p$ of each differ only by a constant.
This result lower bounds the hardness of $\lsn$ by $\lpn$ itself, which is widely conjectured (with significant evidence) to be computationally intractable.
In fact, $\lpn$ with constant rate is the hardest regime of classical random decoding, in the sense that no sub-exponential algorithms for $\lpn$ have been devised in this regime (for non-trivial, e.g. constant, error probability $p$) have been found after decades of effort.
Therefore, our reduction also provides evidence that the answer to Question~\ref{question:attack_algorithms} is negative.
Specifically, since $\lpn$ at constant rate reduces to any instance of $\lsn$, even for one logical qubit, \emph{any} sub-exponential algorithm constructed for $\lsn$ would also serve as a sub-exponential algorithm for $\lpn$ at constant rate.
Such an algorithm would imply significant breakthrough in classical complexity theory and cryptography, which assumes and relies on the strong hardness of $\lpn$.

Next, we explore Question~\ref{question:self_reducibility} across several directions. We prove that, both in the worst and average case, stabilizer decoding satisfies a search-to-decision reduction.
The reverse notion, a decision-to-search reduction, is seldom considered in classical literature because it is usually trivial.
We observe, however, that for $\lsn$ this reduction is not trivial because a proposed solution is not clearly verifiable.
Nonetheless, we are able to prove a decision-to-search reduction, showing that the decision and search versions of stabilizer decoding are equivalent in a certain sense.
On the other hand, we prove significant barriers against the existence of a random self-reduction for stabilizer decoding.
More precisely, an operator implementing a reasonable reduction should have bounded sparsity, because a completely dense operator will scramble the noise to an undecodable level.
Yet, we prove that no sparse operators can implement a reduction.
Similarly, natural reduction operators scramble the stabilizer tableau itself, in addition to the underlying code space.
However, we prove also that under reasonable assumptions about the structure of the reduction itself, no such reduction exists either.
These results together imply that random self-reductions for $\lsn$ must either have an exotic structure so as to overcome these barriers, or not exist at all.

One consequence of these results is that they highlight the many \emph{qualitative} differences between quantum and classical decoding, which pushes back against folklore wisdom within quantum error correction that decoding quantum codes is merely a slight generalization of decoding classical codes.
Furthermore, a useful facet of the proof techniques we use to answer the above questions is that they give explicit interpretations as to the specific quantum phenomena which cause $\lsn$ to differ so strongly from $\lpn$.
These effects include quantum degeneracy, high uniformity in the mixing capabilities of entangling operators, and the non-commutative nature of unitary operators.
We highlight the particular quantum effects at play in our proofs.

Finally, we also discuss in detail the many natural ways to formulate the problem of decoding a stabilizer code.
For example, these include outputting the logical state, finding a valid recovery operator, or retrieving the exact Pauli error applied.
All such definitions are equivalent on a classical code, and yet quantumly some definitions appear to be strictly harder problems than others, via the implicit and unnecessary embedding of short vector problems.
This subtlety suggests that care is required in defining quantum decoding in the first place, as the task does not enjoy a strong amount of definition robustness as compared classical decoding.

\subsection{Technical Overview}
\label{sec:technical-overview}

Before we introduce the full definitions and proofs, we give an overview of the techniques and key insights in this paper.

\subsubsection{Random Classical versus Quantum Stabilizer Decoding}

The precise definition of decoding a stabilizer code varies across quantum coding literature, and includes notions such as finding the maximum-likelihood error or maximum-likelihood error equivalence class.
However, the most natural definition of decoding is simply a procedure that on input the code and noisy code state outputs the underlying logical state.
In the case of classical linear codes, for example, a code is a $n \times k$ binary generator matrix, and the task is to output the logical bitstring under noise.
Quantumly, the code is a description of a Clifford operator (e.g. as a circuit), but the task is also to output the logical state.
We first briefly outline the distinctions in noise model, as well as two distinct notions we will consider of ``outputting the logical state''.

Classically, every bit of the codeword has some probability $p$ independently of being flipped.
Quantumly, every qubit of the code state has a probability $p$ of having a random Pauli in $\{\vec{X}, \vec Y, \vec Z\}$ being applied to it.
If a qubit does get corrupted by a Pauli, the Pauli is chosen uniformly at random.
Such an error model is known as depolarizing noise and is the natural generalization of the classical noise model.
Following~\cite{poremba2025learningstabilizersnoiseproblem}, we define Search $\lsn(k, n, p)$ to be the distributional task of recovering a bitstring $\vec x \in \Z_2^k$, which is chosen uniformly at random, given $\vec E \vec C \ket{0^{n-k}, \vec x}$, where $\vec E$ is the error and $\vec C$ is a Clifford operator chosen uniformly at random.
Such a formulation allows the problem to have classical output without severely affecting the inherent quantumness of the problem.
To fully capture the quantumness of the logical state, we also define Search $\slsn(k, n, p)$ where the logical state is instead a Haar-random $k$-qubit state.
In average-case analysis, which is intimately connected to cryptography, it is of interest to determine whether any efficient algorithm can solve a problem with probability non-negligibly higher than random guessing.
For this reason, we say that a solver succeeds if it outputs $\vec x$ with probability at least $\frac{1}{2^k} + \frac{1}{\poly(n)}$.
If $k = \omega(\log n)$, then the addition of $\frac{1}{2^k}$ is irrelevant.
However, since we will be interested in the regime even when $k = 1$, this extra term makes a significant difference.
The decision version is concerned with distinguishing two possibilities: given a random Clifford $\vec C$ and a state $\rho \in (\C^{2 \times 2})^{\otimes n}$, either $\rho$ is the maximally mixed state or $\rho$ is $\vec E \vec C \ket{0^{n-k}, \vec x}$ with a random logical string $\vec x$ and error $\vec E$ drawn from a depolarizing distribution.
There are two variants of the above.
In the first, we define $\lsn^m(k, n, p)$ by providing $m$ noisy codewords.
Each ``sample'' has a fresh error $\vec E$ and fresh code $\vec C$, but has the \emph{same} secret logical state $\vec x \in \Z_2^k$.
Such a problem is not obviously natural for the purposes of practical quantum error correction, but is nevertheless reasonable as a theoretical problem in the context of complexity and cryptography.
More importantly, the multi-sample version turns out to be a crucial intermediate problem in the reduction from $\lpn$ to $\lsn$.
We also consider a multi-sample variant with random states, known as $\slsn^m$.

\begin{figure}[t!]
    \centering
    \includegraphics[width=0.6\linewidth]{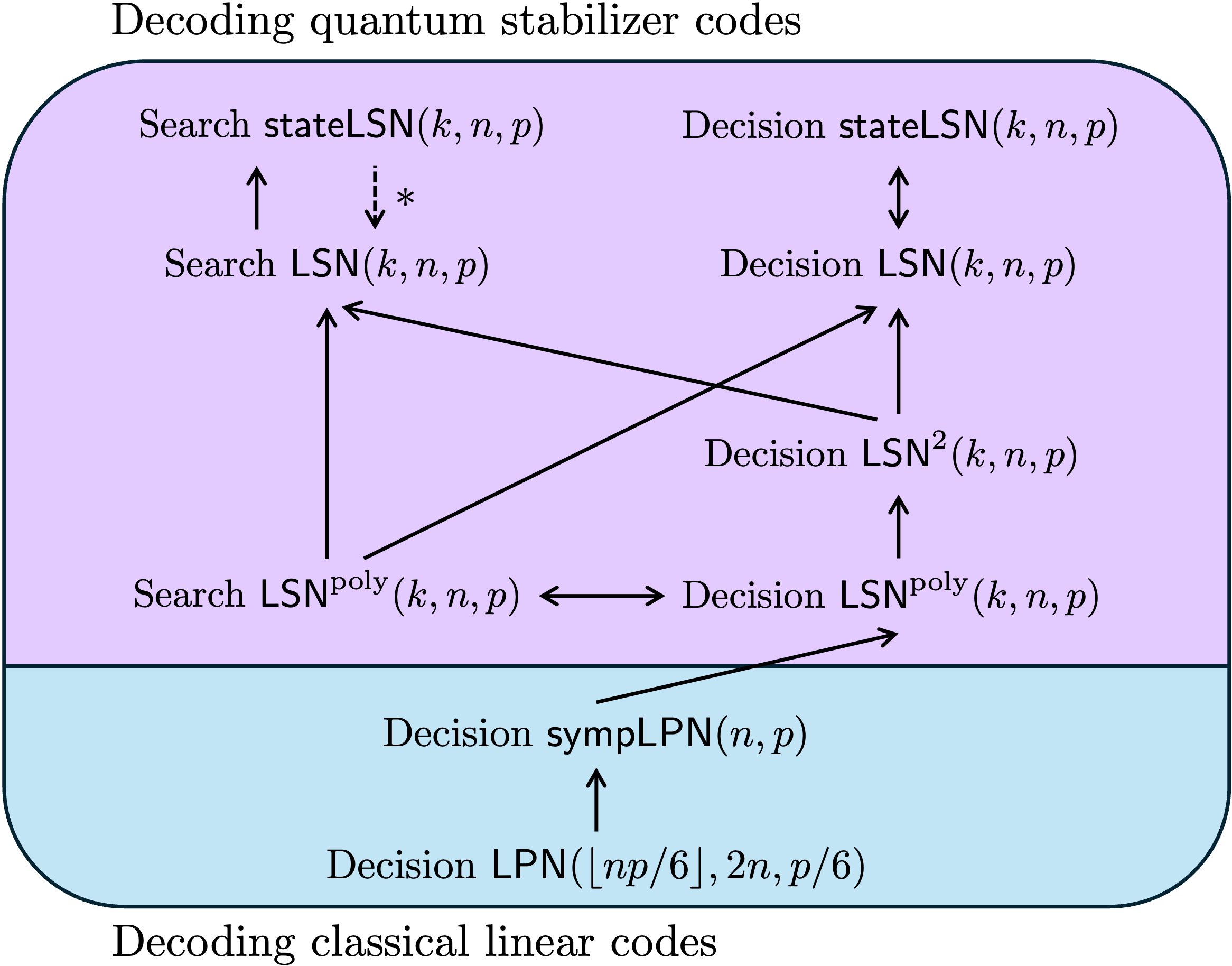}
    \caption{Web of reductions between random classical (blue) and quantum (purple) decoding problems. The parameters (logical (qu)bits, physical (qu)bits, and error probability) shown are exact. 
    Arrows imply a reduction in time $\poly(n)$. The reductions to $\lsn$ hold for any choice of $k \geq 1$, including a single logical qubit $k = 1$. Here $\lpn$ is the classical decoding problem, $\symplpn$ is a symplectically structured version of $\lpn$, $\lsn^{m}$ is quantum decoding with $m$ samples of noisy codewords with the same fixed random logical computational basis state, $\lsn = \lsn^1$, and $\slsn$ uses Haar-random logical states instead of a random logical computational basis state.}
    \label{fig:LSN_main_figure}
\end{figure}

We prove that some of these variants are equivalent by giving efficient reductions from one to the other.
For example, Decision $\lsn(k, n, p)$ and Decision $\slsn(k, n, p)$ are exactly equivalent because the underlying distributions of the logical state are both maximally mixed.
By exploiting the fact that in certain cases the Haar-random state can be replaced by a $t$-design, we can also show that Search $\lsn$ reduces to Search $\slsn$.
The converse holds if the success probability of the $\lsn$ solver is non-negligibly higher than $\frac12 + \frac1{2^{k+1}}$ (rather than $\frac{1}{2^k}$), or when $k = O(\log n)$.
The former is proven by using syndrome decoding as an intermediate problem and a technique that is a form of randomized syndrome matching, while the latter is proven by performing a measurement via randomized syndrome matching and analytical techniques with the measurement distribution of Haar-random states.
We remark that many of these equivalences in the classical analog are immediate, while quantumly our proofs are much more non-trivial.

\begin{theorem}[informal]
    $\lsn$ and $\slsn$ are equivalent decision problems. Search $\lsn$ reduces to Search $\slsn$, and Search $\slsn$ reduces to Search $\lsn$ if the $\lsn$ solver succeeds with probability substantially larger than $\frac12 + \frac1{2^{k+1}}$ (as opposed to $\frac{1}{2^k}$), or if $k = O(\log n)$.
\end{theorem}

The complete diagram of average-case reductions which we give in this paper is shown in Fig.~\ref{fig:LSN_main_figure}.
So far, we have filled in the equivalences between $\slsn$ and $\lsn$.
The reduction from $\lsn^m$ to $\lsn^{m'}$ with $m > m'$ is trivial by way of ignoring $m - m'$ samples.
The remaining reductions are between decision and search, which we discuss later, and path of reductions from $\lpn$ to $\lsn$.

Starting from $\lpn(k = \lfloor np/6\rfloor, 2n, p/6)$ , we here briefly sketch the steps to reduce to $\lsn(1, n, p)$ which we discuss in greater detail in the remainder of the overview.
We first reduce $\lpn$ to an intermediate problem---a version of $\lpn$ we define whose generator matrix columns satisfy a certain symplectic orthogonality condition and whose error model is a symplectic version of depolarizing noise---which we call $\symplpn$.
This reduction relies on two steps. 
First, we extend a smaller $\lpn$ instance to a slightly larger $\lpn$ instance, carefully using the extra degrees of freedom to enforce the symplectic orthogonality condition and maintain independent noise.
We then show how to transform the independent bit flip noise model to depolarizing noise.
In the process, the error parameter $p$ increases by a small constant factor.
Despite this reduction, the symplectic-type $\lpn$ instance still looks qualitatively different from a $\lsn$ instance, where a Clifford circuit and a quantum state are the inputs rather than a matrix and a noisy matrix-vector product.
To bridge this gap, we develop a new representation of $\lsn$ which is entirely classical.
This representation (a) is exactly equivalent to the standard definition of $\lsn$ discussed above and (b) looks manifestly like a modified version of a $\lpn$-type instance of the form $(\vec A, \vec A \vec x + \vec e)$.
This representation renders a reduction from $\symplpn$ to $\lsn$ much more plausible.
In fact, our reduction from the symplectic-type $\lpn$ to the classical representation of $\lsn$ has no dependence on $k$ in $\lsn$, whereas in the $\lpn$ instance from which we reduce the value of $k$ can be linear in $n$.
We show that this independence from $k$ is a direct consequence of quantum degeneracy, which manifests in the classical representation as an extra ``junk'' register of $n$ bits which obfuscates the real value of the logical bitstring even if there is only only one logical qubit.
More precisely, the existence of this junk register arises from the fact that neither stabilizers nor logical $\vec Z$ operators affect a logical computational basis state.
Therefore, we obtain a reduction from constant-rate $\lpn$ to $\lsn$ with any number of logical qubits $k$, with error probabilities within a constant factor of each other.
This answers Question~\ref{question:LPN_to_LSN?}.
Note that we have omitted a key detail above, namely that the reduction from $\symplpn$ to $\lsn$ only holds in the decision variant, whereas we wish to show that Search $\lsn$ is hard.
Almost always, a decision-to-search reduction is trivial because it is implied from efficient certification of solutions.
As we discuss below, $\lsn$ solutions are unlikely to be efficiently verifiable, so a decision-to-search reduction is also unlikely.
We circumvent this by devising a decision-to-search reduction for $\lsn^2$, and proving a stronger reduction from $\symplpn$ to $\lsn^{\poly(n)}$, which we outline below.

\begin{theorem}[informal] \label{thm-informal:lpn_to_lsn}
    Fix any $k \geq 1$ and $p \in (0, 1)$ which is not necessarily a constant. 
    There exists a reduction from $\lpn(\lfloor np/6 \rfloor, 2n, p/6)$ to $\lsn(k, n, p)$. 
    In particular, for a constant error probability $p$, constant rate $\lpn$ reduces to $\lsn$ with a single logical qubit.
\end{theorem}

One application of this reduction is an explanation for the absence of \emph{any} known sub-exponential algorithm for decoding quantum stabilizer codes in \emph{any} non-trivial parameter regime, despite decades of study in quantum error correction.
On the classical side, in spite of much effort, no known sub-exponential algorithm succeeds on $\lpn$ with constant rate.
The above reduction shows that the search for these two types of algorithms are connected, and gives evidence for a negative answer to Question~\ref{question:attack_algorithms}.

\begin{corollary}[informal] \label{thm-informal:attack_algo}
    The existence of \emph{any} sub-exponential  time (in $n$) algorithm which solves $\lsn$ at \emph{any} rate $\frac{k}n$ implies a substantial breakthrough in classical cryptography in the form of a sub-exponential time algorithm for $\lpn$ at constant rate.\footnote{Note that the runtime for $\lpn$ decoders is typically stated with respect to $k$, the length of the secret. In the case of $\lsn$, it is more natural to think of the block length $n$ as the relevant parameter. Since we consider constant rates, the asymptotics are the same either way. We discuss these differences in~\Cref{sec:avg-case}.}
\end{corollary}
Our work thus presents a new and surprising strong barrier for the existence of sub-exponential decoders of quantum codes.

\subsubsection{Classical Representations of $\lsn$}
\begin{figure}[t!]
    \centering
    \includegraphics[width=0.5\linewidth]{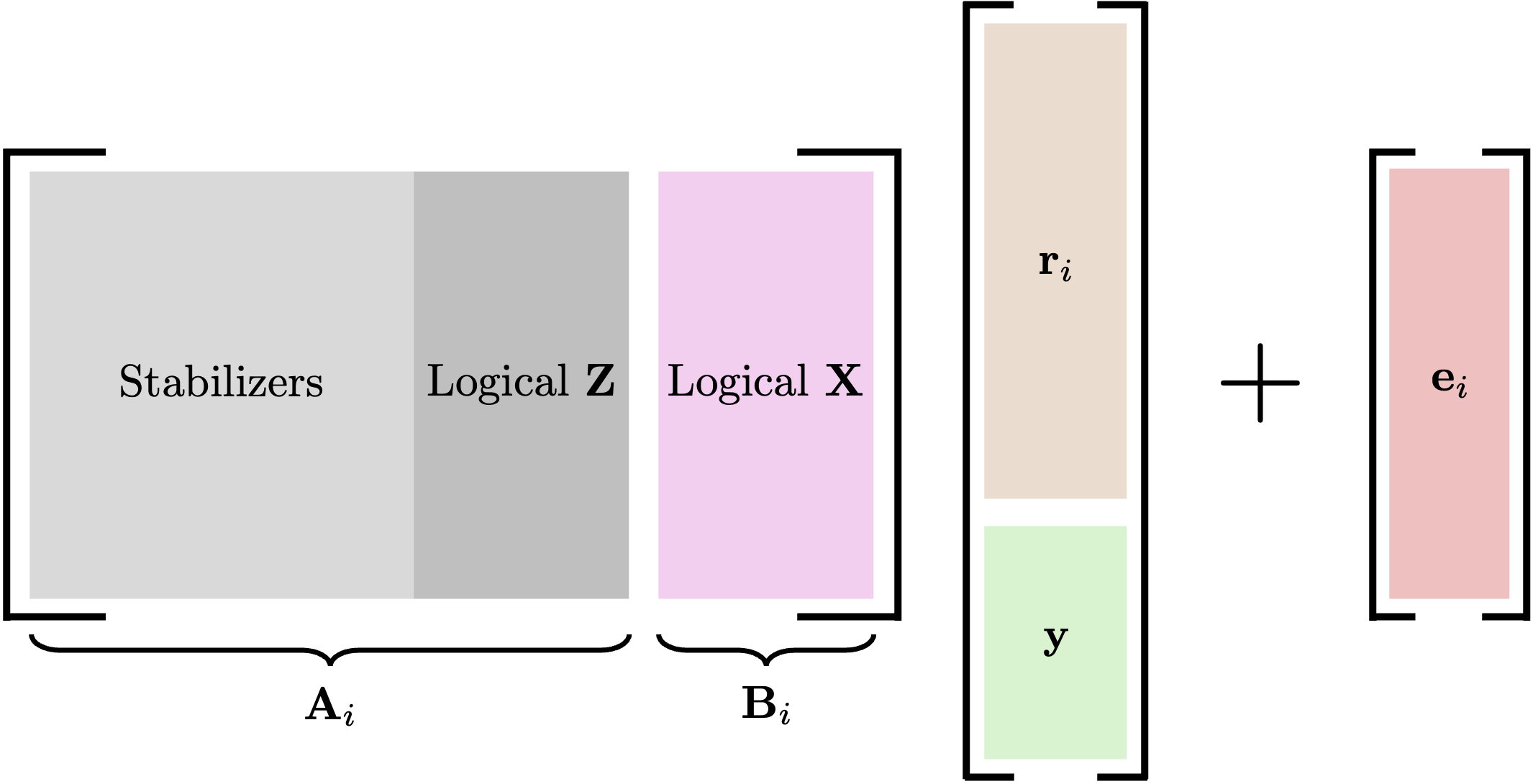}
    \caption{
    A schematic representation for the classical representation of $\lsn^m(k, n, p)$. 
    Everything is written in symplectic representation.
    Each of the $m$ samples given are of the form $(\vec{A}_i, \vec{B}_i, \vec{A}_i \vec{r}_i + \vec{B}_i \vec{y} + \vec{e}_i)$, with the goal of finding $\vec y$.
    The secret bitstring $\vec{y} \sim \Z_2^k$ is a random string which remains \emph{fixed} across samples.
    Here $\vec{A}_i$ is a $2n \times n$ matrix, where $n-k$ columns represent stabilizers and the remaining $k$ represent logical $\vec Z$ operators.
    $\vec{B}_i$ is a $2n \times k$ matrix of the logical $\vec X$ operators. 
    $\vec{r}_i$ is a randomly generated junk vector and $\vec{e}_i$ is an error drawn from $n$ independent single-qubit depolarizing channels each with probability $p$.
    The block multiplication structure emphasizes that the stabilizers and logical $\vec Z$'s act trivially on the secret string $\vec y$; only the logical $\vec X$'s affect $\vec y$.
    For illustrative clarity, this schematic ignores the fact that in the true classical representation $\vec{A}_i$ is randomized in such a way that there is no distinction between stabilizers and logical $\vec Z$ operators. }
    \label{fig:classical_representation}
\end{figure}

One of the main tools which we use to bridge the apparent qualitative gap between the definitions of $\lsn$ and $\lpn$ is a new \emph{classical} representation of $\lsn$, shown in Fig.~\ref{fig:classical_representation}.
The primary motivation for the development of such a representation is to narrow the conceptual gap between the formulations of $\lpn$ and $\lsn$, towards the effort of a reduction from the former to the latter.
As given, $\lpn$ has input $(\vec A, \vec{Ax} + \vec e)$ where $\vec A$ is a generator matrix.
On the other hand, $\lsn$ has input $(\vec C, \vec{EC } \ket{0^{n-k}, \vec{x}})$ for a Clifford circuit $\vec C$.
If $\lsn$ could be represented in an equivalent form that (a) removes the quantum input, becoming fully classical, and (b) structures the given noisy code state in the form of a (non-exponential size) matrix multiplication added to some small error, the reductive path becomes significantly more plausible.

The classical representation task is given as follows. Let $\vec A_i \in \Z_2^{2n \times n}$ be a random matrix whose columns are all orthogonal in the symplectic inner product.
Let $\vec B_i \in \Z_2^{2n \times k}$ be a random matrix with the same property, such that all columns of $\vec A_i$ and $\vec B_i$ are collectively linearly independent.
Let $\vec r_i \in \Z_2^{n}$ be a random string, and $\vec e_i \in \Z_2^{2n}$ be the symplectic representation of a Pauli drawn from the depolarizing distribution.
Fix a random string $\vec{y} \in \Z_2^k$.
The search task is, given $m$ samples of $\Big(\begin{bmatrix}
      \vec A_i \, | \, \vec B_i  
    \end{bmatrix}, 
    \begin{bmatrix}
      \vec A_i \, | \, \vec B_i  
    \end{bmatrix} \cdot \begin{bmatrix}
        \vec r_i \\
        \vec y
    \end{bmatrix} + \vec e_i \Big) $, to find $\vec y$ (in the decision version, to distinguish the given states from random).
We show below that $\vec A_i$ corresponds to the stabilizer and logical $\vec Z$ operators of the code, while $\vec B_i$ represents the logical $\vec X$ operators.
Intuitively, in $\lsn$ logical $\vec Z$ operators and stabilizers do nothing to the logical computational basis state $\vec x$, and thus can be interpreted as ``gauge'' degrees of freedom in the problem.
With a randomly chosen code, these operators obfuscate what is actually relevant---the logical $\vec X$'s---by adding random ``junk'' to the decoding problem, which we encode into the junk vector $\vec r_i$.
It is in fact this $n$ bits of junk---a consequence of a stronger form of quantum degeneracy which includes both stabilizers and logical $\vec Z$'s---which increases the hardness of $\lsn$ to the level of $\lpn$ regardless of how many true logical qubits we give to the stabilizer code.

Two remarks about this classical formulation are in order. 
First, we observe that this representation does indeed meet the criterion of being a simple (small) matrix multiplication with errors problem, as illustrated in Fig.~\ref{fig:classical_representation}.
Second, $\lpn$ is framed in terms of a generator matrix, which produces a codeword that is corrupted with noise.
Stabilizer decoding, on the other hand, is often more naturally phrased as decoding from a syndrome given by the stabilizer tableau.
Yet, in this classical representation, we have constructed a decoding problem which combined both a ``generator'' component (logical Paulis) and the parity check component (stabilizers) into one larger generator matrix.
This construction is unusual in this sense, consisting of embedding a parity check matrix into a larger generator matrix. 

We now sketch a derivation of the classical representation, starting from $\lsn$ itself.
Specifically, we outline a reduction from $\lsn^m$ to a $m$-copy classical representation problem.
The key insight is that $\lsn$ can be framed entirely as correctly deducing a certain logical operator.
Specifically, given a code state corrupted by some random Pauli error $\vec E$, we may measure the syndrome $\vec v$ to obtain a syndrome decoding problem.
If we do not care about the weight of the solution, then it is easy to find \emph{some} Pauli which produces the same syndrome $\vec v$, e.g. by Gaussian elimination.
In fact, we may sample randomly from this space of solutions.
Such solutions, however, will differ from the true error by some combination of stabilizers and logical operators.
A randomly chosen Pauli $\vec P$ with syndrome $\vec v$ may be interpreted as having three components: the error itself $\vec E$, a uniformly random element of the stabilizer subgroup, and a uniformly random logical operator which we can further break down to a uniformly random logical $\vec X$ and a uniformly random logical $\vec Z$.
If our noisy state is of the form $\vec{EC} \ket{0^{n-k}, \vec x}$, then $\vec{C^\dagger PEC} \ket{0^{n-k}, \vec x}$ equals $\ket{0^{n-k}, \vec z}$, where $\vec{z}$ differs from the true logical state $\vec{x}$ by precisely the random $\vec X$ logical encoded within $\vec P$.
But finding this $\vec X$ logical part is precisely the task given by the classical representation---it is the task of finding $\vec y$, the logical $\vec X$ part of a symplectically-represented Pauli $\vec P$ which also includes some random contributions from stabilizers and logical $\vec Z$ operators as specified by $\vec r_i$, in the presence of an error $\vec e$.
Therefore, the classical representation task encodes solutions to $\lsn$.
We apply similar arguments to give a converse, proving that the classical representation and standard $\lsn$ are completely equivalent.
We henceforth use the classical representation to reason about the remainder of the reduction from $\lpn$.

\subsubsection{Symplectic $\lpn$: A Bridge Between $\lpn$ and $\lsn$}

\begin{figure}[t!]
    \centering
    \includegraphics[width=0.7\linewidth]{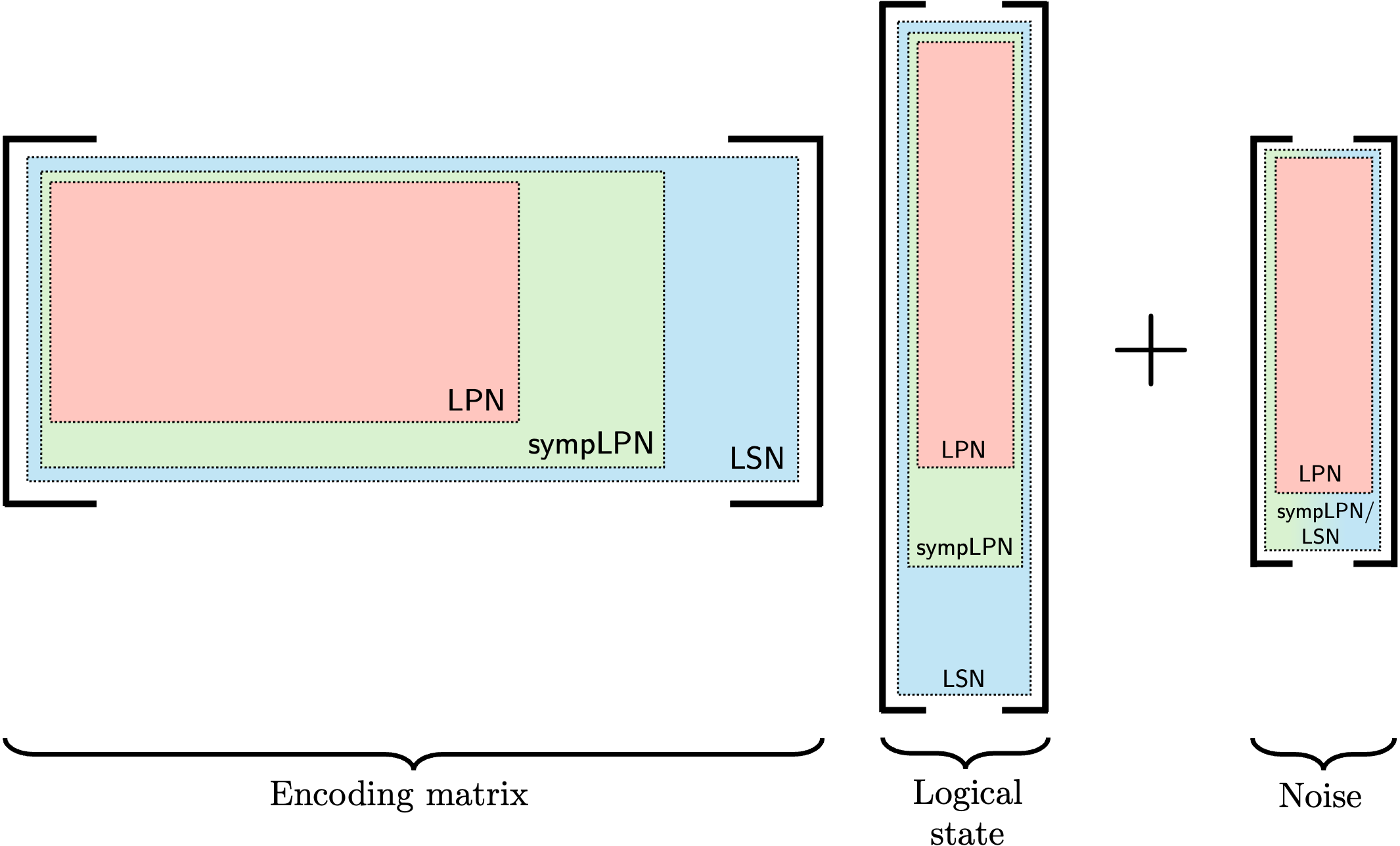}
    \caption{
    Embedding structure for the reduction from $\lpn$ to $\symplpn$ to $\lsn$ (in the classical representation).
    The matrix embeddings are drawn approximately to scale.
    For the code, a $\lpn$ generator matrix is extended by a technique of ``symplectic completion'' into a $\symplpn$ generator matrix.
    Further columns, corresponding to logical $\vec X$ operators, are added to complete the embedding into $\lsn$.
    For the logical state, a $\lpn$ secret is extended into a $\symplpn$ secret, which is all then embedded into the \emph{junk} portion of the $\lsn$ logical state.
    This junk embedding can be observed visually by comparing this figure to Fig.~\ref{fig:classical_representation}.
    (The actual $\lsn$ secret is chosen completely randomly.)
    For the noise, a Bernoulli error is extended to a larger Bernoulli error to embed $\lpn$ to $\symplpn$.
    Mapping from $\symplpn$ to $\lsn$, the dimensions remain the same, but the distribution is modified to be depolarizing instead of Bernoulli.
    }
    \label{fig:embedding_structure}
\end{figure}

The classical representation narrows the conceptual gap between quantum and classical decoding from the quantum side.
We further close this gap from the classical side by defining a quantum-inspired structured version of $\lpn$, which we call $\symplpn(n, p)$.
$\symplpn(n, p)$ differs from $\lpn$ primarily in the distribution of generator matrices $\vec A$.
Rather than letting $\vec A \sim \Z^{n \times k}_{2}$ be a uniformly random matrix, we sample $\vec A$ uniformly at random from the set of full-rank matrices in $\Z_2^{2n \times n}$ such that all columns are orthogonal under the symplectic inner product.
Since the symplectic inner product requires even dimension to be well-defined, we have replaced $n$ with $2n$ in the above.
Moreover, for connection with $\lsn$, we always fix $k = n/2$ in $\symplpn$.
Finally, in $\symplpn$ we use a depolarizing noise model (in symplectic representation) instead of Bernoulli noise.
The intuition behind such a construction is that the classical representation looks like a $\lpn$ instance, except that (a) the generator matrix is structured in a certain way that causes certain columns to be symplectically orthogonal, (b) the noise model is depolarizing instead of Bernoulli, and (c) only part of the vector being multiplied actually contains the secret.
In this sense $\symplpn$ serves as a hybrid problem between $\lpn$ and the classical formulation of $\lsn$ which includes (a) and (b) but not (c).
Our proof technique is then to reduce $\lpn$ to $\symplpn$, and then reduce $\symplpn$ to $\lsn$.

To the first reduction, we begin with a generator matrix $\vec A$ sampled uniformly at random.
The goal will be to pad $\vec A$ with some extra rows, using these new degrees of freedom to ensure that the columns become not only symplectically orthogonal, but a uniformly random symplectically orthogonal matrix.
Specifically, suppose we begin with a uniformly random matrix which is slightly smaller than $2n \times \ell$, namely $(2n - (1+\e)\ell) \times \ell$ for some small constant $\e > 0$ such that $(1+\e)\ell$ is an integer.
We construct $\vec A' \in \Z_2^{(1+\e)\ell \times \ell}$ and stack $\vec A$ on top of it, creating a new matrix $\vec B \in \Z_2^{2n \times \ell}$ such that the distribution of $\vec B$ is negligibly close to a uniformly random full-rank symplectically orthogonal matrix.
The construction proceeds as follows: write \begin{align}
    \vec B = \begin{bmatrix}
        \vec N_1 \\ \vec M \\ \vec N_2 \\ \vec A'
    \end{bmatrix} ,
\end{align}
where $\vec N_1, \vec N_2$ are $(n - (1+\e)l) \times l$ and $\vec M, \vec A'$ are $(1+\e)l \times l$.
In this decomposition, we may re-express the condition of symplectic orthogonality as that of a certain matrix being \emph{symmetric}.
In particular consider the matrix, \begin{align}
\label{eq:overview-S_matrix}
    \vec S := \vec N_1^\intercal \vec N_2 + \vec M^\intercal \vec A' .
\end{align}
Then $S_{ij}$ ($S_{ji}$) is the inner product of the top half of the $i$th ($j$th) column with the bottom half of the $j$th ($i$th) row of $\vec B$.
Hence, the symplectic inner product between columns $i$ and $j$ of $\vec B$ is precisely given by $S_{ij} + S_{ji}$.
This is zero for all $i, j$ if and only if $\vec S$ is a symmetric matrix.
This new structure suggests a simple way to sample $\vec A'$ in the desired way.
First, sample a uniformly random symmetric matrix $\vec S'$.
Then, subtract off the known part of Eqn.~(\ref{eq:overview-S_matrix}), creating $\vec T := \vec S' - \vec N_1^\intercal \vec N_2$.
From there, pick a uniformly random $\vec A'$ such that $\vec M^\intercal \vec A' = \vec T$; this step is efficient since it amounts to randomly sampling solutions to a linear equation.
The only remaining uncertainty to check is the homogeneity of this approach.
That is, whether or not for every choice of $\vec A$ there are (up to negligible failure probability) the same number of symplectic completions $\vec A'$, which would ensure that the distribution of $\vec B$ truly is nearly that of a uniformly random symplectically orthogonal full-rank matrix.
We prove this property by direct counting arguments.

The above technique of symplectic completion addresses the code part of the reduction, but now we have enlarged the decoding problem by adding more rows.
We therefore must also supplement by adding some of our own noise so that the new noise distribution matches the desired one.
If we do not do this, then the transformed instance is easily distinguished from $\symplpn$ because there will never be noise in the padded rows.
Unfortunately, the only noise that we can add is uniformly random bit flips, because we can only write down the final result $z_j := \vec a'_j \cdot \vec y + e_j$ for each padded row $\vec a'_j$.
Since we do not know the secret $\vec y$, we have no control over the distribution of $z_j$ for any choice of the distribution of the new error bit $e_j$ unless it is either 0 or $\frac{1}{2}$.
To circumvent this issue, we first generate a random number of indices in the padded region in which to apply noise---so that the total error weight is approximately correct---and then let the resultant bits in those regions of the noisy codeword be uniformly random.
To erase the positional structure of the padding being at the end, we then apply a random permutation of our bits.
The same permutation is applied to both the first and last $n$ bits, to ensure that symplectic orthogonality is preserved.
Careful analysis of a formal combination of these two techniques---randomized symplectic completion and noise injection---give a random full-rank symplectically orthogonal code with independent bit-flip noise.
To complete the reduction, we show how to add extra noise from an independent distribution such that the new aggregate noise distribution is depolarizing.
The parameters $\ell, \e$ can be chosen such that the total change in error probability is a factor of 6, as shown in Fig.~\ref{fig:LSN_main_figure}.

Finally, we outline the remaining piece of the reduction---from $\symplpn$ to $\lsn$.
Starting from $\symplpn$, we have an instance $( \widetilde{\vec{A}} \in \Z_2^{2n \times n}, \; \vec y \in \Z_2^{2n})$, and we wish to embed it into the classical representation of $\lsn$, which is of the form $(\vec A, \;\vec B, \;\vec z=\vec A \vec r + \vec B \vec x + \vec e)$.
By examining the dimensions alone, it becomes apparent that the way to encode $\symplpn$ into $\lsn$ is by embedding the $\symplpn$ matrix into the \emph{junk} part of the $\lsn$ encoding matrix.
Such an embedding reveals an important problem, however: we must choose the matrix $\vec B$ and the $\lsn$ secret $\vec x$ randomly ourselves, but then a Search $\lsn$ solver would do no good---it beings us no value for a solver to output the $\vec x$ we chose ourselves.
To circumvent this issue, we formally reduce from the \emph{decision} versions of $\symplpn$ and $\lsn$.
Now, because the only goal is to decide if everything is random or if the given string is actually a noisy codeword, we can rely on the observation that the transformed $\lsn$ problem is random (actually a noisy codeword) if and only if the $\symplpn$ problem is random (actually a noisy codeword).
Therefore, a distinguisher for $\lsn$ implies a distinguisher for $\lpn$, which completes the reduction between the decision problems.
Figure~\ref{fig:embedding_structure} illustrates the entire sequence of embeddings from $\lpn$ to $\lsn$.

For almost all classical problems, the issue of being able only to reduce directly to a decision problem when we in reality desire a reduction to a search problem is entirely trivial, because decision-to-search reductions always exist due to the ability to efficiently verify proposed solutions.
In the setting of $\lsn$, however, it is not at all clear how one might verify a proposed solution.
Specifically, given a noisy code state $\vec{EC} \ket{0^{n-k}, \vec x}$, and a proposed solution $\ket{\vec x'}$, is there a (quantum) algorithm to check with high probability whether $\vec x = \vec x'$?
Given nothing else but $\vec x'$, this task is equivalent to checking if $\vec C \ket{0^{n-k}, \vec x'}$ and $\vec{EC} \ket{0^{n-k}, \vec x}$ differ by a low-weight Pauli, which is essentially a short vector problem.
We therefore are unable to find a direct decision-to-search reduction for $\lsn$.

Fortunately, we are able to prove a decision-to-search reduction in the case of multiple samples.
More precisely, we can reduce Decision $\lsn^{2m}(k, n, p)$ to Search $\lsn^m(k, n, p)$.
In this case, the key proof technique is to use randomness to our advantage.
Specifically, split up the samples into two halves.
For one set of the query the Search $\lsn^m$ solver exactly as given.
For the other, generate a random bitstring $\vec y \in \Z_2^k$ and apply a logical operator $\overline{\vec X}_{\vec y}$ to each sample.
If the given Decision $\lsn^{2m}$ instance has true noisy code states with logical state $\vec x \in \Z_2^k$, then there is a non-negligible probability that the solvers will output $\vec x$ and $\vec x + \vec y$, at which point we can subtract to verify that we obtained $\vec y$.
If the instance has completely random states, then applying $\overline{\vec X}_{\vec y}$ does nothing to the distribution, the Search $\lsn^m$ solver will have success probabilities negligibly close to random guessing.
Hence, there is a non-negligible advantage, giving a reduction.

The fact that a decision-to-search reduction appears to require multiple samples---a rather unphysical formulation of quantum decoding---is not only interesting in its own right, but it complicates our reduction from $\symplpn$ to Decision $\lsn$, because we now must reduce to multi-sample Decision $\lsn$ instead of Decision $\lsn$ itself.
The multi-sample version is potentially easier than $\lsn$, and therefore we must prove a new reduction directly to Decision $\lsn^2$.
This reduction is the true reduction which we give in the main text, but suffers more technical complications because we must carefully embed a single $\symplpn$ sample into a multi-sample version of $\lsn$.
In spite of this difficulty, we are able to give a reduction by devising a randomized version of a hybrid argument.
Effectively, we hide the $\symplpn$ input into a randomly chosen sample for the $\lpn^m$ solver, and interpolate by choosing samples before the chosen one to be true code states and samples after to be completely random strings.
By analyzing this scheme using the averaging principle, we prove that a reduction to multi-sample Decision $\lsn^m$ holds for any $m = \poly(n)$, thereby completing the full reduction from Decision $\lpn$ to Search $\lsn$, as shown in Fig.~\ref{fig:LSN_main_figure}.

\subsubsection{Worst-case Decoding}
Before we outline our results in self-reducibility, we will define and discuss equivalences between the worst-case variants of stabilizer decoding.
This is because our self-reduction results require both worst-case and average-case notions of decoding.
The classical formulation of decoding linear codes is, in the worst case, known as the nearest codeword problem $\ncp(k, n, w)$~\cite{brakerski2019worst}.
In this setting a code is specified by a generator matrix $\vec A \in \Z_2^{n \times k}$, where $n \geq k$.
Noisy codewords are of the form $\vec{Ax} + \vec e$, where $\vec x \in \Z_2^k$ and $\vec e \in \Z_2^n$.
The decision version of $\ncp$ is to determine, given $(\vec A, \vec y)$ whether there is a weight $\leq w$ error $\vec e$ such that $\vec y = \vec{Ax} + \vec e$ for some $\vec x \in \Z_2^k$.
The search version is to find $\vec x$, promised that $\vec y$ is indeed of the above form.
Quantumly, we may define an analogous decision task as follows: given a Clifford operator $\vec C \in \CC_n$ (e.g. specified classically as a circuit) and a state $\ket{\phi} \in (\C^2)^{\otimes n}$, decide if there exists a weight $\leq w$ Pauli $\vec E \in \CP_n$ and $\ket{\psi} \in (\C^2)^{\otimes k}$ such that $\ket{\phi} = \vec E \vec C \ket{0^{n-k}, \psi}$.
We denote the quantum nearest code state problem $\qncp(k, n, w)$.
The search formulation of $\qncp(k, n, w)$ suffers from a fragility to changes in the definition, which we now sketch.

\begin{figure}[t!]
    \centering
    \includegraphics[width=0.5\linewidth]{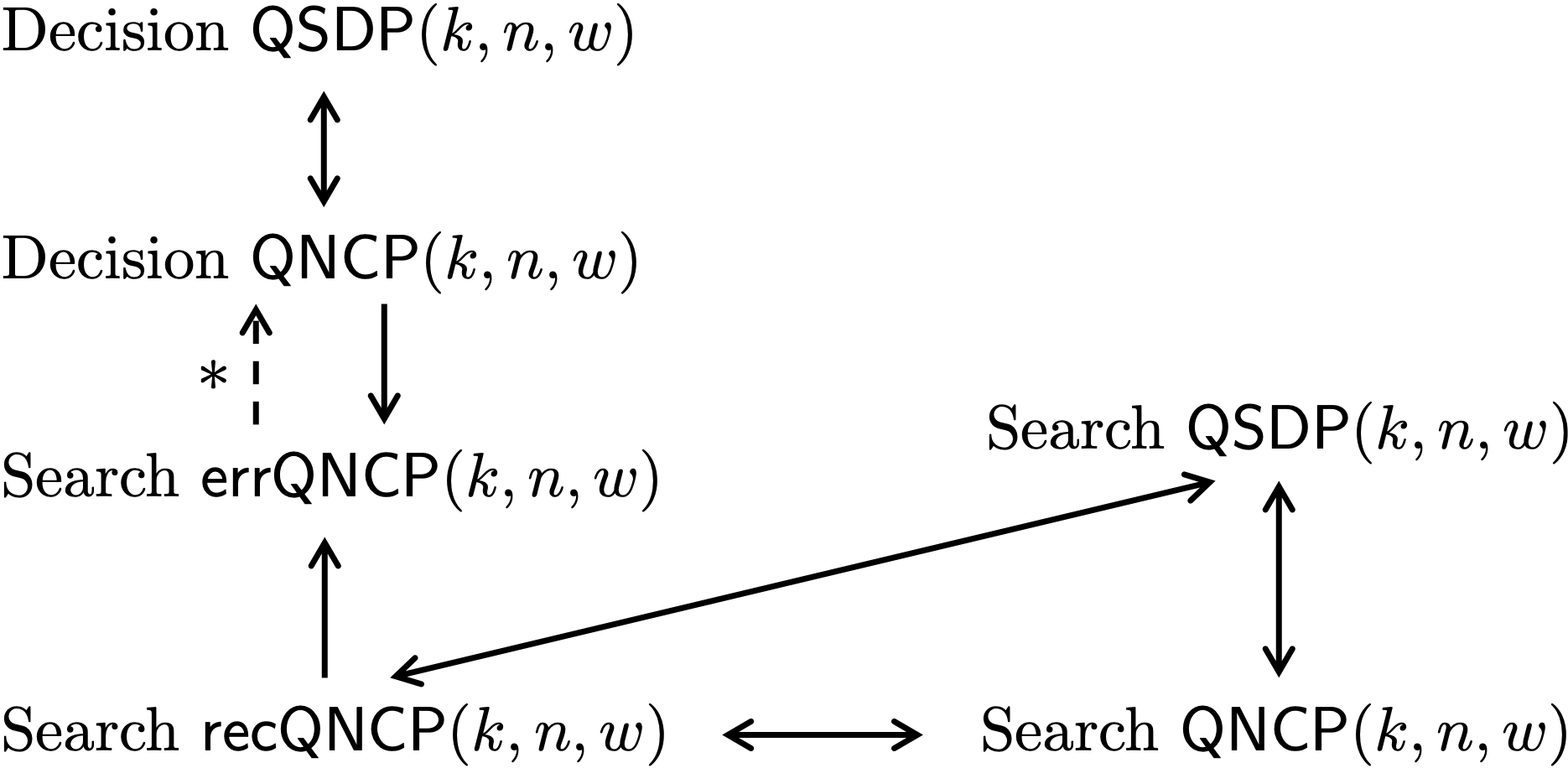}
    \caption{Web of reductions between worst-case quantum decoding problems. The parameters are the logical qubits $k$, physical qubits $n$, and maximum error weight $w$. $\qncp$ is the task of recovering the underlying logical state; $\recqncp$ synthesizing a correct recovery operator; $\errqncp$ outputting a low-weight recovery operator; $\qsdp$ the problem of finding a low-weight error causing a given syndrome. The dotted arrow with an asterisk signifies that the search-to-decision reduction for $\errqncp$ actually requires a Decision $\qncp$ solver for every $w' \leq w$, not just a solver with $w$ itself.}
    \label{fig:QNCP_map}
\end{figure}

The following three formulations of Search $\ncp$ are exactly equivalent in the sense that if an efficient algorithm succeeds in any one task with probability $q$, there is an efficient algorithm on any other task which succeeds with probability $q$: find $\vec x$, synthesize a vector $\vec v$ such that $\vec y - \vec v = \vec{Ax}$, or find $\vec e$.
This is because if we know $\vec x$, we can compute $\vec e = \vec y - \vec{Ax}$. 
If we know $\vec{e}$, we can compute $\vec{Ax} = \vec y - \vec e$ and use Gaussian elimination to compute $\vec x$.
Finally, finding the recovery vector $\vec v$ amounts to finding $\vec e$ itself.
In the setting of quantum codes, however, these three notions are \emph{not} immediately equivalent.
The analogous three formulations for Search $\qncp$ are as follows:
find $\ket{\psi}$ (Search $\qncp$); find any $\vec{E'}$ such that $\vec{E' E}$ is a stabilizer (Search $\recqncp$); and find a low-weight $\vec E'$ such that $\vec{E' E}$ is a stabilizer of the code, e.g. find $\vec{E}$ itself, if all stabilizers are high weight (Search $\errqncp$).
We say that $\vec E'$ is a \emph{recovery operator} for $\vec E$ if $\vec {E' E}$ is a stabilizer.
Certainly $\qncp$ reduces to $\recqncp$, which reduces to $\errqncp$, but the reverse directions are significantly less clear.
In particular, suppose we have a solver for $\recqncp$, which generates with probability $2/3$ some Pauli $\vec E'$ such that $\vec{E' E}$ is a stabilizer.
Any algorithm which upon calling this solver can solve $\errqncp$ with non-negligible probability is by \emph{definition} a solver for a certain short vector problem (SVP), similar to the (binary) SVP problem introduced in~\cite{applebaum_et_al:LIPIcs.ITCS.2017.7}.
Specifically, the SVP is finding a low-weight element of the coset space $\vec{E} \mathcal{S}$, where $\mathcal{S}$ is the stabilizer subgroup.
By using the symplectic representation, wherein $n$-qubit Paulis are represented by a bitstring of length $2n$ (the first $n$ bits indicating $\vec X$'s and last $n$ bits indicating $\vec Z$'s), such a coset is the coset space of a linear subspace over $\Z_2$ and is a bona fide short vector problem.
We believe therefore that it is unlikely that a reduction exists from $\errqncp$ to $\recqncp$, given that it implies an efficient algorithm for a variant of SVP.

It is similarly not immediately clear as to whether $\recqncp$ reduces to $\qncp$.
Given a solver to $\qncp$ which with probability $2/3$ yields the correct state $\ket{\psi} \in (\C^2)^{\otimes k}$, how could we extract the correct recovery operator? 
There is no direct analog of ``subtracting'' out the logical state in the classical sense, wherein $(\vec{Ax} + \vec{e}) - \vec{Ax} = \vec{e}$.
Nonetheless, it turns out that $\qncp$ and $\recqncp$ are exactly equivalent, but to prove it we must pass through the intermediate problem of \emph{syndrome decoding}.
This problem, which we call $\qsdp(k, n, w)$, is the problem of synthesizing a valid recovery operator given a syndrome.
It is entirely classical in the sense that both input and output are classical.
From the standard picture of stabilizer codes, the syndrome is obtained by measuring the stabilizers one by one, recording $+1$ if the stabilizer commutes with the error, and $-1$ otherwise.
In our case, it is convenient to define $\qsdp$ completely in terms of the symplectic representation. 
First, write each stabilizer of the code $\vec{S}_i$ as a symplectic vector $\vec{h}_i \in \Z_2^{2n}$.
We express the symplectic representation explicitly via the bijection $\Symp$.
The syndrome $\vec v \in \Z_2^{n-k}$ becomes a vector such that $v_i = 0$ if $\vec{S}_i$ commutes with the error and $1$ if not.
The decoding problem is to then output some $\vec e' \in \Z_2^{2n}$ such that $\Symp^{-1}(\vec e') \vec{E}$ is a stabilizer.
A corresponding decision problem is to decide if there exists a low-weight error $\vec e$ which produces the given syndrome $\vec v$.
As decision problems, $\qncp$ and $\qsdp$ are tightly equivalent in that an efficient algorithm succeeding with probability $q$ for one implies an efficient algorithm succeeding with probability $q$ for the other.
Since the syndrome does not depend on the underlying logical state, such an equivalence proves that the choice of the logical state itself for Decision $\qncp$ bears no relevance to its hardness.
The reduction from $\qncp$ to $\qsdp$ is immediate: measure the syndrome and call the syndrome decoder.
In the other direction, we exploit the fact that Gaussian elimination enables us to produce Pauli errors which when applied to a clean code state produce \emph{some} noisy code state with the same syndrome as the one given, $\vec v$.
This noisy code state will differ from a clean code state by a low-weight Pauli error if and only if the original code state differs from a clean code state by a low-weight Pauli error, completing the reduction.
Reducing from Search $\qncp$ to $\qsdp$ is a straightforward generalization of the decision reduction. 

The reduction from $\qsdp$ to Search $\qncp$ is slightly more subtle.
Again we can use Gaussian elimination to produce a Pauli error $\vec E'$ that gives the same syndrome as the given one.
However, we must explicitly transform this error to that of a valid recovery operator.
In general, $\vec E'$ is a product of a valid recovery operator, some logical $\vec X$, and some logical $\vec Z$.
Thus, it is enough to deduce what the logical parts are.
We may do this by calling the $\qncp$ solver twice.
First, we call it on $\vec{E' C} \ket{0^n}$, which gives the logical $\vec X$ part.
Then, it call it on $\vec{E' C} \ket{0^{n-k}, +^k}$, which gives the logical $\vec Z$ part.
Finally, since Search $\qsdp$ and Search $\recqncp$ are both tasks for synthesizing recovery operators, we can show by techniques similar to the proof of equivalence between Decision $\qncp$ and Decision $\qsdp$ that the two tasks are equivalent. 
Hence, we have proven the following piece of Fig.~\ref{fig:QNCP_map}.

\begin{theorem}[informal]
    Decision $\qncp$ and Decision $\qsdp$ are equivalent.
    Search $\qncp$, Search $\qsdp$, and Search $\recqncp$ are all equivalent in the sense that an efficient algorithm which solves any one problem with high probability exists if and only if an efficient algorithm which solves any other problem with high probability exists.
    Search $\recqncp$ reduces to Search $\errqncp$.
\end{theorem}

These equivalences are proven formally in the setting where a solution to $\qsdp$, $\qncp$, $\recqncp$, or $\errqncp$ is an algorithm that is correct with probability at least $\frac{2}{3}$.
For any of these problems, that correctness probability can be amplified to $1-\negl(n)$. 
For $\errqncp$, this may be done by repeating the algorithm, and verifying if the output is indeed a low-weight error and has the correct syndrome.
For $\qsdp$, we can run the solver many times and group outputs together
if they differ by a stabilizer. This is efficiently computable, since it is easy to test membership in a
stabilizer subgroup. To amplify, we choose the equivalence class corresponding to the group with the
most members. This majority vote boosts the success probability close to 1. 
Using the equivalence of $\qsdp$ with $\qncp$ and $\recqncp$, the success probabilities of these two problems can likewise be amplified.

\subsubsection{Reductions Between Search and Decision}

We next outline the reductions we make between search and decision variants of the stabilizer decoding problem, both in the worst and average case.
For search-to-decision reductions, all of our approaches involve sequentially learning one bit of information at a time by querying a Decision solver with a carefully constructed input.
In the worst case, we use an approach inspired by the technique of a classical decoding search-to-decision reduction.
Specifically, we reduce Search $\errqncp(k, n, w)$ to Decision $\qncp(k, n, w')$, assuming that we have a solver for every $w' \leq w$.
This assumption is quite strong---it enables a reduction from the hardest version of the search problem---but it is unclear if any reduction can be made, even from the weaker variant Search $\qncp$, with a weaker assumption.
The reduction proceeds by learning one bit of the Pauli error itself at a time.
In particular, we first query the Decision solver for the existence of a solution with weight $\leq w$.
If one exists, then we apply $\vec X$ to the first qubit and query with weight $\leq w-1$.
We do the same with $\vec Y$ and $\vec Z$ until we get a \texttt{YES} response.
If we do get a \texttt{YES} on Pauli $\vec P$, we know that a low-weight error has $\vec P$ on it.
If we never get a \texttt{YES}, we know that the error we seek is not supported on the first qubit.
We can then proceed in sequence to the next qubit, and so on.
In the average case, can directly reduce from Search $\lsn^{m_1}(k, n, p)$ to Decision $\lsn^{m_2}(k, n, p)$, where $m_1 \geq m_2$ but both are $\poly(n)$, so long as $p \gg \frac{\log n}{n}$.
In this case, the key insight is that the randomness of the encoding Clifford operator enables us to apply extra Clifford operations of our choice without affecting the distribution of the Clifford.
Specifically, we can select a random Pauli and apply a controlled version of it where the control is the first logical qubit.
One can show that if the Pauli is applied, then the problem has the distribution of a true noisy code state, while if it is not applied, then the problem has the distribution of a random state.
Therefore, the Decision solver which can distinguish these two situations will determine the value of the first logical qubit.
This procedure can be repeated sequentially to extract the entire logical state, thereby solving the search problem.
We have in this sketch omitted technical complications arising from the fact that all qubits needs to be re-randomized, but the qubit which serves as the control cannot be randomized by the controlled Pauli.
This issue can be addressed by using the Pauli error which comes as part of the $\lsn$ problem to scramble this final qubit.

As for decision-to-search reductions, we sketched in the previous section the procedure in the average case.
However, the proof strategy relied critically on the fact that the decision problem was to distinguish between a true $\lsn$ noisy code state and a completely random state.
In the worst case, the decision problem is formulated as determining whether or not a solution exists at all, which appears substantially harder.
In fact, in light of the search-to-decision reduction from Search $\errqncp$ to Decision $\qncp$, a decision-to-search reduction would imply at least some form of equivalence between Search $\errqncp$ and Search $\qncp$, which effectively would require solving the short vector problem discussed earlier.
We therefore believe it unlikely that such a reduction holds in the worst case.
For the same reasons, however, there is an immediate decision-to-search reduction to Search $\errqncp$.
We summarize our reductions below.

\begin{theorem}[informal]
    There is a reduction between Search $\lsn^{m_1}(k, n, p)$ to Decision $\lsn^{m_2}(k, n, p)$, where $m_1 \geq m_2$ are both $\poly(n)$ and $k = \omega(\log n / n)$.
    There is a reduction between Decision $\lsn^{2m}(k, n, p)$ to Search $\lsn^m(k, n, p)$.
    There is a reduction between Search $\errqncp(k, n, w)$ and the set of problems Decision $\qncp(k, n, w')$ for all $w' \leq w$, and a reduction from Decision $\qncp(k, n, w)$ to Search $\errqncp(k, n, w)$.
\end{theorem}

\subsubsection{Random Self-(ir)reducibility} 

One feature of $\lpn$ is its worst-case to average-case reduction, the random self-reduction shown by \cite{brakerski2019worst}.
They show that for a class of worst-case instances of $\lpn$, balanced linear codes, an instance $(\vec{A}, \vec{Ax}+\vec{e})$ can be re-randomized. 
In particular, by applying a well-chosen distribution of sparse matrices $\vec{R} \in \Z_2^{n \times n}$, they show that $(\vec{RA}, (\vec{RA})\vec{x}+\vec{Re})$ is close to a uniformly random instance, with $\vec{RA}$ close to a uniformly random code and $\vec{Re}$ close to independent Bernoulli noise. 
The parameters of the reduction are rather weak: the initial error weight is $\sim \log^2 n$, and the problem reduces to $\lpn$ with error probability $\frac{1}{2} - \frac{1}{\poly(n)}$. 
Nevertheless, it still establishes that a particularly hard version of $\lpn$ is harder than a non-trivial class of easier worst-case instances.

We show that there are significant barriers to a similar type of reduction for quantum stabilizer codes. 
In this case, the analogous reduction would find some distribution of Cliffords $\mathcal{R}_n$ such that for certain starting Cliffords $\vec{C}$, $(\vec{RC}, \vec{REC}\ket{0^{n-k}, \vec{x}})$ is close to a sample of $\lsn$, for $\vec{R} \sim \mathcal{R}_n$. 
In particular, the sample can be rewritten as $(\vec{R}\vec{C}, \vec{(RER^\dag)RC}\ket{0^{n-k}, \vec{x}})$. 
Hence, showing that $\vec{RC}$ is close to uniformly random while $\vec{RER^\dag}$ is still a low-weight error would be sufficient to complete the reduction. 
However, there is a subtlety here, namely that $\vec{RC}$ need not be a uniformly random Clifford to have a reduction.
Indeed, there could be two Cliffords $\vec{C}_1$ and $\vec{C}_2$ that both implement the same encoding on input states $\ket{0^{n-k}, \vec{x}}$.
Extra degrees of freedom which do not affect the distribution of the encoding need not be included in the design of $\vec R$, as we can always randomize these degrees of freedom without any change to the encoded state.

However, while different Cliffords can correspond to the same encoding of code states, if the stabilizer subspaces $\langle \vec{C}_1\vec{Z}_1\vec{C}_1^\dag, \;\dots, \;\vec{C}_1\vec{Z}_{n-k}\vec{C}_1^\dag\rangle$ and $\langle \vec{C}_2\vec{Z}_1\vec{C}_2^\dag, \;\dots, \;\vec{C}_2\vec{Z}_{n-k}\vec{C}_2^\dag\rangle$ are different, the encodings are always different. 
Hence, we may demonstrate barriers for reductions by showing that these stabilizer subspaces cannot be randomized. 
Our first barrier uses precisely this model. 
We show with entropy arguments that no distribution of ``sparse'' Cliffords, i.e. Clifford operators whose symplectic representation is sparse, has sufficient entropy to randomize over stabilizer subspaces. 
A sparse Clifford distribution is the most reasonable distribution for randomizing an encoding without amplifying errors, just as sparse matrices $\vec{R}$ were the most promising to randomize linear codes $\vec{A}$ while not amplifying the error weight of $\vec{e}$ excessively. 
Indeed, a distribution of Cliffords that is not sparse but somehow does not amplify almost all low-weight errors seems unlikely to exist. 
We show that no distribution of sparse Cliffords suffices either, because the uniform distribution over stabilizer subspaces has entropy $\Omega(n^2)$, but any distribution of sparse Cliffords has entropy $o(n^2)$. 

A similarly challenging subtlety is that it is technically only necessary to scramble the distribution of \emph{subspaces}, not the basis representing the subspaces.
In other words, it is possible to have a distribution over stabilizer tableaus which is far from uniform, but represents a uniformly random stabilizer subspace.
While possible, such distributions are exotic and no explicit examples are known.
Unfortunately, we show significant barriers against a reduction that randomizes the tableau.
We prove unconditionally using a symmetry argument that no tableau-randomizing reduction exists if the worst-case code is arbitrary.
The question remains as to whether a reduction could exist for a restricted set of codes.
We study this question under a concrete model, namely a sequence of random local Cliffords.
This model has some semblance of locality, yet is not bounded absolutely in sparsity.
In this setting, we show that it is impossible to randomize over stabilizer tableaus without also randomizing errors. 
The key method is by characterizing the rate of Pauli mixing---in order for a distribution of Clifford operators to randomize stabilizer tableaus (not just the stabilizer subspace but a basis), the distribution must be randomizing at least one Pauli, namely one of the stabilizers. 
But our Pauli mixing results show that the rate of mixing of any two Paulis are nearly identical, so that if one Pauli is taken close to a uniform distribution, then the error $\vec{E}$ must also be randomized soon after. 
It is therefore impossible to use this model to randomize the stabilizer tableau without also amplifying the error. 
The techniques for this specific model show that even randomizing a single Pauli forces any others to be randomized as well, a result that is interestingly much stronger than necessary here, and that could be of independent interest. 
Taken together, these results imply that if any notion of a random self-reduction exists for stabilizer decoding, it must satisfy very exotic properties, including (a) having unbounded sparsity while also not mapping most Paulis to maximal weight and (b) scrambling the code spaces without scrambling the stabilizer tableau.
Alternatively, the reduction would have to increase the dimensions $k, n$ of the code, a slightly unconventional model that we do not consider here.

\subsection{Related Work}
Classically, the hardness of decoding a random classical linear code has been studied thoroughly~\cite{10.1007/978-3-642-03356-8_35,gilbert2008encrypt,10.1007/11535218_18,10.1109/FOCS.2006.51,arora2011new,Alekhnovich03,10.1007/978-3-642-27660-6_9,10.5555/647097.717000,10.1007/11538462_32,yu2019collision}.
Often, $\lpn$ is phrased in terms of a \emph{query problem}: an oracle which knows the secret $\vec x \in \Z_2^k$ can be queried, at which point it gives a random draw of $\vec{a}_i \sim \Z_2^k$ and $e_i \sim \Ber(p)$ and returns $\vec{a}_i \cdot \vec x + e_i$.
An efficient algorithm, for example, is allowed to query oracle as many times $n$ as it would like, bounded only by being $n = \poly(k)$ for efficiency's sake.
For the purpose of coding theory, however, we are interested in a more fine-grained formulation.
Specifically, the goodness of a code is loosely measured by its rate and distance.
The distance is not within our control, but the rate $R = \frac{k}n$, where $n$ is the number of physical bits, or, equivalently, the number of queries.
Therefore, in our work, we use $n$ to specifically study $\lpn$ with a fixed number of queries.
This formulation explicitly interprets $\lpn$ as decoding a randomly chosen code of a given rate, and also makes clear the valid parameter regime for certain attack algorithms.
For example, \cite{BKW03} gives a $2^{O(k / \log k)}$ time algorithm which operates in the regime when $n$ is superpolynomial in $k$, and \cite{10.1007/11538462_32} gives a $2^{O(k / \log \log k)}$ time algorithm which operates in the regime where $n$ is slightly more than linear in $k$.
Importantly, there is no sub-exponential algorithm known after many decades of effort in the regime where $n = \Theta(k)$, i.e. at constant rate.

Until the work of \cite{poremba2025learningstabilizersnoiseproblem}, the study of decoding a stabilizer code has focused on the worst case, and usually on a certain variant of decoding known as maximum-likelihood decoding~\cite{HsiehLeGall11,iyer2015hardness,kuo2020hardnesses}.
Such a model has two flavors in the quantum realm. 
First, given a noisy code state, the task is to output the most likely error to have been applied, assuming that the error was drawn from a given fixed distribution such as the depolarizing channel.
The second flavor is known as degenerate maximum-likelihood decoding, and arises from the fact that many Paulis act identically on a code state; these are the \emph{stabilizers} of the code.
Therefore, one can instead ask what the most likely equivalence class of errors occurred given the noisy code state, i.e. the set of all errors which differ mutually by a stabilizer.
The presence of such an equivalence class is an inherently quantum effect known as \emph{quantum degeneracy}.
Because there are exponentially many stabilizers, the latter model is known to be \textbf{\#P}-complete in the worst case~\cite{iyer2015hardness}, rather than \textbf{NP}-complete as in the former case~\cite{HsiehLeGall11,kuo2020hardnesses}.
Intuitively, this upgraded hardness arises from the fact that a carefully designed stabilizer code can have error equivalence class with extremely similar probabilities.
As a result, one must add up exponentially many error probabilities to compute the probability of each equivalence class and compare them, giving an intractable counting problem.

In practice, however, the purpose of decoding is unambiguous---to recover the underlying logical state.
It does not matter how one chooses to approach this recovery, whether through maximum likelihood decoding or any other technique.
Our formulation of $\lsn$---both in worst- and average-case variants---is therefore to recover the logical state, consistent with~\cite{poremba2025learningstabilizersnoiseproblem}.
Moreover, given a random stabilizer code, it is not clear whether the effects quantum degeneracy play any significant role in complexity at all.
As a consequence, \cite{iyer2015hardness} remarked that an important question of study is precisely Question~\ref{question:main_question}, and the role of quantum degeneracy in it.
Part of our contributions may be viewed as a direct answer to these questions---in our reduction from $\lpn$ at constant rate to $\lsn$ with even a single qubit, we show that this increased hardness of decoding random stabilizer codes is explicit sourced from quantum degeneracy.
Therefore, in the worst case, quantum degeneracy lifts certain models of stabilizer decoding into a much harder complexity class, while in the average case, quantum degeneracy lifts the easiest random quantum decoding problem to at least the level of the hardest random classical decoding problem.

As discussed previously, the formal definition of decoding random stabilizer codes was proposed only recently by \cite{poremba2025learningstabilizersnoiseproblem}.
Their initial assessment of this problem, $\lsn$, included an exponential-time attack algorithm and a proof that $\lsn$ is contained within a certain unitary complexity class known as $\mathsf{unitaryBQP}^{\mathsf{NP}}$.
They do not, however, make progress on Question~\ref{question:main_question}. 
Namely, they explicitly leave Questions~\ref{question:LPN_to_LSN?} and  \ref{question:attack_algorithms} open.
Moreover, they consider only random self-reductions in the context of self-reducibility as phrased in Question~\ref{question:self_reducibility}, and are only able to show a random self-reduction that (a) reduces into an unnatural distribution of problem instances that depends extremely heavily on the worst case instance, which is not the conventionally accepted definition of an average-case problem, (b) succeeds with at most inverse-quasipolynomial probability, which is negligible, and (c) only applies in a regime of worst-case instances wherein even brute-force computation succeeds in quasipolynomial time.
Despite these limitations of their reduction, it is unclear as to how it could be improved at all.
In addition to answering the remaining questions about self-reducibility (between search and decision), this work sheds light into why the random self-reduction of \cite{poremba2025learningstabilizersnoiseproblem} is difficult to improve upon, by proving that reasonable reductions do not exist at all.

\paragraph{Organization.} 
The remainder of the paper is organized as follows. 
In Section~\ref{sec:preliminaries}, we give the requisite facts about quantum coding theory and random quantum states to which we refer in the remainder of the paper.
Section~\ref{sec:formal_defs} formally defines the $\lsn$ problem, as well as its worst-case variant, which we call the quantum nearest code state problem ($\qncp$) in reference to the classical nearest codeword problem ($\ncp$).
Here, we define multiple variants of these problems---all of which being classically equivalent---and discuss how certain quantum phenomena render them inequivalent (or, in cases where they are non-trivially equivalent, we prove them).
In Section~\ref{sec:classical_lsn}, we introduce a useful representation of the $\lsn$ task which is completely classical.
The classical representation is a key technical tool for our proofs.
We show in Section~\ref{sec:self_reducibility} that $\lsn$ satisfies a certain search-to-decision reduction, as well as a decision-to-search reduction (which is non-trivial, as we will show, because solutions are not efficiently verifiable).
However, under certain structural assumptions, $\lsn$ does not admit a random self-reduction.
In Section~\ref{sec:hardness_lsn}, we prove that $\lpn$ with constant rate reduces to $\lsn$ with arbitrary $k$, including $k=1$.
Finally, we conclude with discussion in Section~\ref{sec:discussion}.

\section{Preliminaries} \label{sec:preliminaries}
We establish our notation for common objects in stabilizer coding theory, including the Pauli and Clifford groups, stabilizer codes, and the noise model we consider.

\subsection{Pauli and Clifford Groups}
Let $\mathrm{U}(d)$ be the group of $d \times d$ unitary matrices.
The operators $\vec{I}, \vec{X}, \vec{Y}, \vec{Z} \in \mathrm{U}(2)$ denote the usual Pauli operators, i.e.
\begin{align}
    \vec{I} := \ketbra{0}{0} + \ketbra{1}{1},\; \vec{X} :=  \ketbra{0}{1} + \ketbra{1}{0},\; \vec{Y} := -i \ketbra{0}{1} + i\ketbra{1}{0},\;  \vec{Z} := \ketbra{0}{0} - \ketbra{1}{1}.
\end{align}
The Pauli operators square to the identity, and further satisfy \begin{align} \label{eq:Pauli_product_rule}
    \vec{P}_i \vec{P}_j = - \vec{P}_j \vec{P}_i, \;\vec{XY} = \vec{Z}, \;\vec{YZ} = \vec{X}, \;\vec{ZX} = \vec{Y} .
\end{align}
A $n$-qubit Pauli operator is an operator of the form $\vec{P}= \vec{P}_1 \otimes \cdots \otimes \vec{P}_n \in (\C^{2 \times 2})^{\otimes n}$, where $\vec{P}_i \in \set{\vec{I}, \vec{X}, \vec{Y}, \vec{Z}}$.
When convenient, we omit the explicit $\otimes$ and write the Paulis as strings $\vec{P}_1, \ldots, \vec{P}_n$.
The \emph{weight} of $\vec{P}$, denoted $\wt(\vec{P})$ is given by the number of operators in the string which are not $\vec{I}$.
The set of $n$-qubit Pauli operators form a group in the following way.
\begin{definition}[Pauli group] \label{def:Pauli_group}
Let
\begin{align}
    \widetilde{\CP}_n := \set{c \vec{P}^{(1)} \otimes \cdots \otimes \vec{P}^{(n)} \,:\, c \in \set{\pm 1, \pm i}, \vec{P}^{(i)} \in \set{\vec{I}, \vec{X}, \vec{Y}, \vec{Z}}}
\end{align}
be the set of $n$-qubit Pauli operators with a phase. 
By Eqn.~(\ref{eq:Pauli_product_rule}), $\widetilde{\CP}_n$ is a group under the standard product, known as the $n$-qubit \emph{Weyl-Heisenberg group}.
We call $\CP_n := \widetilde{\CP}_n / \langle i \vec{I} \rangle$ the $n$-qubit \emph{Pauli group}, which consists of all Pauli strings with no phase. 
\end{definition}
Almost always, when we refer to a Pauli, we are implicitly excluding the identity string $\vec{I}$; this is inferred from context.
We next define the Clifford group, a larger group of operators whose normalizer is the Pauli group.
\begin{definition}[Clifford group] \label{def:Clifford_group}
    The $n$-qubit Clifford group $\CC_n$ is the set
    \begin{align}
        \CC_n := \set{\vec{V} \in \mathrm{U}(2^n) \,:\, \vec{V} \vec{P} \vec{V}^\dagger \in \CP_n, \,\forall \vec{P} \in \CP_n}
    \end{align}
    under the standard product.
\end{definition}
Let $\ket{\pm} = \frac{1}{\sqrt{2}}(\ket{0} \pm \ket{1})$.
A basic fact about the Clifford group is that it can be generated by gates $\vec{H}_i$, $\vec{S}_i$, and $\vec{CX}_{ij}$~\cite{nielsen2010quantum}, where
\begin{align}
    \vec{H}_i & := \frac{1}{\sqrt{2}} (\ketbra{+}{0}_i + \ketbra{-}{1}_i) ,\; \vec{S}_i := \ketbra{0}{0}_i + i \ketbra{1}{1}_i , \\
    \vec{CX}_{ij} & := \ketbra{00}{00}_{ij} + \ketbra{01}{01}_{ij} + \ketbra{11}{10}_{ij} + \ketbra{10}{11}_{ij} .
\end{align}
The subscript on $i$ indicates that the gate only acts on the $i^\text{th}$ qubit. The coefficient $i$ in the description of $\vec{S}_i$ is $\sqrt{-1}$.
Implicitly, the remaining qubits all have the identity acted upon them.
By generated, we mean that every $\vec{C} \in \CC_n$ can be written as a finite product of these generators.
\begin{definition}[Clifford circuit] \label{def:clifford_circuit}
    A $n$-qubit Clifford circuit is a sequence of gates, each of which is $\vec{H}_{i}$, $\vec{S}_i$, or $\vec{CX}_{ij}$ for some $i, j$. A specification of this circuit is done so classically.
\end{definition}

\subsection{Stabilizer Codes and Logical Operations}
A set of $n$-qubit Paulis $\vec{P}_1, \dots, \vec{P}_n \in \CP_n$ is called \emph{independent} if for any subset $\vec{P}_{i_1}, \dots, \vec{P}_{i_m}$, $\prod_{j=1}^m \vec{P}_{i_j} \neq \vec{I}$.
Equivalently, no Pauli in the set can be written as the product of other Paulis in the set.
\begin{definition}[Stabilizer codes] \label{def:stabilizer_codes}
    A $\llbracket n, k \rrbracket$ stabilizer code $S$ is the space of states which are $+1$ eigenstates of $n-k$ independent commuting Pauli strings $\vec{P}_1, \dots, \vec{P}_{n-k}$ in $\CP_n$. The set of $\llbracket n, k \rrbracket$ stabilizer codes is denoted $\Stab(n, k)$.
\end{definition} 
If a state is a linear combination of $+1$ eigenstates of $\vec{P}_1, \dots, \vec{P}_{n-k}$, then the state itself is also a $+1$ eigenstate of the same operators. Thus a stabilizer code is a linear subspace of $(\C^2)^{\otimes n}$.
One can specify a stabilizer code by providing a valid choice of $n-k$ Pauli strings.
The resulting $(n-k) \times n$ table of Pauli operators is called the \textit{stabilizer tableau} representation of the corresponding code.
Due to the Gottesman-Knill theorem~\cite{gottesman_book}, a stabilizer code is generated by a corresponding Clifford operator $\vec{U} \in \CC_n$ via the map $\ket{0^{n-k}} \otimes \ket{\psi} \to \vec{U} \ket{0^{n-k}} \otimes \ket{\psi} $. 
The stabilizers of this code are $\vec{U} \vec{Z}_1 \vec{U}^\dagger, \;\dots, \;\vec{U} \vec{Z}_{n-k} \vec{U}^\dagger$, where
\begin{align}
\begin{matrix}
\vec Z_1 &= \,\, \vec Z & \vec I & \vec I & \cdots & \vec I & \vec I & \cdots & \vec I\\
\vec Z_2 &=\,\, \vec I & \vec Z & \vec I & \cdots & \vec I & \vec I &\cdots & \vec I\\
&\vdots\,\,\,\,\,\, & \vdots & \vdots & \vdots &\ddots & \vdots & \ddots & \vdots \\
\vec Z_{n-k} &= \,\, \vec I & \vec I & \vec I & \cdots & \vec Z & \vec I &\cdots & \vec I
\end{matrix}    
\end{align}
\begin{remark}
    One way to generate a uniformly random $\llbracket n, k \rrbracket$ stabilizer code is to choose a uniformly random Clifford $\vec U \sim \CC_n$, and to output $\vec{U}\vec{Z}_1 \vec{U}^\dagger, \;\dots, \;\vec{U}\vec{Z}_{n-k} \vec{U}^\dagger$~\cite{Graeme_thesis,poremba2025learningstabilizersnoiseproblem}.
\end{remark}
The stabilizers $\vec{P}_1, \;\dots, \;\vec{P}_{n-k}$ themselves generate a subgroup of $\CP_n$, which we denote $\CS$.
Stabilizer codes generated by Paulis which are each either purely comprised of $\vec{X}$ operators or of $\vec{Z}$ operators are known as Calderbank-Steane-Shor (CSS) codes~\cite{CSS1,CSS2}.
CSS codes are in a certain sense the most general class of stabilizer codes which are ``essentially classical''.

We remark that two stabilizer tableaus can specify the same stabilizer code because there are many choices of basis for a vector space. 
Operations which preserve the code space are known as \emph{logical operators}. 
In this work, we will be primarily interested in logical operators which are themselves Pauli operations.
\begin{definition}[Logical Paulis] \label{def:logical_Paulis}
    Let $\vec{P}_1, \dots, \vec{P}_{n-k} \in \CP_n$ be a stabilizer tableau of a code $S$ which generates a subgroup $\CS \leq \CP_n$. The logical Paulis of $S$ is the set $\Log(S) = \set{\vec{Q} \in \CP_n \,:\, \vec{Q} \vec{P} \vec{Q} \in \CS, \forall \vec{P} \in \CS}$.
\end{definition}
The logical Paulis of $S$ themselves form a group.

\begin{definition}[Code distance]
    The distance of a code $S$ is $d = \min \set{\wt(\vec{P}) \,:\, \vec{P} \in \Log(S)}$.
\end{definition}
The detection of errors in stabilizer codes occurs by measuring the stabilizer generators $\vec{P}_1, \dots, \vec{P}_{n-k}$, which each yield a sign $\pm 1$ known collectively as a syndrome $\mathbf{a} \in \set{\pm 1}^{n-k}$.
$a_i = -1$ if and only if the error $\vec{E}$ and $\vec{P}_i$ anticommute. 
A basic fact about the code distance is that any Paulis of weight $w \leq \lfloor \frac{d-1}{2} \rfloor$ give unique stabilizer syndromes.

\paragraph{Symplectic Representation.}
It is often convenient to represent the stabilizer tableau as a binary matrix, which can be done by mapping Paulis into bitstrings known as symplectic representations.
\begin{definition}[Symplectic isomorphism] \label{def:symplectic_isomorphism}
    Let $\vec{P} = \vec{P}^{(1)} \otimes \cdots \otimes \vec{P}^{(n)} \in \CP_n$ be a (phase-free) $n$-qubit Pauli operator, where $\vec{P}^{(i)} \in \set{\vec{I}, \vec{X}, \vec{Y}, \vec{Z}} = \CP_1$. We define \begin{align}
        \Symp \,:\, \CP_n \to \Z_2^{2n}
    \end{align}
    to map $\vec{P}$ to $\vec{y}$, such that $(\vec{y}_{i}, \vec{y}_{i+n})$ is $(0, 0)$ if $\vec{P}^{(i)} = \vec{I}$, $(1, 0)$ if $\vec{P}^{(i)} = \vec{X}$, $(0, 1)$ if $\vec{P}^{(i)} = \vec{Z}$, and $(1, 1)$ if $\vec{P}^{(i)} = \vec{Y}$.
\end{definition}
$\Symp$ is a bijective map because there is a bijection between $\set{\vec{I}, \vec{X}, \vec{Y}, \vec{Z}}$ and $\Z_2^2$. 
It is also a homomorphism between $\CP_n$ as a group under the standard product and $\Z_2^{2n}$ as a group under addition.
As such, $\Symp$ is an isomorphism that completely characterizes elements of $\CP_n$.
\begin{definition}[Symplectic form and inner product] \label{def:symplectic_form_inner_product}
    Let \begin{align}
        \boldsymbol{\W} = \begin{bmatrix}
            0 & \vec{I} \\ \vec{I} & 0
        \end{bmatrix} \in \Z_2^{2n \times 2n}
    \end{align}
    be the \emph{symplectic form}. 
    The symplectic inner product between $\vec{v}, \vec{w} \in \Z_2^{2n}$ is defined as $\vec{v} \odot \vec{w} = \mathbf{v}^T \boldsymbol{\W} \vec{w}$. 
    If $\vec{v} \odot \vec{w} = 0$, then we say that $\vec{v}$ and $\vec{w}$ are \emph{symplectically orthogonal}.
\end{definition}
Vectors $\vec{v}$ and $\vec{w}$ are symplectically orthogonal if and only if their corresponding Paulis $\Symp^{-1}(\vec{v}), \;\Symp^{-1}(\vec{w}) \in \CP_n$ commute.

\begin{definition}[Symplectic check matrix] \label{def:symplectic_check_matrix}
    A matrix $\vec{H} \in \Z_2^{(n-k) \times 2n}$ is a \emph{symplectic check matrix} if every pair of rows in $\vec{H}$ are symplectically orthogonal. The symplectic check matrix corresponding to a stabilizer tableau $\CT = \set{\vec{P}_1, \;\dots, \;\vec{P}_{n-k}}$, denoted $\Symp(\CT) \in \Z_2^{2n \times m}$, is a symplectic check matrix such that the $i$th row of $\Symp(\CT)$ is $\Symp(\vec{P}_i)$.
\end{definition}
A symplectic check matrix is simply a bitstring representation of a valid stabilizer tableau. 
It follows that a uniformly random symplectic check matrix in $\Z_2^{(n-k) \times 2n}$ corresponds to a uniformly random $\llbracket n, k \rrbracket$ stabilizer code.

The stabilizer tableau/symplectic check matrix and the corresponding Clifford circuit are equivalent in the sense that one can be efficiently computed from the other.
\begin{theorem}[Equivalence of circuits and tableaus] \label{thm:circuit_tableau_equivalence}
    Given a stabilizer tableau of $n-k$ $n$-qubit Paulis $\vec{P}_1, \dots, \vec{P}_{n-k}$ or a symplectic check matrix, there exists a classical algorithm running in polynomial time which computes a description of a Clifford circuit $\vec{C}$ such that the tableau $\vec{C}\vec{Z}_1 \vec{C}^\dagger, \;\dots, \;\vec{C} \vec{Z}_{n-k} \vec{C}^\dagger$ generates the same stabilizer subgroup as $\vec{P}_1, \;\dots, \;\vec{P}_{n-k}$. The same holds vice versa: given $\vec{C}$, there is a polynomial time classical algorithm computing a corresponding stabilizer tableau.
\end{theorem}
\begin{proof}
    The forward direction is a classic result based on a symplectic variant of Gaussian elimination, and the backward direction is equivalent to the Gottesman-Knill theorem~\cite{gottesman_book,nielsen2010quantum}.
\end{proof}

In the symplectic representation, the syndrome $\mathbf{a} \in \set{\pm 1}^{n-k}$ transforms into $\mathbf{v} = \frac{1}{2} (1 -  \mathbf{a}) \in \Z_2^{n-k}$. 
If an error $\vec{E}$ gives stabilizer syndrome $\mathbf{a}$ on a code $S$, then $\vec{H \W} \Symp(\vec{E}) = \mathbf{v}$, where $\vec{H}$ is the corresponding check matrix.
$E$ is a stabilizer or logical operator if and only if 
$\vec{H \W} \Symp(\vec{E}) = 0$.

\begin{definition}[Sympectic representation of Clifford operator]
For any Clifford operator $\vec{C}$, it has an action on $\mathcal{P}_n$ by conjugation, taking $\vec{P}$ to $\vec{C}\vec{P}\vec{C}^{-1}$. 
We define the matrix $\vec{T} = \Symp(\vec{C}) \in \Z_2^{2n \times 2n}$ to be the corresponding linear map such that 
\begin{equation}
\vec{T}\cdot \Symp(\vec{P})=\Symp(\vec{C}\vec{P}\vec{C}^{-1}).
\end{equation}
\end{definition}
Linear maps $\vec{T} \in \Z_2^{2n \times 2n}$ that preserve the symplectic inner product $\odot$ correspond precisely to Clifford operators $\vec{C}$. 
Furthermore, any two Cliffords that have the same symplectic representation differ by a Pauli operator, i.e. if $\Symp(\vec{C}_1) = \Symp(\vec{C}_2)$, then $\vec{C}_1 = \vec{P}\vec{C}_2$ for some Pauli operator $\vec{P}$.

\subsection{Noise Models}
For the purpose of this work, a noise model is a distribution of linear operators which are applied to a quantum state.
It is known that any noise model is equivalent to a noise model in which the support of the distribution consists exclusively of $n$-qubit Pauli operators~\cite{PhysRevA.52.R2493}.
The most standard such model---and the model in which we will be interested---is a model in which every Pauli occurs with equal probability.
\begin{definition}[Depolarizing channel] \label{def:depolarizing_channel}
    The 1-qubit depolarizing channel $\CD_p$ with parameter $p \in [0, \frac34]$ is a distribution over $1$-qubit Pauli operators such that $\Pr[\vec{I}] = 1 - p$, and $\Pr[\vec{P}] = \frac{p}{3}$ for each $\vec{P} \in \set{\vec{X}, \vec{Y}, \vec{Z}}$.
    As a quantum channel, the action of $\CD_p$ on a state $\rho \in \C^{2 \times 2}$ is given by \begin{align}
        \CD_p(\rho) = (1 - p) \rho + \frac{p}{3} \vec{X} \rho \vec{X} + \frac{p}{3} \vec{Y} \r \vec{Y} + \frac{p}{3} \vec{Z} \rho \vec{Z} = \left(1 - \frac{4}{3} p\right) \rho + \frac{4}{3} p \operatorname{Tr}[\rho] \frac{\vec{I}}{2} .
    \end{align}
\end{definition}
For a $n$-qubit state, we use the noise model $\CD_p^{\otimes n}$, i.e. independent single-qubit depolarizing noise on each qubit.
Such a model is the most widely considered for purposes of using quantum error-correcting codes as a \emph{quantum memory}, though there are more sophisticated models required for error-corrected quantum \emph{computation} which we do not consider here.
Nevertheless, even for models restricted to Pauli errors, there are many reasonable distributions that one can consider for practical purposes.
For example, we can decouple the $\vec{X}$ and $\vec{Z}$ errors (so that the probability of a $\vec{Y}$ error is equal to the probability that \textit{both} an $\vec{X}$ and a $\vec{Z}$ error occur), or bias the noise so that one Pauli error is substantially more likely than the others.
While such models are certainly useful as techniques to capture the exact noise model of a certain quantum device, they are not natural relative to uniformly random stabilizer codes, because such a code is symmetric under a relabeling of $\vec{X}$, $\vec{Y}$, and $\vec{Z}$.
Therefore, for the purpose of studying the complexity of decoding a random stabilizer code, it is natural to fix $\CD_p^{\otimes n}$ as the error model.

Recall that classically, the noise model is always given by $\Ber(p)$ independently for each bit.
We now show that there is a duality between the Bernoulli and depolarizing distributions, in the sense that we can always decompose a depolarizing distribution with probability parameter $p$ into  the product of independent random $\vec X, \vec Y$, and $\vec Z$ errors each with a different probability $q$.
In the symplectic representation, this decomposition appears as the sum of three independent Bernoulli-type random variables.

\begin{lemma}[Depolarizing decomposition]
    Let $b_{\vec X}, b_{\vec Y}, b_{\vec Z} \sim \Ber(q)$ be independent random bits.
    Then $\vec P = \vec X^{b_{\vec X}} \vec Y^{b_{\vec Y}} \vec Z^{b_{\vec Z}}$ (as a phase-free Pauli) is distributed as $\CD_p$, where $p = 3[q^2(1-q) + q(1-q)^2]$ for any $q \in [0, \frac{1}{2}]$ and $p \in [0, \frac34]$.
\end{lemma}
\begin{proof}
    Under $\CD_{p}$, the probabilities of $\vec I, \;\vec X, \;\vec Y, \;\vec Z$ are respectively $1-p, \;p/3, \;p/3, \;p/3$.
    Consider the distribution of a sampling process in which we apply a $\vec X$ error with probability $q$, and independently a $\vec Z$ error with probability $q$.
    Then the probabilities of $\vec I, \;\vec X, \;\vec Y, \;\vec Z$ are respectively $(1-q)^2, \;q(1-q), \;q^2, \;q(1-q)$.
    Now, also apply a $\vec Y$ error with probability $q$.
    Then the probabilities of $\vec I, \;\vec X, \;\vec Y, \;\vec Z$ become $(1-q)^3 + q^3, \;q^2(1-q) + q(1-q)^2, \;q^2(1-q) + q(1-q)^2, \;q^2(1-q) + q(1-q)^2$.
    This is precisely the distribution of $\CD_p$ as claimed.
\end{proof}

\begin{corollary}[Bernoulli representation of depolarizing noise] \label{corollary:bernoulli-depolarizing_duality}
    For $p \in [0, \frac34]$, let $\vec e \in \Z_2^{2n}$ have a distribution such that $\Symp^{-1}(\vec e) \sim \CD_p^{\otimes n}$.
    Then $\vec e$ is equidistributed as $(\vec e_{\vec X}, \vec e_{\vec Z}) + (\vec e_{\vec Y}, \vec e_{\vec Y})$, 
    where $e_{\vec X}, e_{\vec Y}, e_{\vec Z} \in \Z_2^n$ are i.i.d. $\Ber(q)^{\otimes n}$ and $q$ is the unique solution to $p = 3[q^2(1-q) + q(1-q)^2] \leq 3q$ such that $q \in [0, \frac{1}{2}]$.
\end{corollary}

It is also useful to observe that given a certain Bernoulli or depolarizing random variable, we can always add another Bernoulli or depolarizing random variable to increase the noise parameter to a desired level.

\begin{lemma}[Bernoulli convolution] \label{lemma:Bern_convolution}
    Let $\vec y \sim \Ber(p)$ for $p \in [0, \frac{1}{2}]$. Then for any $q \in [p, \frac{1}{2}]$, $\vec y + \vec y' \sim \Ber(q)$, where $\vec y' \sim \Ber(u)$ is an independent random variable and $u = \frac{q-p}{1 - 2p}$.
\end{lemma}
\begin{proof}
    $\Pr[\vec y + \vec y' = 1] = p(1-u) + u(1-p) = p + u - 2 p u = p + (q-p) = q$.
\end{proof}

\begin{lemma}[Depolarizing convolution] \label{lemma:depolarizing_convolution}
    Let $\vec P_1, \vec P_2 \in \CP_1$ be single-qubit phase-free Paulis such that $\vec P_1 \sim \CD_p$ and, independently, $\vec P_2 \sim \CD_u$ where $u = \frac{q-p}{1 - \frac{4}{3} p}$ for any $p \in [0, \frac34]$ and $q \in [p, \frac34]$.
    Then $\vec P_1 \vec P_2 \sim \CD_{q}$.
\end{lemma}
\begin{proof}
    Let $\vec P = \vec P_1 \vec P_2$. Then $\Pr[\vec P = \vec I] = (1-p)(1-u) + 3 \frac{p}{3} \frac{u}{3}$ since $\vec P = I$ if and only if $\vec P_1 = \vec P_2$.
    Likewise, for each $\vec P_0 \in \set{\vec X, \vec Y, \vec Z}$, $\Pr[\vec P = \vec P_0] = (1-p)\frac{u}{3} + \frac{p}{3} (1-u) + 2 \frac{p}{3} \frac{u}{3}$.
    Thus, $\vec P \sim \CD_{(1-p) u + (1-u)p + \frac{2}{3} pu}$.
    By simplifying, $(1-p) u + (1-u)p + \frac{2}{3} pu = p + u - \frac{4}{3} pu = p + p + (1-p) = q$ as claimed.
\end{proof}

\subsection{Haar Measure Over States}
To generate instances of $\lpn$, we sample a random logical state $\vec x \in \Z_2^k$ and encode it into a random code before adding noise.
Quantumly, one way to analogously generate a random logical state is to choose a uniformly random $k$-qubit state.
This notion of a uniformly random state is given by the Haar measure~\cite{mele2024introduction}.

\begin{definition}[Haar measure] \label{def:haar_measure}
    The Haar measure $\mu$ over the group of $D \times D$ unitary matrices $\text{U}(D)$ is the unique probability measure which is invariant under translation. That is, for all $\vec{V} \in \text{U}(D)$ and integrable functions $f$,
    \begin{align}
        \int f(\vec{V}\vec{U}) \, d\m(\vec{U}) = \int f(\vec{U}) \, d\m(\vec{U}) = \int f(\vec{U}\vec{V}) \, d\m(\vec{U}) .
    \end{align}
\end{definition}
We refer to a \emph{Haar-random state} $\ket{\psi} \sim \mu_k$ to be a state constructed by sampling $\vec{U} \sim \mu$ a unitary drawn from the Haar measure over $\text{U}(2^k)$ and then letting $\ket{\psi} = \vec{U} \ket{0^k}$.
It will be useful in this work to study the measurement distribution of a Haar-random state.
We therefore give the following lemma.

\begin{lemma}[Measurement probabilities are marginally $\text{\textsf{Beta}}$-distributed]
\label{lemma:haar-beta_distribution}
Let $\ket{\psi} \sim \mu_k$ and let $K = 2^k$.
For $\vec x \in \Z_2^k$, let $p_{\psi}(\vec x) = |\langle \vec x | \psi \rangle|^2$.
Then for all $\vec x \in \Z_2^k$, $p_{\psi}(\vec x) \sim \text{\textsf{Beta}}(1, K-1)$, where the $\text{\textsf{Beta}}(a, b)$ distribution has density $f(z) \propto z^{a - 1} (1-z)^{b-1}$ for $z \in [0, 1]$ and $f(z) = 0$ otherwise.
\end{lemma}
\begin{proof}
    One way to generate a Haar-random state is via standard complex Gaussians, i.e. for $\ket{\psi} \propto \sum_{\vec x \in \Z_2^k} \a_{\vec x} \ket{\vec x}$, let $\a_{\vec x} = Z_{\text{Re}, \vec x} + i Z_{\text{Im}, \vec x}$, where all $Z_{\text{Re}, \vec x} , Z_{\text{Im}, \vec x}$ across $\vec x$ are i.i.d. standard Gaussians.
    After performing this sampling, we normalize to obtain a valid quantum state.
    Thus, \begin{align}
        p_{\psi}(\vec x) = |\a_{\vec x}|^2 = \frac{Z_{\text{Re}, \vec x}^2 + Z_{\text{Im}, \vec x}^2}{\sum_{\vec y \in \Z_2^k} Z_{\text{Re}, \vec y}^2 + Z_{\text{Im}, \vec y}^2} = \frac{W_{\vec x}}{W_{\vec x} + \overline{W}_{\vec x}} ,
    \end{align}
    where $W_{\vec x} = Z_{\text{Re}, \vec x}^2 + Z_{\text{Im}, \vec x}^2$ and $\overline{W}_{\vec x} = \sum_{\vec y \in \Z_2^k : \vec y \neq \vec x} Z_{\text{Re}, \vec y}^2 + Z_{\text{Im}, \vec y}^2$.
    Note that $W_{\vec x} \sim \chi^2_{2}$ and $\overline{W}_{\vec x} \sim \chi^2_{K-2}$---where $\chi^2_{L}$ is the chi-squared distribution with $L$ degrees of freedom---and the two random variables are independent.
    The chi-squared distribution is a special case of the Gamma distribution; in particular, $\chi^2_{L} = \text{\textsf{Gamma}}(L/2, 2)$.
    Thus, $W_{\vec x} \sim \text{\textsf{Gamma}}(1, 2)$ and independently $\overline{W}_{\vec x} \sim \text{\textsf{Gamma}}(K-1, 2)$.
    By the Beta-Gamma duality (see, e.g., \cite{blitzstein2019introduction}), if $X \sim \text{\textsf{Gamma}}(a, \lambda)$ and independently $Y \sim \text{\textsf{Gamma}}(b, \lambda)$, then $X/(X+Y) \sim \text{\textsf{Beta}}(a, b)$.
    Thus, \begin{align}
        p_{\psi}(\vec x) = \frac{W_{\vec x}}{W_{\vec x} + \overline{W}_{\vec x}} \sim \text{\textsf{Beta}}(1, K-1) 
    \end{align}
    as claimed.
\end{proof}

It will also be useful to consider states which are not fully Haar-random, but agree with the distribution of a Haar-random state up to the first $t$ moments. 
Such objects are known as state $t$-designs.
\begin{definition}[State $t$-design] \label{def:t_design}
    An ensemble $\nu$ over $k$-qubit states is called a state $t$-design, if \begin{align}
        \mathbb{E}_{\ket{\psi} \sim \nu}[\ketbra{\psi}{\psi}^{\otimes t}] = \mathbb{E}_{\ket{\psi} \sim \mu_k}[\ketbra{\psi}{\psi}^{\otimes t}] .
    \end{align}
\end{definition}
Intuitively, an ensemble $\nu$ over states is a $t$-design if it exactly replicates the distribution over $t$ copies of a random state sampled according to the Haar measure $\mu_k$.
We exclusively work with \emph{exact} state $t$-designs, which exist for all $t$~\cite{Roy_2009}. However, our results would also carry over in the case of approximate state $t$-designs with minor modifications, provided the error is sufficiently small.

\subsection{Distances between Distributions and States}
Given two probability distributions $P, Q$ over $\Z_2^m$, the \emph{total variation}  (TV) distance between them is \begin{align}
    \label{eq:TV_fundamental_bound}
    \TV(P, Q) = \frac{1}{2} \sum_{\mathbf{x} \in \Z_2^m} |P(\mathbf{x}) - Q(\mathbf{x})| = \max_{A \subseteq \Z_2^m} (P(A) - Q(A)) .
\end{align}
The right-hand side expresses an operational interpretation of the TV distance as the best partitioning $A$ of the sample space $\Z_2^m$ for the purpose of distinguishing $P$ and $Q$.
The TV distance between $P$ and $Q$ bounds the difference in even probabilities when we draw from $Q$ instead of $P$, in the sense that \begin{align}
\Pr_P[A] - \TV(P, Q) \leq \Pr_Q[A] \leq \Pr_P[A] + \TV(P, Q) .
\end{align}
Being an $L^1$ norm, the TV distance is a bona fide metric in the sense that it is positive-definite, symmetric, and satisfies the triangle inequality.
A basic fact about the TV distance is that it is \emph{monotonic} under marginalization. 
That is, for a joint distribution $(P, Q)$, \begin{align}
    \label{eq:TV_monotonicity}
    \TV(P_1, P_2) \leq \TV((P_1, Q_1), (P_2, Q_2)) .
\end{align}

Quantumly, we require also a distance measures between states.
Given two pure quantum states $\ket{\psi}, \ket{\phi} \in (\C^2)^m$, we use the \emph{fidelity} between the states, which is given by \begin{align}
    \mathsf{F}(\ket{\psi}, \ket{\phi}) = |\langle \psi | \phi \rangle|^2 .
\end{align}
If instead $\ket{\psi}$ is a mixed state $\r \in (\C^{2\times 2})^{\otimes m}$, then $\mathsf{F}(\rho, \ket{\phi}) = \langle \phi | \r | \phi \rangle$.
A related distance measure is the \emph{trace distance} $\mathsf{T}(\rho, \sigma) = \frac{1}{2} \lVert \rho - \sigma \rVert_1$, where $\lVert \vec A \rVert_1$ is the sum of the absolute value of the eigenvalues of $\vec A$.
For our purposes, the only key facts about the trace distance is that the trace distance is negligible if and only if the infidelity $1 - \mathsf{F}$ is negligible~\cite{fuchs2002cryptographic} and that the trace distance decreases monotonically with quantum operations, i.e. for any quantum channel $\CE$, $\mathsf{T}(\CE(\rho), \CE(\sigma)) \leq \mathsf{T}(\rho, \sigma)$.

\section{Quantum Stabilizer Decoding: Problems and Equivalences}  \label{sec:formal_defs}

Classically the exact definition of decoding is quite robust in the sense that minor modifications to the precise details of the computational problem usually result in an equivalent task.
We will show that this intuition fails quantumly in the sense that certain parts of the definitions are surprisingly sensitive to technical modifications.
For that reason, we will define various flavors of decoding. These include three types: decision or search; the structure of the logical state; and whether the recovery involves finding the logical state, the minimum weight error, or a valid recovery operator.
Furthermore, we define versions in both the worst case and the average case, as they tend to exhibit different phenomena even \emph{qualitatively}. 

We remark that there are two general notions of decoding both in the classical and quantum literature. The first is in the context of Shannon theory, known as maximum likelihood decoding, in which we define a probabilistic error model with the goal of recovering the encoded state with high probability after the channel has been applied. The second is in the context of Hamming theory, in which we aim to recover the encoded state after an arbitrary error of weight at most $w$ has been applied.
The Hamming picture is only well-defined as a worst-case definition, and is indeed the definition which we adopt in the worst case.
In our average-case definitions, we adopt the Shannon picture using the depolarizing channel.
One can also formulate a worst-case definition in Shannon theory. 
Namely, given a corrupted encoded state, determine the error which has the highest probability under the channel to have produced this state.
This worst case definition in the Shannon picture has been studied by \cite{iyer2015hardness}, and variants of the definitions range from \textbf{NP}-complete to \textbf{\#P}-complete.

\subsection{Worst-Case Problems}

In this section, we will enumerate several worst-case decoding problems and draw equivalences between them. 
Each of these problems has a classical analogue, and some of them are equivalent in the classical setting as well. 
Firstly, there are the nearest codeword and error recovery problems: the classical task they correspond is recovering $\vec{x} \in \Z_2^{k}$ from $\vec{Ax} + \vec{e}$, where $\wt(\vec{e}) \leq w$, and $\vec{A} \in \Z_2^{n \times k}$. 
In the classical setting, there are three distinct ways of solving this problem, that are easily seen to be equivalent:
\begin{enumerate}
\item Find the secret $\vec{x}$.
\item Find a recovery operator $\vec{v}$, such that $(\vec{Ax}+\vec{e}) - \vec{v} = \vec{Ax}$. 
\item Find the error $\vec{e}$. 
\end{enumerate}
Given an oracle for any of these three tasks, the other two can also be solved. The second and third tasks are clearly equivalent, since the only possible value of $\vec{v}$ is the error $\vec{e}$. 
Meanwhile, knowing the error allows one to obtain the secret, and vice versa. 
We will define these three tasks quantumly---as a direct result of stabilizer degeneracy, they are no longer obviously equivalent. 
The most relevant decoding task is (1), since by construction it characterizes the task of decoding.
For this problem, we will also define a decision version of the task; in the sequel we will craft a search to decision reduction. 
We then define the remaining two problems, the error recovery versions. 
The relationship between all three definitions is important to characterizing the qualitative difference between classical and quantum decoding. 

We then define stabilizer decoding. 
Classically, stabilizer decoding is the problem of finding a low-weight vector $\vec{e}$ from its syndrome $\vec{Be}$, for $\vec{B} \in \Z_2^{(n-k) \times n}$.
This problem is dual to the nearest codeword problem: if $\vec{B} \in \Z_2^{(n-k) \times n}$ is the left kernel of $\vec{A}$, i.e. its rows form a basis of the left-kernel of $\vec{A}$, then the syndrome $\vec{B}(\vec{Ax}+\vec{e})=\vec{BAx}+\vec{Be} = \vec{Be}$ can be calculated. 
If the stabilizer decoding problem can be solved, then $\vec{e}$ may be recovered, solving the nearest codeword problem. 
On the other hand, given $\vec{Be}$, any value $\vec{y}$ so that $\vec{By} = \vec{Be}$ satisfies $\vec{y} = \vec{Ax} + \vec{e}$ for some $\vec{x} \in \Z_2^k$. 
Hence, if the nearest codeword problem can be solved, recovering $\vec{x}$ from $\vec{y}$ is sufficient to find $\vec{e}$ and decode the syndrome. 
We develop the quantum generalization of this duality.

Let $\ket{0^{n-k},\psi } :=  \ket{0^{n-k}} \otimes \ket{\psi}$.
Note that all equalities of states and Paulis are implicitly understood to be up to a phase.

\subsubsection{Nearest Codeword Problems}

Let us now give a formal definition of the quantum analog of the nearest codeword problem.

\begin{definition}[Quantum Nearest Codeword Problem, $\qncp$] \label{def:worst_cast_state_QNCP}

The quantum nearest codeword problem $\qncp(k,n,w)$ is characterized by integers $k,n,w \in \mathbb{N}$. Given as input \begin{align}
    (\vec C \in \CC_n, \;\vec E \vec C  \ket{0^{n-k},\psi})
\end{align}
where $\vec C$ is a Clifford encoding for a stabilizer subgroup $\mathcal{S}$ of code distance $d$ with $w \leq \lfloor \frac{d-1}{2} \rfloor$, where $\vec E \in \CP_n$ is a Pauli error, and where $\ket{\psi} \in (\C^2)^{\otimes k}$ is a logical state, we consider two variants:
 \begin{description}
     \item $\bullet$ \textbf{Search} $\qncp(k,n,w)$ is the task
of producing a $k$-qubit state within fidelity at least $\frac{2}{3}$ of $\ket{\psi}$;
     \item $\bullet$ \textbf{Decision} $\qncp(k,n,w)$ is the task of deciding whether there exists a Pauli $\vec E' \in \CP_n$ such that $\vec E' \vec E \in \CS$ with $\wt(\vec E') \leq w$ with advantage at least $\frac{2}{3}$.
 \end{description}
\end{definition}

\subsubsection{Error Recovery Problems}

We now give a formal definition of the two relevant error recovery problems.

\begin{definition}[Error Recovery Problems: $\errqncp$ and $\recqncp$] \label{def:worst_case_error_QNCP}

An error recovery problem is characterized by integers $k,n,w \in \mathbb{N}$. Given as input \begin{align}
    (\vec C \in \CC_n, \vec E \vec C  \ket{0^{n-k},\psi})
\end{align} 
where
 $\vec C$ is a Clifford encoding for a stabilizer subgroup $\mathcal{S}$ of code distance $d$ with $w \leq \lfloor \frac{d-1}{2} \rfloor$, where $\vec E \in \CP_n$ is a Pauli error, and where $\ket{\psi} \in (\C^2)^{\otimes k}$ is a logical state, we consider two variants:
 \begin{description}
     \item $\bullet$ $\recqncp(k,n,w)$ is the task
of finding an error Pauli $\vec E' \in \CP_n$, with probability at least $\frac{2}{3}$, such that $\vec E\vec E'$ belongs to the stabilizer subgroup $\mathcal{S}$; and we let
      \item $\bullet$ $\errqncp(k,n,w)$ is the task
of finding an error Pauli $\vec E' \in \CP_n$, with probability at least $\frac{2}{3}$, such that $\vec E\vec E'$ belongs to the stabilizer subgroup $\mathcal{S}$, and $\wt(E') \leq w$. 
 \end{description}
\end{definition}
$\qncp$ is analogous to the classical task of recovering $\vec{x}$ from $\vec{Ax}+\vec{e}$. 
$\recqncp$ corresponds to finding some recovery operator, i.e. some vector $\vec{v}$ for which $(\vec{Ax}+\vec{e}) - \vec{v} = \vec{Ax}$. 
Finally, $\errqncp$ corresponds to finding the error $\vec{e}$ itself, because it requires one to find not only a recovery operator, but a low-weight one.

\subsubsection{Syndrome Decoding Problems}
Here, for the purpose of comparison with the above decoding tasks, we define quantum syndrome decoding, using the symplectic representation of stabilizer codes.

\begin{definition}[Quantum Syndrome Decoding Problem, $\qsdp$] \label{def:decision_syndrome_decoding}
The quantum syndrome decoding problem $\qsdp(k,n,w)$ is characterized by integers $k,n,w \in \mathbb{N}$.
Given a symplectic stabilizer check matrix $\vec{H}\in \Z_2^{(n-k) \times 2n}$ (as in Definition~\ref{def:symplectic_check_matrix}) which implicitly defines a stabilizer subgroup $\CS$ whose code distance $d$ satisfies $w \leq \lfloor \frac{d-1}{2} \rfloor$, we let
\begin{description}
     \item $\bullet$ \textbf{Search} $\qsdp(k,n,w)$ be the task
of finding, with probability at least $\frac{2}{3}$, an error $\mathbf{e}' \in \Z_2^{n-k}$ such that $\Symp^{-1}(\mathbf{e}+\mathbf{e}') \in \CS$ when given as input $(\vec{H},\vec v)$ with $\mathbf{v} = \vec H \boldsymbol{\Omega}\mathbf{e}\Mod{2}  \in \Z_2^{n-k}$ and $\wt(\Symp^{-1}(\mathbf{e})) \leq w$; and we let
     \item $\bullet$ \textbf{Decision} $\qsdp(k,n,w)$ be the task
of determining, with probability at least $\frac{2}{3}$, if there exists an error $\mathbf{e}' \in \Z_2^{n-k}$ such that $\Symp^{-1}(\mathbf{e}+\mathbf{e}') \in \CS$ when given as input $(\vec{H},\vec v)$ with $\mathbf{v} = \vec H \boldsymbol{\Omega}\mathbf{e}\Mod{2}  \in \Z_2^{n-k}$ and $\wt(\Symp^{-1}(\mathbf{e})) \leq w$.
 \end{description}
\end{definition}

\begin{remark}
The success probability threshold $\frac{2}{3}$ was just chosen as a constant larger than $\frac{1}{2}$, as in each case this success probability can be amplified to $1 - \negl(n)$. 
For the decision problems, this fact is immediate. 
Running an oracle with advantage multiple times and taking the majority vote is sufficient to amplify the success probability. 
For $\errqncp$, the success probability amplification is also simple. 
Indeed, repeating the algorithm with the same input and choosing the majority outcome will yield the desired output with amplified success probability. 
Finally, for $\qsdp$, the success probability can be amplified as well. Each time the oracle is run, some error $\vec{e'}$ is obtained, and for a majority of these errors $\Symp^{-1}(\vec{e} + \vec{e'}) \in \mathcal{S}$. 
In particular, a majority of errors are in the same stabilizer equivalence class. 
Running the oracle repeatedly and choosing the majority equivalence class amplifies the probability of choosing the correct class. 

We will later show that for $\recqncp$ and $\qncp$, the success probability can be amplified as well, because both are equivalent to $\qsdp$.
Because both of these problems have quantum inputs, the possibility of amplification is not immediately obvious. 
\end{remark}

\subsubsection{Inequivalence of Quantum Definitions and the Short Vector Problem} \label{sec:inequivalence_and_SVP}

The task $\errqncp$ of finding $\vec{E}' \in \CP_n$ such that $\wt(\vec{E}') \leq w$ and $\vec{E}'\vec{E} \in \CS$ is at least as hard as $\recqncp(k, n, w)$.
However, the converse need not be true.
Suppose we find $\vec{E}'$ such that $\vec{E}' \vec{E} \in \CS$.
$\vec{E}'$ is a perfectly valid recovery map, and yet since $|\CS| = 2^{n-k}$ there are exponentially many (assuming $k = O(n)$ which is virtually always the case for practical considerations) possible recovery operators.
Finding one such operator whose weight is $\leq w$ is precisely a short vector problem over a certain subspace of $\Z_2^n$.
Unless this short vector problem can be solved efficiently in all cases, $\errqncp(k, n, w)$ does not reduce to $\recqncp(k, n, w)$.

At the same time, $\errqncp(k, n, w)$ is an unnaturally strong computational task as a model for decoding, since any valid recovery operator will successfully execute the decoding.
This additional constraint of weight is more accurately described as a misgeneralization of classical intuition which unnecessarily injects an additional short vector problem into the task which does not aid the actual recovery operation search in any way. 

In summary, we have presented three quantum worst-case decoding problem, and one dual syndrome decoding problem. 
The task $\errqncp(k, n, w)$ encodes short vector problems in such a way that significantly impacts the hardness of the task without actually improving the quality of the recovery operator, and we therefore do not consider it a reasonable definition.
The two other definitions---finding a valid recovery operator and finding the logical state itself---are both reasonable but not obviously equivalent.
We next show, however, that in fact the two definitions are equivalent, and are themselves both equivalent to Search $\qsdp$.
Therefore, we will for the remainder of this paper work exclusively with $\qncp$.

\subsubsection{Tight Equivalence Between Search Decoding Definitions}

We here study the relation between decoding, in the sense of finding the recovery operator or of purely finding the underlying logical state, and syndrome decoding. 
First, we show that in the Decision setting, $\qncp$ and $\qsdp$ are equivalent problems in an \emph{instance-by-instance} fashion.
That is, if there is a quantum algorithm to solve one problem for a particular code $S$ (with success probability $2/3$), then there is a quantum algorithm with the same runtime up to $\poly(n)$ factors, which solves the other problem on the same code $S$.
A useful interpretation of this result is that the most general operation one can perform in decoding a code in the worst case is to first measure the syndrome.
At least in the worst case then, decoding quantum codes is truly a classical computational problem in the sense that no quantum state processing is required to capture the complexity of the task.

\begin{lemma} \label{lemma:esqncp-qsdp_reduction}
    There exists an instance-by-instance reduction from Decision $\qncp[k, n, w]$ to Decision $\qsdp[k, n, w]$.
\end{lemma}
\begin{proof}
    Say that $\CA(\vec{C})$ is an oracle which solves Decision $\qsdp(k, n, w)$ with success probability at least $\frac{2}{3}$, taking as input a Clifford circuit $\vec{C}$ corresponding to a stabilizer code and an error weight $w$. 
    Define $\CB$ as follows.
    Given $\vec{C}$ and a noisy state $\ket{\phi} = \vec{E} \vec{C} \ket{0^{n-k}, \psi }$, first compute a corresponding stabilizer tableau $\CT = \set{\vec{P}_1, \;\dots, \;\vec{P}_{n-k}}$ and symplectic check matrix $\vec{H} \in \Z_2^{(n-k) \times 2n}$ via Theorem~\ref{thm:circuit_tableau_equivalence}. 
    Measure each stabilizer $\vec{P}_i$ on $\ket{\phi}$, getting a syndrome sign $a_i \in \set{\pm 1}$, and let $v_i = \frac{1}{2} (1 - a_i) \in \Z_2$. Run $\CA(\vec{C})$ on $(\vec{H}, \mathbf{v})$ and output its answer.

    Let $\CS$ be the stabilizer subgroup generated by $\CT$. 
    Two possible Pauli errors differ by a stabilizer if and only if they have the same stabilizer syndrome measurements since $w \leq \lfloor \frac{d-1}{2} \rfloor$. 
    Suppose $\exists \vec{E}' \in \CP_n$ such that $\vec{E}' \vec{E} \in \CS$ and $\wt(\vec{E}') \leq w$. 
    Then $\vec{E}\vec{C} \ket{0^{n-k}, \psi}$ generates the same syndrome as $\vec{E}' \vec{C} \ket{0^{n-k}, \psi}$. 
    Therefore, $\Symp(\vec{E}')$ is a valid solution on the $\qsdp$ instance  $(\vec{H}, \mathbf{v})$.
    Conversely, suppose $\exists \mathbf{e}' \in \Z_2^{n}$ such that $\vec{H \W} \mathbf{e}' = \mathbf{v}$ and $\wt(\Symp^{-1}(\mathbf{e}')) \leq w$. 
    Then $\vec{E}' = \Symp^{-1}(\mathbf{e}')$ produces the same syndrome as $\vec{E}$, and thus satisfies $\vec{E}\vec{E}' \in \CS$. 
    Thus, $\CB$ is correct if and only if $\CA(\vec{C})$ is correct, which by assumption happens with probability at least $\frac{2}{3}$.
\end{proof}

\begin{lemma} \label{lemma:qsdp-esqncp_reduction}
 There exists an instance-by-instance reduction from Decision $\qsdp[k, n, w]$ to Decision $\qncp[k, n, w]$.
\end{lemma}
\begin{proof}
    Say that $\CB(\vec{C})$ is an oracle which solves Decision $\qncp(k, n, w)$ with success probability at least $\frac{2}{3}$, for a Clifford circuit $\vec{C}$ and error weight $w$. 
    Define $\CA$ as follows. 
    Given a check matrix $\vec{H} \in \Z_2^{(n-k) \times 2n}$ and a symplectic syndrome $\mathbf{v} \in \Z_2^{n-k}$, compute the corresponding Clifford circuit $\vec{C} \in \CC_n$ via Theorem~\ref{thm:circuit_tableau_equivalence}. 
    Next, using Gaussian elimination, find any vector $\mathbf{e}'' \in \Z_2^n$ such that $\vec{H\W} \mathbf{e}'' = \mathbf{v}$. Let $\vec{E}'' = \Symp^{-1}(\mathbf{e}'') \in \CP_n$.
    Prepare the state $\ket{\phi_0} = \vec{E}'' \vec{C} \ket{0^n}$ and run $\CB(\vec{C})$ on $(\vec{C}, \ket{\phi_0})$ and output its answer.

    Suppose that $\exists \mathbf{e} \in \Z_2^{n}$ such that $\wt(\Symp^{-1}(\mathbf{e})) \leq w$ and $\vec{H\W} \mathbf{e} = \mathbf{v}$. 
    Then $\mathbf{e}'' = \mathbf{e} + \mathbf{y}$, where $\vec{H\W} \mathbf{y} = 0$. 
    Thus, $\Symp^{-1}(\mathbf{y})$ is either a stabilizer or a logical operator. 
    As a result, $\vec{E}'' \vec{C} \ket{0^n} = \vec{E C} \ket{0^{n-k}, \mathbf{x}}$, where $\mathbf{x}$ is defined by the logical operation of $\mathbf{y}$. This state is a structured instance of $\qncp(n, k ,w)$. Conversely, suppose that $\exists \vec{E} \in \CP_n$ with $\wt(\vec{E}) \leq w$ such that $\ket{\phi_0} = \vec{EC} \ket{0^{n-k}, \psi}$. Since $\ket{\phi_0} = \vec{E'' C} \ket{0^n}$, $\ket{\psi} = \ket{\mathbf{x}}$ for some $\mathbf{x}\in \Z_2^k$. Then $\Symp(\vec{E}) = \mathbf{e}$ induces the same syndrome as $\mathbf{e}''$, so the instance $(\vec{H}, \mathbf{v})$ is a structured instance for $\qsdp(k, n, w)$.
    Thus, $\mathcal{A}$ is correct precisely when $\mathcal{B}$ is, and so also has success probability at least $\frac{2}{3}$.
\end{proof}

\begin{corollary}\label{thm:Decision_sqncp-qsdp_equivalence}
Decision $\qncp[k, n, w]$ and Decision $\qsdp[k, n, w]$ are instance-by-instance equivalent. 
\end{corollary}
\begin{proof}
    Follows immediately from Lemmas~\ref{lemma:esqncp-qsdp_reduction} and \ref{lemma:qsdp-esqncp_reduction}.
\end{proof}

We next show equivalence of the search versions of the problem, using $\recqncp$ instead of $\qncp$.
\begin{theorem}\label{thm:Search_esqncp-qsdp_equivalence}
There exists instance-by-instance reductions from $\recqncp[k, n, w]$ to $\qsdp[k, n, w]$ and from $\qsdp[k, n, w]$ to $\recqncp[k, n, w]$. 
In particular, they are instance-by-instance equivalent. 
\end{theorem}
\begin{proof}
    The proof proceeds similarly to those of Lemmas~\ref{lemma:esqncp-qsdp_reduction} and \ref{lemma:qsdp-esqncp_reduction}.
The reduction procedures are identical.
The only differences are as follows.
Reducing to Search $\qsdp$, we obtain $\mathbf{e}' \in \Z_2^n$ such that $\Symp^{-1}(\mathbf{e}+\mathbf{e}') \in \CS$, where $\mathbf{e} = \Symp(\vec{E})$ for the true error $\vec{E}$ applied. Thus, we output $\vec{E}' = \Symp^{-1}(\mathbf{e}')$ to obtain the correct $\recqncp$ answer.
Reducing to Search $\recqncp$, we start with a trial guess $\mathbf{e}'' \in \Z_2^n$ obtained from Gaussian elimination which gives the correct syndrome $\mathbf{v} \in \Z_2^{n-k}$. $\vec{E}'' = \Symp^{-1}(\mathbf{e}'')$ may be written as $\vec{E}'' = \vec{E} \,\overline{\vec{X}}_{\mathbf{u}} \,\overline{\vec{Z}}_\mathbf{h}$ (for $\mathbf{u}, \mathbf{h} \in \Z_2^k$), where $\overline{\vec{X}}_{\mathbf{u}}$ is a logical $\vec{X}$ Pauli such that
\begin{align}
    \overline{\vec{X}}_{\mathbf{u}} \vec{C} \ket{0^n} = \vec{C} \vec{X}_{\mathbf{u}} \ket{0^n} = \vec{C} \ket{ 0^{n-k},\mathbf{u}} 
\end{align}
and $\overline{\vec{Z}}_\mathbf{h}$ is a logical $Z$ Pauli such that $\overline{\vec{Z}}_\mathbf{h} \vec{C} = \vec{C} \vec{Z}_{\mathbf{h}}$.
Thus,
\begin{align}
    \ket{\phi_0} = \vec{E}'' \vec{C} \ket{0^n} = \vec{E} \vec{C} \vec{X}_{\mathbf{u}} \vec{Z}_{\mathbf{h}} \ket{0^n} = \vec{E} \vec{C} \ket{0^{n-k},\mathbf{u} } .
\end{align}
By querying the $\recqncp$ solver with $(\vec{C}, \ket{\phi_0})$, we obtain some $\vec{E}' \in \CP_n$ such that $\vec{E}' \vec{E} \in \CS$ with probability at least $\frac{2}{3}$.
\end{proof}

Now, we will show that Search $\qncp$ and Search $\qsdp$ are in fact instance-by-instance equivalent as well. As a consequence of these two statements, we will conclude that Search $\recqncp$ and Search $\qncp$ are equivalent. 
\begin{theorem}[Instance-by-instance equivalence of Search $\qncp$ and Search $\qsdp$] \label{thm:Search_sqncp-qsdp_equivalence}
There exists (instance-by-instance) reductions from Search $\qncp[k, n, w]$ to Search $\qsdp[k, n, w]$ and from Search $\qsdp[k, n, w]$ to Search $\qncp[k, n, w]$. 
In particular, they are instance-by-instance equivalent. 
\end{theorem}
\begin{proof}
    Let the noisy state be $\ket{\phi} = \vec{E C} \ket{0^{n-k},\psi}$, where $\wt(\vec{E}) \leq w$.
    Reducing to $\qsdp$, we will obtain the recovery operator $\vec{E}' = \Symp^{-1}(\mathbf{e}')$, where $\mathbf{e}'$ is given by the $\qsdp$ solver. 
    With probability $\frac{2}{3}$, $\vec{E' E} \in \CS$ and thus by applying $\vec{E}'$ to $\ket{\phi}$, we recover $\vec{C} \ket{0^{n-k},\psi}$, at which point we apply $\vec{C}^\dagger$ and discard the last $n-k$ qubits to recover $\ket{\psi}$. 
    With probability at most $\frac{1}{3}$, the output state is arbitrarily far from $\ket{\psi}$, and we denote it $\rho_{\text{junk}}$. Hence the output $\rho$ is \begin{align}
        \r = \frac{2}{3}\ketbra{\psi}{\psi} + \frac{1}{3}\rho_{\text{junk}} 
    \end{align}
    which satisfies $\langle \psi | \rho |\psi \rangle \geq \frac{2}{3}$.

    Reducing to $\qncp$, we start with a trial guess $\mathbf{e}'' \in \Z_2^n$ obtained from Gaussian elimination which gives the correct symplectic syndrome $\mathbf{v} \in \Z_2^{n-k}$, i.e. $\vec{H\W} \mathbf{e}'' = \vec{H\W} \mathbf{e} = \mathbf{v}$, where $\mathbf{e} = \Symp(E)$ represents the true error and $\vec{H}$ is the parity check matrix of the $\qsdp$ instance.
    Defining $\vec{E}'' = \Symp^{-1}(\mathbf{e}'')$, we prepare $\ket{\phi_0} = \vec{E'' C} \ket{0^n}$. As in the proof of Theorem~\ref{thm:Search_esqncp-qsdp_equivalence}, $\vec{E}'' = \vec{E} \overline{\vec{X}}_{\mathbf{u}} \overline{\vec{Z}}_{\mathbf{h}}$ for $\mathbf{u}, \mathbf{h} \in \Z_2^k$. Thus, \begin{align}
        \ket{\phi_0} = \vec{E'' C} \ket{0^n} = \vec{E C} \vec{X}_{\mathbf{u}} \vec{Z}_{\mathbf{h}} \ket{0^{n}} = \vec{E C} \ket{0^{n-k}, \mathbf{u}} .
    \end{align}
    We submit $(\vec{C}, \ket{\phi_0})$ to the $\qncp$ solver and obtain a state $\r_0$. Similarly, we prepare \begin{align}
        \ket{\phi_+} = \vec{E'' C} \ket{0^{n-k},+^k} = \vec{E C }\vec{X}_{\mathbf{u}} \vec{Z}_{\mathbf{h}} \ket{0^{n-k},+^k } = \vec{E C} (\vec{I}^{\otimes (n-k)} \otimes \vec{H}^{\otimes k}) \ket{0^{n-k}, \mathbf{h} } ,
    \end{align}
    where $\vec{H}$ is the Hadamard gate (distinguished from a parity check matrix by context).
    Submitting $(\vec{C}, \ket{\phi_+})$ to the $\qncp$ solver give a state $\r_+$.
    By assumption, $\langle \mathbf{u} | \rho_0 | \mathbf{u} \rangle \geq \frac{2}{3}$, so with probability at least $\frac{2}{3}$ we obtain $\mathbf{u}$ when we measure $\rho_0$ in the computational basis. 
    Furthermore, we may repeat this step multiple times and output the majority outcome to obtain $\mathbf{u}$ with an amplified probability of at least $\frac{5}{6}$. 
    Likewise we obtain $\mathbf{h}$ with probability at least $\frac{5}{6}$ when we measure $\rho_+$ in the $\ket{\pm}$ basis. 
    By a union bound, both are correct with probability at least $\frac{2}{3}$. 
    From here, we efficiently compute logical Pauli representations $\overline{\vec{X}}_{\mathbf{u}}'$ and $\overline{\vec{Z}}_{\mathbf{h}}'$.
    We then output $\mathbf{e}' = \Symp(\vec{E}'' \overline{\vec{X}}_{\mathbf{u}}' \overline{\vec{Z}}_{\mathbf{h}}')$ as the solution to the $\qsdp$ problem.
    So long as $\mathbf{u}$ and $\mathbf{h}$ were computed correctly, $\vec{E''} \overline{\vec{X}}_{\mathbf{u}}' \overline{\vec{Z}}_{\mathbf{h}}'$ and $\vec{E}$ differ by at most a stabilizer.
    This completes the proof of correctness. 
\end{proof}

\begin{corollary}\label{corollary:Search_sqncp-esqncp_equivalence}
Search $\qncp[k, n, w]$ and $\recqncp[k, n, w]$ are instance-by-instance equivalent. 
\end{corollary}
\begin{proof}
    Follows immediately from Theorems~\ref{thm:Search_esqncp-qsdp_equivalence} and \ref{thm:Search_sqncp-qsdp_equivalence}.
\end{proof}

Henceforth, we shall work exclusively with Search $\qncp$ as the search variant of Decision $\qncp$.
We will not directly work with any syndrome decoding definitions, since they equivalent to a $\qncp$ analog.
\begin{remark}
As a consequence of the above equivalence, the success probability of $\qncp$ and $\recqncp$ may also be amplified, just like $\qsdp$. 
Indeed, there an algorithm for either of these problems implies one for $\qsdp$, whose success probability may be amplified. 
The reductions from $\qncp$ and $\recqncp$ to $\qsdp$ in Theorem \ref{thm:Search_sqncp-qsdp_equivalence} and Theorem \ref{thm:Search_esqncp-qsdp_equivalence} are easily seen to preserve arbitrary success probability, and therefore the amplified algorithm for $\qsdp$ yields one for these two problems as well. 
\end{remark}

\begin{remark}
    The instance-by-instance reductions involving $\qsdp$ do not hold in the average case, because our reductions rely on being able to carefully choose the right logical state and error to give to the solver. 
    In the average case, we do not have freedom to choose the logical state and error as we desire.
    Since $\qsdp$ is not the most natural definition of decoding (as compared to the direct state recovery task), we do not even define an average-case analog.
\end{remark}

\subsection{Average-Case Problems}
\label{sec:avg-case}

The main average-case decoding task of interest to us is termed \textit{Learning Stabilizers with Noise} ($\lsn$) and was first defined by \cite{poremba2025learningstabilizersnoiseproblem}.
The security parameter of the problem will be $n$, the number of physical qubits. 
The choice of $n$ as security parameter differs from the convention for $\lpn$. However, it is necessary because we will eventually relate $\lsn$ to other problems whose security depends solely on $n$; for example, when $k=1$.
Nevertheless, in the regimes where $n = \poly(k)$, the parameter $k$ may be equivalently be taken as the security parameter. Typically, we consider constant rates, where either $k = \Theta(n)$, or vice versa, and thus the overall asymptotics end up being the same.

\subsubsection{Learning Stabilizers with Noise}

The average-case quantum decoding task directly corresponding to worst-case problem $\qncp$ is $\slsn$, to be introduced shortly. 
But we will begin by introducing $\lsn$, which is an easier decoding problem where the output is a classical bitstring. 
Due to this simplifying fact, $\lsn$ has many favorable properties, which will allow us to establish self-reducibility properties and eventually relate it to $\lpn$, obtaining a lower bound on the hardness of $\slsn$ as well.  

\begin{definition}[Learning Stabilizers with Noise, $\lsn$]\label{def:search-LSN}
The \emph{Learning Stabilizers with Noise} problem, denoted by $\mathsf{LSN}(k, n, p)$, is characterized by integers $k,n \in \mathbb{N}$ and a noise parameter $p \in (0,1)$. Here, both $p$ and $k$ can vary with $n$. We consider two variants:
\begin{description}
    \item $\bullet$ \textbf{Search} $\mathsf{LSN}(k, n, p)$ is the following task: given as input a sample of the form \begin{align}
        \left(\vec C \sim \CC_n, \, \vec E  \vec C \ket{0^{n-k},\vec x}\right)  
    \end{align}
where $\vec C$ is a random $n$-qubit Clifford operator and $\vec E \sim \mathcal{D}_p^{\otimes n}$ is an $n$-qubit Pauli error, and $\vec x \sim \Z_2^k$ is a random logical basis state, the task is to find $\vec x$; and
 \item $\bullet$ \textbf{Decision} $\mathsf{LSN}(k, n, p)$ is the following task: given as input a simple which is either \begin{align}
     \overbrace{\left(\vec C \sim \CC_n, \; \vec E  \vec C\ket{0^{n-k},\vec x}\right)}^{\text{Structured sample}}  \quad \text{ or } \quad \overbrace{\Big(\vec C \sim \CC_n, \; \frac{\vec{I}^{\otimes n}}{2^n}\Big)}^{\text{Unstructured sample}} \quad
 \end{align}
where $\vec C$ is a random $n$-qubit Clifford, $\vec E \sim \mathcal{D}_p^{\otimes n}$ is an $n$-qubit Pauli, $\vec x \sim \Z_2^k$ is a random basis state and $\vec I_n$ is the $n$-qubit identity, the task is to distinguish which is the case.
\end{description}
In addition, we also consider the multi-shot, or multi-sample, variant of the problem, denoted by $\mathsf{LSN}^m(k, n, p)$, which is characterized by an additional parameter $m \in \mathbb{N}$. In the search variant, the input consists of many independent samples of the form
\begin{align}
    \left\{\vec C_i \sim \CC_n, \; \vec E_i  \vec C_i\ket{0^{n-k},\vec x}\right\}_{i=1}^m \,,
\end{align}
where $\vec E_i \sim \mathcal{D}_p^{\otimes n}$ are independently chosen Pauli errors, for $i \in [m]$, and the task is to recover $\vec x \in \Z_2^k$ which appears in each of the samples. In the decision variant, the input either consists of the multi-sample instances above, or $m$ copies of the identity matrix $(\vec I_n/2^n)^{\otimes m}$.
\end{definition}


\begin{remark}[Worst-case vs average-case]
    In the worst case, the decision variant of stabilizer decoding is deciding whether or not a solution exists at all.
    In the average case, the decision variant is instead to distinguish between either a true noisy codeword or a maximally mixed state.
    This distinction reflects a broader convention wherein worst-case formulations consist of solution existence, whereas average-case formulations consist of distinguishing two statistically different scenarios.
    The latter convention is motivated by consistency with the classical $\lpn$ formulation.
\end{remark}

\begin{remark}[Parameters]\label{remark:lsn-parameters}
For purposes of practical error correction, the relevant parameter regime for $\mathsf{LSN}(k, n, p)$ and $\mathsf{LSN}^m(k, n, p)$ with security parameter $n \in \mathbb{N}$, $k < n$, $m=\poly(n)$, and where $p$ is a small constant.
\end{remark}

We say that a quantum algorithm solves Search $\mathsf{LSN}(k, n, p)$ (respectively, $\mathsf{LSN}^m(k, n, p)$) problem if it runs in time $\poly(n)$ and succeeds with probability at least $\frac{1}{2^k}+\frac{1}{\poly(n)}$ over the random choice of Clifford, Pauli error, basis state $\vec x$, and its internal randomness. 
The extra term $2^{-k}$ is motivated by our desire to study the parameter regime where $k$ is extremely small, e.g. a constant, at which point random guessing would trivially succeed with probability $\frac{1}{\poly(n)}$.
For $k = \Omega(\log n)$, the presence of the $\frac{1}{2^k}$ term becomes irrelevant.
In the decision variants of $\mathsf{LSN}(k, n, p)$ (respectfully, $\mathsf{LSN}^m(k, n, p)$), we require an algorithm that can distinguish with advantage $\frac{1}{\poly(n)}$, or equivalently we require an algorithm that determines the sample type correctly with probability $\frac{1}{2}+\frac{1}{\poly(n)}$ when the instance type (structured or unstructured) is chosen uniformly at random. Finally, when $m = \poly(n)$, we frequently also use the notation $\mathsf{LSN}^{\poly}(k, n, p)$ to denote the multi-sample variant of the problem.


When the code distributions are uniformly random encodings, we will use the same abbreviations as for the search version. 
We also define the polynomial-sample version of the decision problem in the same way. 
We note a simple observation that the success probability of a $\lsn$ solver for any distribution of the secret $\vec x \in \Z_2^k$ is uniformly lower bounded by its success probability when $\vec x$ is drawn from the uniform distribution.
Therefore, our choice of letting $\vec x$ be uniformly random in the definition is the hardest variant of the problem in terms of the distribution of the secret.

\begin{lemma}[Secret re-randomization] \label{lemma:secret_rerandomize}
    Suppose $\CA$ is an algorithm which solves Search $\lsn^m(k, n, p)$ with probability $q$.
    Let $\CD$ be an arbitrary distribution over $\Z_2^k$, and define Search $\lsn^m(k, n, p, \CD)$ to be the same task as Search $\lsn^m(k, n, p)$, but with $\vec x$ drawn from $\CD$.
    Then there is an algorithm which solves Search $\lsn^m(k, n, p, \CD)$ to probability $q$, using efficient quantum pre-processing and efficient classical post-processing, with a single call to $\CA$ in between.
    The same statement holds for Decision $\lsn^m(k, n, p)$ with advantage $\delta$.
\end{lemma}
\begin{proof}
    Consider first the search variant.
    Sample a random $\vec y \sim \Z_2^k$ and for each sample $(\vec C_i, \;\vec E_i \vec C_i \ket{0^{n-k}, \vec x})$ (where $\vec x \sim \CD$), apply the logical operator $\vec{\overline{X}}^{(i)}_{\vec y} = \vec C_i \vec X_{\vec y} \vec C_i^\dagger$ to the noisy code state. 
    This produces the states $\vec E_i \vec C_i \ket{0^{n-k}, \vec z}$, where $\vec z = \vec x + \vec y$ is a uniformly random element of $\Z_2^k$.
    Running $\CA$ on the transformed samples, we obtain $\vec z$ with probability $q$.
    We then output $\vec z - \vec y = \vec x$, thereby answering the Search $\lsn^m(k, n, p, \CD)$ problem correctly with probability $q$.

    In the decision variant, the transformation is the same.
    We note that in the unstructured case, the given state is maximally mixed, so applying any unitary operator leaves the state invariant.
    Thus, the advantage remains $\delta$ in the decision variant as well.
\end{proof}

\subsubsection{Learning Stabilizers With Noise, State Variant}
The problem $\lsn$ represents the average-case complexity of decoding a \textit{classical} input after errors. 
It is primarily useful when a certain application for $\lsn$ is easier to analyze when the output of the task is classical.
The more general problem of stabilizer decoding is more naturally defined when the logical state itself is a completely uniformly random state.
We therefore next define $\slsn$, which replaces the logical computational basis state of $\lsn$ with a Haar-random state.

\begin{definition}[Learning Stabilizers with Noise, state variant $\slsn$]\label{def:search-sLSN}
The $\mathsf{state}$ variant of the Learning Stabilizers with Noise problem, denoted by $\slsn^m(k,n,p)$, is characterized by integers $k,n,m \in \mathbb{N}$ and $p \in (0,1)$. Both $p$ and $k$ can vary with $n$. We consider two variants of the problem:
\begin{description}
    \item $\bullet$ \textbf{Search} $\slsn^m(k,n,p)$ is the following task: given as input samples of the form \begin{align}
        \left\{\vec C_i \sim \CC_n, \; \vec E_i \vec C_i\ket{0^{n-k},\psi}\right\}_{i=1}^m \,,
    \end{align}
where $\vec C_i$ are random $n$-qubit Clifford operators, $\vec E_i \sim \mathcal{D}_p^{\otimes n}$ are $n$-qubit Paulis, and $\ket{\psi} \sim \mathfrak{\mu}_k$ is a Haar random $k$-qubit state, the task is to output a quantum state $\rho$ within average fidelity at least $\frac{1}{2^k}+\frac{1}{\poly(n)}$ of $\ket{\psi}$ over the choice of $\ket{\psi} \sim \mu_k$; that is,
\begin{align}
    \underset{\ket{\psi} \sim \mu_k}{\mathbb{E}}[\braket{\psi|\rho|\psi}] \geq \frac{1}{2^k} + \frac{1}{\poly(n)}.
\end{align}
 \item $\bullet$ \textbf{Decision} $\slsn^m(k,n,p)$ is the following task: given as input a sample which is either \begin{align}
     \overbrace{\left\{\vec C_i \sim \CC_n, \; \vec E_i  \vec C_i\ket{0^{n-k},\psi}\right\}_{i=1}^m }^\text{Structured samples} \quad \text{ or } \quad \overbrace{\left\{\vec C_i \sim \CC_n, \;\vec I_n/2^n\right\}_{i=1}^m}^{\text{Unstructured samples}} \quad
 \end{align}
where $\vec C_i$ are random $n$-qubit Cliffords, $\vec E_i \sim \mathcal{D}_p^{\otimes n}$ are $n$-qubit Paulis, $\ket{\psi} \sim \mathfrak{\mu}_k$ is a Haar random $k$-qubit state, and $\vec I_n$ is the $n$-qubit identity, the task is to distinguish which is the case with non-negligible advantage.
\end{description}
\end{definition}
Like with \emph{Search} $\lsn$, a naive algorithm for \emph{Search} $\slsn$ of outputting a maximally mixed state would succeed with probability exactly $\frac{1}{2^k}$. 
A solution therefore constitutes a non-negligible improvement over this baseline. Overall, we define the problem parameters and advantage similar as for the regular $\lsn$ problem (see Remark~\ref{remark:lsn-parameters}).

\begin{remark}[Amplification]
The definitions of a solution to \emph{Search} $\lsn$ and \emph{Search} $\slsn$ both require the weakest solutions, requiring only a chance of success non-negligibly better than random guessing.
A similar definition is conventionally used for $\lpn$.
In the classical setting, however, an algorithm for $\lpn$ which succeeds with non-negligible probability can be efficiently amplified to an algorithm which succeeds with high probability.
This originates from the observation that classically, given a purported solution $\vec x'$ to a sample $(\vec{A}, \vec{A}\vec{x}+\vec{e})$, one can subtract $\vec{A}\vec{x'}$ from $\vec{A}\vec{x}+\vec{e}$ and check if it is low-weight.
With high probability, $\vec{A}\vec{x} - \vec{A}\vec{x}' + \vec{e}$ is low-weight if and only if $\vec x = \vec x'$.
In order to amplify the success probability, we split a $\lpn$ sample into different chunks of rows, effectively obtaining multiple samples $\{(\vec{A}_i, \;\vec{A}_i\vec{x}_{i} + \vec{e}_{i} )\}_{i \in [m]}$ from a single one. 
Then, the secrets can each be re-randomized by sampling uniformly random $\vec y_i$ and then adding $\vec{A}_i\vec{y}_{i}$ to each secret, yielding independent samples $\{(\vec{A}_i, \;\vec{A}_i\vec{x}_{i} + \vec{e}_{i})\}_{i \in [m]}$. 
Finally, applying the algorithm that solves $\lpn$ for each of these independent samples, eventually one of them will be verified correct, and can be used to reconstruct $\vec{x}$. 
In this way, the probability of success may be amplified. 

\begin{remark}[Increasing noise] \label{remark:lsn_larger_p}
    A direct consequence of Lemma~\ref{lemma:depolarizing_convolution} is that for $q > p$, $\lsn^m(k, n, p)$ reduces to  $\lsn^m(k, n, q)$ in both the search and decision variants, since we can always add extra noise to increase the noise parameter.
    The same holds for $\slsn^m$.
\end{remark}

Meanwhile, the success amplification procedure \textit{cannot} be performed for either $\lsn$ or $\slsn$.
One key issue is that $\lsn$ does not admit an analogous recursive structure that enables us to split a single $\lsn$ sample into many smaller $\lsn$ samples.
More importantly, there does not seem to be an efficient way to verify a purported solution to $\lsn$.
Therefore, it is unclear as to whether, quantumly, decoding with high probability is genuinely harder or equivalent to decoding with non-negligible probability.
Since we are interested in proving hardness lower bounds, we take the problem in the easiest form, wherein non-negligible success probability algorithms suffice. 

Moreover, we later show that for Search $\lsn[][\poly]$, succeeding with non-negligible probability and with high probability are polynomially equivalent, by giving an efficient verification protocol that requires multiple samples.
The argument is fundamentally different from our previous classical one, because although we now have access to multiple samples, there is still no way to verify a correct solution to an $\lsn$ problem efficiently. 
\end{remark}
When there are few samples, the decision and search versions of $\slsn$ may be related directly to their $\lsn$ counterparts.

\begin{lemma}[Decision $\slsn =$ Decision $\lsn$] \label{lem:equiv-slsn-lsn}
Let $k,n \in \mathbb{N}$ and $p \in (0,1)$ be parameters. A quantum algorithm solves
Decision $\slsn(k, n, p)$ with success probability $q$ if and only if it solves Decision $\lsn(k, n, p)$ with success probability $q$. 
\end{lemma}
\begin{proof}
It suffices to observe that the distributions of 
\begin{align}
    (\vec C \sim \CC_n, \; \vec E  \vec C\ket{0^{n-k},\psi}) \quad \text{ and } \quad (\vec C \sim \CC_n, \; \vec E  \vec C\ket{0^{n-k},\vec x})
\end{align} 
for a $\vec C$ random $n$-qubit Clifford operator and $n$-qubit Pauli error $\vec E \sim \mathcal{D}_p^{\otimes n}$ are identical when $\ket{\psi}\sim \mu_k$ is a Haar-random state and $\vec x \sim \Z_2^k$. The density matrix of a Haar-random state is precisely the maximally mixed state $\vec I/2^{k}$ on $k$ qubits~\cite{mele2024introduction}---just like for the uniform distribution over $k$ random bits. 
Therefore, the
the two decision tasks are identical.  
\end{proof}

We next prove that with either one or two samples, $\lsn$ reduces directly to $\slsn$ with the same parameters.
This reduction follows from elementary properties of $t$-designs, which we first capture in a lemma.

\begin{lemma}[Relaxation to $t$-designs] \label{lem:designLSN}
Let $k,n \in \mathbb{N}$ and $p \in (0,1)$ be parameters. Then, the Search $\slsn^m(k,n,p)$ problem is equivalent (in the sense that an efficient algorithm for one exists if and only if an efficient algorithm exists for the other) to a variant of the problem in which the encoded Haar state $\ket{\psi} \sim \mu_k$ is instead sampled according to a state $t$-design ensemble $\nu$, whenever $t \geq m+1$.
\end{lemma}
\begin{proof}
We may interpret the input to $\lsn^m$ as the output of a channel $\CE$ which maps
$\psi = \ket{\psi}\bra{\psi}^{\otimes m}$ to 
\begin{equation}
\rho_{\psi}  = \CE(\psi) = \mathop{\mathbb{E}}\limits_{\substack{\vec E_i \sim \CD_p^{\otimes n}\\\vec C_i \sim \CC_n}} \left[ \bigotimes_{i=1}^m \ket{\vec C_i}\bra{\vec C_i} \otimes (\vec E_i\vec C_i)\ket{0^{n-k},\psi}\bra{0^{n-k},\psi}(\vec E_i\vec C_i)^\dag\right] .
\end{equation}
Consider a state $\sigma$ given by $m+1$ copies of the logical state. 
That is, let \begin{align}
    \sigma = \mathbb{E}_{\ket{\psi} \sim \mu_k}\left[(\ket{\psi}\bra{\psi})^{\otimes (m+1)}\right] .
\end{align}
An algorithm $\CD$ solving $\slsn(k, n, p)$, can be written as a map taking $\rho_\psi$ to $\tau_\psi = \CD(\rho_\psi)$.
Then, the trace map $f$ taking $\tau_\psi \otimes \psi$ to $\langle \psi | \tau_\psi | \psi \rangle$ measures the success probability.
In particular, 
\begin{align}
\Pr_{\ket{\psi} \sim \mu_k} \left[\langle \psi | \tau_\psi | \psi \rangle \geq \frac{1}{2^k} + \frac{1}{\poly(n)} \right] & = \underset{\ket{\psi} \sim \mu_k}{\mathbb{E}} [f(\CD(\rho_\psi) \otimes \psi)] = \underset{\ket{\psi} \sim \nu}{\mathbb{E}} [f(\CD(\rho_\psi) \otimes \psi)] \\
& = \Pr_{\ket{\psi} \sim \nu} \left[\langle \psi | \tau_\psi | \psi \rangle \geq \frac{1}{2^k} + \frac{1}{\poly(n)} \right]
\end{align}
by the defining property of a state $t$-design for $t \geq m+1$.
\end{proof}

\begin{proposition}[Search $\lsn$ to $\slsn$ with $\leq 2$ samples] Let $n,k \in \mathbb{N}$ and $p \in (0,1)$ be parameters.
There exists an efficient reduction from Search $\lsn^{m}(k, n, p)$ to Search $\slsn^{m}(k, n, p)$, for both $m=1$ and $m=2$.
\end{proposition}
\begin{proof}
To construct the reduction, we transform a given sample of $\lsn^m(k,n,p)$ and apply the solver for $\slsn^m(k,n,p)$. 
The input for the reduction is given by
\begin{equation}
\left\{\vec C_i \sim \CC_n, \; \vec E_i \vec C_i\ket{0^{n-k},\vec x}\right\}_{i=1}^m.
\end{equation}
Sample a random Clifford operator $\vec B \sim \CC_k$.
The reduction outputs 
\begin{equation}
\left\{\vec C_i (\vec I^{\otimes(n-k)} \otimes \vec B), \; \vec E_i \vec C_i\ket{0^{n-k},\vec x}\right\}_{i=1}^m.
\end{equation}
Note that, for any $i \in [m]$, $\vec C_i(\vec I^{\otimes(n-k)} \otimes \vec B)$ is a uniformly random Clifford because $\vec C_i$ is a uniformly random Clifford. 
Substituting $\vec C_i' \coloneqq \vec C_i (\vec I^{\otimes(n-k)} \otimes \vec B)$, the distribution can be expressed as 
\begin{equation}
\left\{ (\vec C_i', \;\vec E_i\vec C_i'(\vec I^{\otimes(n-k)} \otimes \vec B)^\dag\ket{0^{n-k},\vec x}) = (\vec C_i', \;\vec E_i \vec  C_i'\ket{0^{n-k},\psi})\right\}_{i=1}^m ,
\end{equation}
where $\ket{\psi} = \vec B^\dag\ket{\mathbf{x}}$. 
Since the Cliffords form a unitary 3-design~\cite{webb2015clifford}, $\ket{\psi}$ is now in the transformed state drawn from a 3-design.
The reduction is therefore completed by application of Lemma~\ref{lem:designLSN}.
\end{proof}

It is unclear whether a reduction from Search $\slsn$ to Search $\lsn$ exists in general.
That is, it is possible that decoding on average with a typical logical state is harder than decoding on average with a typical computational basis state.
However, we are able to give a reduction from Search $\slsn(k, n, p)$ to $\lsn(k, n, p)$ if the solver for the latter succeeds with probability non-negligibly greater than $\frac{1}{2}$, rather than the previously assumed $\frac{1}{2^k}$.

\begin{proposition}[Search $\slsn$ to high-success Search $\lsn$] \label{prop:slsn_to_lsn}
Suppose that there is an efficient algorithm $\mathcal{A}$ that solves Search $\lsn[k, n, p]$ with probability $\frac{1}{2}\left(1+\frac{1}{2^{k}}+\frac{1}{\poly(n)}\right)$. 
Then there exists an algorithm $\mathcal{B}$ that solves Search $\slsn[k, n, p]$ with probability $\frac{1}{2^{k}} + \frac{1}{\poly(n)}$. 
\end{proposition}
\begin{proof}
Consider a sample 
\begin{equation}
(\vec{C}, \;\vec{EC}\ket{0^{n-k},\psi})
\end{equation}
of $\slsn[k, n, p]$. 
First, measure the syndrome of the state and produce a random Pauli $\vec{P}$ with the same syndrome. 
We can efficiently sample $\vec P$ since the distribution is specified by a uniformly random solution to a linear system of equations.
We may write $\vec{P}$ in the form $\vec{ES}\overline{\vec{X}}_{\mathbf{y}_1}\overline{\vec{Z}}_{\mathbf{y}_2}$, $\vec S$ is a stabilizer, and $\overline{\vec{X}}_{\mathbf{y}_1}$ and $\overline{\vec{Z}}_{\mathbf{y}_2}$ are random logical $\vec{X}$ and $\vec{Z}$ operators for the code described by $C$. 
Now, run the algorithm $\mathcal{A}$ on 
\begin{equation}
(\vec{C}, \;\vec{PC}\ket{0^n}) = (\vec{C}, \;\vec{EC}\ket{0^{n-k},\mathbf{y}_1}),
\end{equation}
and on 
\begin{equation}
\left(\vec{C}(\vec{I}^{\otimes(n-k)} \otimes \vec{H}^{\otimes k}), \;\vec{P}\left(\vec{C}(\vec{I}^{\otimes(n-k)} \otimes \vec{H}^{\otimes k})\right)\ket{0^n}\right) = \left(\vec{C}', \;\vec{E}\vec{C}'\ket{0^{n-k},\mathbf{y}_2}\right),
\end{equation}
where $\vec{C}' = \vec{C}(\vec{I}^{\otimes(n-k)} \otimes \vec{H}^{\otimes k})$ and $\vec H$ is the Hadamard gate.
The equality above holds because conjugating the $\vec{Z}$ Pauli operator on the last $k$ qubits past the Hadamard operator $(\vec{I}^{\otimes(n-k)} \otimes \vec{H}^{\otimes k})$ converts it to an $\vec{X}$ Pauli.
Both of these states are valid samples of $\lsn$, since $\mathbf{y}_1$ and $\mathbf{y}_2$ are uniformly random, $\vec{C}$ and $\vec{C}'$ are both uniformly random Cliffords, and $\vec{E} \sim \mathcal{D}_p^{\otimes n}$. 

These two applications of $\mathcal{A}$ yield $\mathbf{y}_1'$ and $\mathbf{y}_2'$ as outputs. 
By a union bound, the probability that both outputs are correct is
\begin{equation}
1 - 2\cdot \left(1-\frac{1}{2}\left(1 + \frac{1}{2^k} + \frac{1}{\poly(n)}\right) \right)= \frac{1}{2^k} + \frac{1}{\poly(n)}.
\end{equation}
In this case, $\mathbf{y}_1' = \mathbf{y}_1$ and $\mathbf{y}_2' = \mathbf{y}_2$, and both the $\vec{X}$ and $\vec{Z}$ logical Paulis are known.
Now, computing 
\begin{equation}
\vec{C}^{-1} \overline{\vec{X}}_{\mathbf{y}_1}\overline{\vec{Z}}_{\mathbf{y}_2}\vec{P E C}\ket{0^{n-k},\psi} = \ket{0^{n-k},\psi}
\end{equation}
yields the desired state $\ket{\psi}$.
We have therefore described a procedure $\mathcal{B}$ that solves the problem exactly with probability $\frac{1}{2^k} + \frac{1}{\poly(n)}$. 
The expected fidelity of the output state with the desired state $\ket{\psi}$ is $\frac{1}{2^k} + \frac{1}{\poly(n)}$, as required for a solution for $\slsn$. 
\end{proof}

The above proposition is not a reduction from $\slsn$ to $\lsn$, because by our definition of $\lsn$ a solution need only solve the problem with probability $\frac{1}{2^k} + \frac{1}{\poly(n)}$.
However, one consequence of this proposition is that it proves a strong connection between the relative hardness of $\lsn$ and $\slsn$, and the ability to amplify algorithms for $\lsn$.
Specifically, if an efficient algorithm solving $\lsn$ with probability $\frac{1}{2^k} + \frac{1}{\poly(n)}$ implies an efficient algorithm with probability $1 - \negl(n)$, Proposition \ref{prop:slsn_to_lsn} implies a reduction from $\slsn$ to $\lsn$. 
We may therefore interpret Proposition~\ref{prop:slsn_to_lsn} as a barrier to the amplification for $\lsn$, as this amplification is at least as difficult as reducing $\slsn$ to $\lsn$.
By contrast, an algorithm that solves $\lpn$ with probability $\geq \frac{1}{\poly(n)}$ implies a solution with probability $\geq 1 - \negl(n)$. 

However, we are able to show that Search $\slsn$ and Search $\lsn$ are completely equivalent when there are sufficiently few logical qubits.
\begin{proposition}[Search $\slsn$ to Search $\lsn$ with very few logical qubits]
Let $k = O(\log n)$.
Suppose that there exists an oracle $\CO$ which solves Search $\lsn[k, n, p]$ with probability $\frac{1}{2^k} + \frac{1}{\poly(n)}$.
Then there is a quantum algorithm running in time $\poly(n)$ that 
solves Search $\slsn[k, n, p]$ with probability $\frac{1}{2^k} + \frac{1}{\poly(n)}$ using a single call to $\CO$.
\end{proposition}
\begin{proof}
We begin with a sample $(\vec{C}, \;\vec{EC} \ket{0^{n-k}, \psi})$ of $\slsn[k, n, p]$.
Measure the syndrome of $\vec{EC} \ket{0^{n-k}, \psi}$ and produce a Pauli $\vec{P}$ with the same syndrome.
Hence, $\vec{PE}$ is some logical operator for the code described by $\vec{C}$, i.e. $\vec{PE} =\overline{ \vec{X}}_{\vec{y}_1}\overline{\vec{Z}}_{\vec{y}_2}\vec{S}$ for some $\vec{y}_1, \vec{y}_2 \in \{0, 1\}^k $ and some stabilizer $\vec{S}$. 
By applying $\vec{C}^{-1}$, we obtain
\begin{equation}
\ket{0^{n-k}, \phi} = \vec{C}^{-1}\vec{PEC} \ket{0^{n-k}, \psi}
\end{equation}
By the logical description of $\vec{PE}$, $\ket{\phi} = \vec{X}_{\vec{y}_1}\vec{Z}_{\vec{y}_2}\ket{\psi}$ since $\vec{S}$ does not change $\ket\psi$.
If $\ket{\psi} = \sum_{\vec{x}} c_{\vec{x}} \ket{\vec{x}}$, then a measurement of $\ket{\psi}$ in the computational basis yields $\vec x$ according to the distribution $\Pr[\vec{x}] = |c_{\vec{x}}|^2$.
The operator $\vec{Z}_{\vec{y}_2}\ket{\psi}$ does not change this measurement distribution, while $\vec{X}_{\vec{y}_1}$ shifts it by $\vec{y}_1$.
Hence, the measurement distribution of $\ket{\phi}$ is $\Pr[\vec x + \vec y_1] = |c_{\vec x}|^2$. 

Measuring $\ket{0^{n-k}, \phi}$ in the standard basis, the state $\ket{0^{n-k}, \vec{x} + \vec{y}_1}$ is obtained with probability $|c_{\vec x}|^2$. 
Then, we may again apply $\vec{PC}$ to yield
\begin{align}
\vec{PC}\ket{0^{n-k}, \vec{x} + \vec{y}_1} &= (\vec{PE})\vec{EC}\ket{0^{n-k}, \vec{x} + \vec{y}_1}\\&=\vec{EC}\ket{0^{n-k}, \vec{x}},
\end{align}
where the last equality holds because $\vec{PE}$ is a product of $\vec{S}$ and $\vec{Z}_{\vec{y}_2}$ which act trivially, and $\vec{X}_{\vec{y}_1}$ which shifts the state by $\vec{y}_1$. 
Hence, this procedure produces a sample of $\lsn[k, n, p]$, where the input bitstring follows the distribution $\Pr[\vec x] = |c_{\vec x}|^2$. 
We now use Lemma \ref{lemma:secret_rerandomize} to re-randomize the secret and then apply the $\lsn$ oracle $\CO$, which therefore succeeds with probability $\frac{1}{2^k} + \frac{1}{\poly(n)}$. 
If $\CO$ returns $\vec x' \in \Z_2^k$, return $\ket{\vec x'}$ as the solution for the $\slsn[k, n, p]$ instance. 

We now show that this procedure solves $\slsn[k, n, p]$. 
Let $F(\vec w)$ be the probability that the Search $\lsn$ solver outputs an answer which differs from the correct answer by $\vec w$, and let $p_\psi(\vec z) = |c_{\vec z}|^2$.
Then the state $\ket{\vec x'}$ is outputted with probability 
\begin{equation}
\Pr[\CO = \vec x'] = 
\sum_{\vec{z} \in \Z_2^k} p_\psi(\vec{z})F(\vec{x}' + \vec z) .
\end{equation}
Furthermore, the fidelity of the outputted state $\ket{\vec{x}'}$ with $\ket{\psi}$ is $p_{\psi}(\vec{x}')$, so if we let $\vec{w} = \vec{x} + \vec{x}'$, then the expected fidelity is
\begin{equation}
\underset{\ket{\psi} \sim \mu_k}{\mathbb{E}}
\left[\sum_{\vec x', \vec w \in \Z_2^k} p_\psi(\vec{x}') p_\psi(\vec{x}' + \vec{w}) F(\vec{w}) \right] = \sum_{\vec x', \vec w \in \Z_2^k} \underset{\ket{\psi} \sim \mu_k}{\mathbb{E}}
\left[p_\psi(\vec{x}') p_\psi(\vec{x}' + \vec{w}) \right] F(\vec{w}) ,
\end{equation}
where the equality is due to the fact that $F$ depends only on the $\lsn$ solver and not on $\ket{\psi}$.
For a $k$-qubit Haar-random state $\ket{\psi}$, we claim that $\mathbb{E}_{\ket{\psi} \sim \mu_k}[p(\vec{x})p_\psi (\vec{x+w})] = \frac{2}{2^k(2^k+1)}$ when $\vec{w} = 0$, and $\mathbb{E}_{\ket{\psi} \sim \mu_k}[p_\psi(\vec{x})p_\psi(\vec{x+w})]= \frac{1}{2^k(2^k+1)}$ otherwise.
The first equality can be derived by Lemma~\ref{lemma:haar-beta_distribution}, which states that for all $\vec x$, $p_\psi(\vec x)$ follows a $\text{\textsf{Beta}}(1, 2^k-1)$ distribution.
The second moment of a $\text{\textsf{Beta}}(a, b)$ distribution is known to be $\frac{a(a+1)}{(a+b+1)(a+b)}$.
In our case, therefore, \begin{align}
    \underset{\ket{\psi} \sim \mu_k}{\mathbb{E}}[p_\psi(\vec x)^2] = \frac{1(1+1)}{(1+2^k-1+1)(1+2^k-1)} = \frac{2}{2^k (2^k+1)} .
\end{align}
As for the second equality, $\mathbb{E}_{\ket{\psi} \sim \mu_k}[p_\psi(\vec{x})p_\psi(\vec{x+w})]$ is the same for any $\vec{w}$ by symmetry. 
Expanding \begin{align}
    1 = \underset{\ket{\psi} \sim \mu_k}{\mathbb{E}} \left[\left(\sum_{\vec x \in \Z_2^k}p_\psi(\vec{x})\right)^2 \right] & = \sum_{\vec x, \vec w \in \Z_2^k} \underset{\ket{\psi} \sim \mu_k}{\mathbb{E}}[p_\psi(\vec x) p_\psi(\vec w)] \\
    & = 2^k \frac{2}{2^k(2^k + 1)} + (2^{2k} - 2^k) \underset{\ket{\psi} \sim \mu_k}{\mathbb{E}}[p_\psi(\vec x) p_\psi(\vec w_0)] ,
\end{align}
where $\vec w_0 \in \Z_2^k$ is any fixed vector.
Rearranging gives the second equality. 
Consequently, 
\begin{align}
\underset{\ket{\psi} \sim \mu_k}{\mathbb{E}} \left[\sum_{\vec{x}', \vec{w} \in \Z_2^k} p_\psi(\vec{x}')p_\psi(\vec{x}'+\vec{w})F(\vec{w})\right] &= \frac{1}{2^k(2^k+1)}\sum_{\vec{x}} \left(2F(\vec{0})+ \sum_{\vec{w} \neq 0} F(\vec{w})\right)\\
&= \frac{1}{2^k(2^k+1)}\sum_{\vec{x}}\left( F(\vec{0}) + 1\right)\\
&= \frac{1}{2^k+1}\left( \frac{1}{2^k} + \frac{1}{\poly(n)} + 1\right)\\
&= \frac{1}{2^k} + \frac{1}{(2^k+1)\poly(n)}. 
\end{align}
In the first line, we used the fact that $F$ is a probability distribution.
In the second line, note that by assumption $F(\vec 0) = \frac{1}{2^k} + \frac{1}{\poly(n)}$ since $F(\vec 0)$ is the success probability of the $\lsn$ oracle $\CO$.
The remaining lines are direct algebra.
Since $k = O(\log(n))$, $\frac{1}{2^k}+\frac{1}{(2^k+1)\poly(n)} = \frac{1}{2^k}+\frac{1}{\poly(n)}$, and hence the expected fidelity is sufficiently large to solve $\slsn[k, n, p]$. 
\end{proof}

In the case of $\lpn$, there is a simple reduction that shows that increasing the number of logical bits $k$ makes $\lpn$ harder.
There is likewise a simple proof in the case of $\lsn$, but our proof holds only when $k$ is super-logarithmic, due to the extra term $\frac{1}{2^k}$ in the success probability. 
When $k = O(\log n)$ we do not rule out the possibility that a highly correlated decoder (e.g. one that tends to either get all logical bits correct or get most of them wrong) to do well on a larger $k$ but poorly on a smaller $k$.
\begin{lemma}[Increasing $k$ makes $\lsn$ harder]
Let $n,k \in \mathbb{N}$ and $p \in (0,1)$ be parameters.
For any $k' \geq k$, with $k = \omega(\log n)$, there exists a reduction from Search $\lsn^m(k,n,p)$ to Search $\lsn^m(k',n,p)$. 
\end{lemma}
\begin{proof}
Given an input sample of the form
\begin{equation}
\left\{\vec C_i \sim \CC_n, \; \vec E_i \vec C_i\ket{0^{n-k},\vec x}\right\}_{i=1}^m \, ,
\end{equation}
we sample $\vec y_i' \in \Z_2^{k' - k}$ and define $\vec y_i = \vec 0^{k} \vec y'_i $. We then transform our samples into
\begin{equation}
\left\{\vec C_i \sim \CC_n, \;  \overline{\vec X}_{\vec y_i} \vec E_i  \vec C_i\ket{0^{n-k},\vec x}\right\}_{i=1}^m = \left\{\vec C_i \sim \CC_n, \; \vec E_i  \vec C_i\ket{0^{n-k'},\vec y_i,\vec x}\right\}_{i=1}^m \, ,
\end{equation}
Since global phases may be neglected, this yields a sample for the $\lsn_{k',n,p}^m$ problem. In particular, with non-negligible probability, we can recover $(\vec y_i,\vec x)$, and thus $\mathbf{x}$ with non-negligible probability.
Since $k = \omega(\log n)$, this amounts to $\frac{1}{2^k} + \frac{1}{\poly(n)}$, thereby serving as a solution for $\lsn^m(k,n,p)$.
\end{proof}

To achieve a proof for $\slsn$, however, we must make use of Lemma \ref{lem:designLSN}. 
The proof applies in the case where the number of samples $m \in \set{1, 2}$ and when the distribution of codes is uniform. 
\begin{lemma}[Increasing $k$ makes $\slsn$ harder]
Let $n,k \in \mathbb{N}$ and $p \in (0,1)$ be parameters.
For $k' \geq k$, with $k = \omega(\log n)$, there exists a reduction from Search $\slsn^m(k,n,p)$ to Search $\slsn^m(k',n,p)$, for both $m=1$ and $m=2$. 
\end{lemma}
\begin{proof}
By Lemma~\ref{lem:designLSN}, with one or two samples $\slsn$ is equivalent to the variant where $\ket{\psi}$ is drawn from the uniform distribution $\nu$ over stabilizer states (i.e. a random Clifford applied to $\ket{0^k}$, which is a unitary $3$-design).
We will thus perform the reduction for this variant.  
The input to the reduction is of the form
\begin{align}
    \left\{\vec C_i \sim \CC_n, \; \vec E_i \vec C_i\ket{0^{n-k},\psi}\right\}_{i=1}^m \,,
\end{align}
where $\ket{\psi}$ is a random $k$-qubit stabilizer state. 
We will query a $\slsn(k',n,p)^m$ oracle on
\begin{align}
    \left\{\vec C_i (\vec I^{\otimes(n-k')} \otimes \vec B), \; \vec E_i \vec C_i\ket{0^{n-k},\psi}\right\}_{i=1}^m \,,
\end{align}
where $\vec B \sim \CC_{k'}$ is a random Clifford operator on the last $k'$ qubits. 
Since $\vec C_i \in \CC_n$ is uniformly random, we can set $\vec C_i' \coloneqq \vec C_i(\vec I^{\otimes(n-k')} \otimes \vec B)$, and instead consider 
\begin{align}
    \left\{\vec C_i', \; \vec E_i \vec C_i'(\vec I^{\otimes(n-k')} \otimes \vec B)^\dag \ket{0^{n-k},\psi}\right\}_{i=1}^m \,,
\end{align}
where $\vec C_i' \sim \CC_n$ is uniformly random and where $\vec B^\dag\ket{0^{n-k},\psi}$ results in $\ket{0^{n-k'},\phi }$,
where $\ket{\phi}$ is a random $k'$-qubit stabilizer state. 

Now, applying the solver for $\slsn(k',n,p)^m$, which also solves the stabilizer state variant, we obtain a state $\sigma$ such that $\mathbb{E}_{\ket{\psi}, \vec B}\left[\braket{\phi | \sigma | \phi} \right] = \frac{1}{2^{k'}}+ \frac{1}{\poly(n)}=\frac{1}{\poly(n)}$. 
Equivalently,
\begin{equation}
\mathop{\mathbb{E}}\limits_{\substack{\ket{\psi} \sim \nu\\ \vec B \sim \CC_{k'}}} \left[\braket{0^{k'-k} , \psi| \vec B\sigma \vec B^\dag| 0^{k'-k}, \psi }\right] = \frac{1}{\poly(n)}. 
\end{equation}
Next, prepare the state $\tau = \mathbb{E}_{\vec B}[\vec B\tau \vec B^\dag]$ by applying $\vec B$ to $\sigma$, and output the first $k$ qubits of $\tau'$. 
The resulting state is of the form
\begin{equation}
\tau' = \sum_{\mathbf{x} \in \Z_2^{k'-k}} \braket{\mathbf{x} | \tau | \mathbf{x}}.
\end{equation}
Using the positivity of diagonal matrix elements of a density matrix,
\begin{align}
\underset{\ket{\psi} \sim \mu_k}{\mathbb{E}}
[\braket{\psi | \tau' | \psi}] & = \underset{\ket{\psi} \sim \nu}{\mathbb{E}} [\braket{\psi | \tau' | \psi}] = \sum_{\mathbf{x} \in \Z_2^{k'-k}} \underset{\ket{\psi} \sim \nu}{\mathbb{E}} \left[\braket{\vec x,\psi | \tau |  \mathbf{x},\psi}\right] \\
& \geq \underset{\ket{\psi} \sim \nu}{\mathbb{E}} \left[\braket{0^{k'-k}, \psi| \tau | 0^{k'-k},\psi}\right] = \mathop{\mathbb{E}}\limits_{\substack{\ket{\psi} \sim \nu\\ \vec B \sim \CC_{k'}}} [\braket{0^{k'-k} ,\psi| \vec B\sigma \vec B^\dag| 0^{k'-k},\psi}] \\
&= \frac{1}{\poly(n)}.
\end{align}
Since $k = \omega(\log n)$, $\frac{1}{\poly(n)} = \frac{1}{2^k} + \frac{1}{\poly(n)}$,  as desired. 
\end{proof}

\section{Classical Representation of Quantum Stabilizer Decoding}\label{sec:classical_lsn}
As formulated in Section~\ref{sec:formal_defs}, the definitions of quantum stabilizer and classical decoding appear qualitatively quite different, rendering a reduction between them challenging to conceive.
More precisely, $\lpn$ involves a generator matrix $\vec A$ and a codeword corrupted by additive noise, $\vec{Ax} + \vec e$, whereas $\lsn$ involves a Clifford encoder $\vec C$ and a noisy quantum state $\vec{EC} \ket{0^{n-k}, \vec{x}}$.
In this section, we take a step to bridge this conceptual gap by introducing a problem, which we call the \emph{classical representation} of $\lsn$.
Specifically, this problem satisfies two properties.
First, the classical representation is \emph{equivalent} to the definition of $\lsn$ from Section~\ref{sec:formal_defs} in the sense that given an oracle which solves one problem, there is an efficient quantum algorithm which solves the other using a single call to the oracle.
Second, the classical representation has inputs of (roughly speaking) the form $(\vec A, \vec{Ax} + \vec{e})$, very much like that of $\lpn$.
This representation is therefore a key tool in many of our arguments in this work.

We now define this classical representation of $\lsn$. Note that we use the quantum and classical variant of $\lsn$ interchangeably, and the precise variant is always clear in context.

\begin{definition}[Learning Stabilizers with Noise, classical representation]\label{def:classical-search-LSN}
The classical formulation of $\lsn(k, n, p)$, is characterized by integers $k,n \in \mathbb{N}$ and $p \in (0,1)$. We consider two variants:
\begin{description}
    \item $\bullet$ \textbf{Search} $\lsn(k,n,p)$: this is the task of finding $\vec y \sim \Z_2^k$ given as input a sample of the form
    \begin{align}
        \Big(\begin{bmatrix}
      \vec A \, | \, \vec B  
    \end{bmatrix}, \;
    \begin{bmatrix}
      \vec A \, | \, \vec B  
    \end{bmatrix} \cdot \begin{bmatrix}
        \vec r \\
        \vec y
    \end{bmatrix} + \vec e \Big)
    \end{align}
    where $\vec A \in \Z_2^{2n \times n}$ and $\vec B \in \Z_2^{2n \times k}$ are uniformly random matrices subject to the constraint that the columns of $\vec A$ and also of $\vec B$, respectively, are pairwise symplectically orthogonal (as in Definition~\ref{def:symplectic_form_inner_product}) and that $\begin{bmatrix}
      \vec A \, | \, \vec B  
    \end{bmatrix}$ is full-rank, where $\vec r \sim \Z_2^n$ is random, and where $\mathbf{e}\in \Z_2^{2n}$ is a random depolarizing Pauli error in symplectic form.
    That is, $\vec{e} = \Symp(\vec E)$ for $\vec E \sim \mathcal{D}_p^{\otimes n}$.

   \item $\bullet$ \textbf{Decision} $\lsn(k,n,p)$: this is the task of distinguishing between samples of the form
   \begin{align}
       \Big(\begin{bmatrix}
      \vec A \, | \, \vec B  
    \end{bmatrix}, \;
    \begin{bmatrix}
      \vec A \, | \, \vec B  
    \end{bmatrix} \cdot \begin{bmatrix}
        \vec r \\
        \vec y
    \end{bmatrix} + \vec e \Big) 
    \quad \text{ and } \quad
     \Big(\begin{bmatrix}
      \vec A \, | \, \vec B  
    \end{bmatrix}, 
    \;\vec u \sim \Z_2^{2n}\Big)
   \end{align}
    where all variables are as before, and $\vec u \in \Z_2^{2n}$ is a uniformly random vector.
\end{description}
We refer to $\vec{r}$ as the \emph{junk} and $\vec y$ as the \emph{secret}.
In the \emph{search variant}, we say that a $\poly(n)$ (classical or quantum) algorithm solves the problem if it outputs $\vec{x} \in \Z_2^k$ with probability at least $\frac{1}{2^k} + \frac{1}{\poly(n)}$, whereas in the \emph{decision} variant of the problem we require a distinguishing advantage of at least $\frac{1}{\poly(n)}$.

In addition, we also consider in the classical representation the multi-sample variant of both \emph{Search} and \emph{Decision} $\lsn(k,n,p)$, which we denote by $\lsn^m(k,n,p)$; this problem features an additional parameter $m \in \mathbb{N}$ which captures the number of independent samples. In the case of \emph{Search} $\lsn^m(k,n,p)$, these consist of 
   \begin{align}
        \left\{\Big(\begin{bmatrix}
      \vec A_i \, | \, \vec B_i  
    \end{bmatrix}, 
    \begin{bmatrix}
      \vec A_i \, | \, \vec B_i  
    \end{bmatrix} \cdot \begin{bmatrix}
        \vec r_i \\
        \vec y
    \end{bmatrix} + \vec e_i \Big) \right\}_{i \in [m]}
   \end{align}
which are distributed as before, and where $\vec y \sim \Z_2^k$ is the same in each sample. In the case of \emph{Decision} $\lsn^m(k,n,p)$, these are either $m$ samples of the previous form, or $m$ samples of the form
\begin{align}
     \left\{\Big(\begin{bmatrix}
      \vec A_i \, | \, \vec B_i  
    \end{bmatrix}, 
    \vec u_i \sim \Z_2^{2n}\Big)\right\}_{i \in [m]}.
\end{align}
which are also independently distributed as in the single sample variant.
We emphasize that, just like in the standard formulation of $\lsn$, the secret is the same across samples.
However, the junk \emph{is} re-sampled each time.
\end{definition}
The classical representation of $\lsn$ much more closely resembles $\lpn$, but we highlight the three major ways in which they differ.
After we establish the equivalence of the classical representation, the remainder of the reduction from $\lpn$ to $\lsn$ will be to perform transformations that bridge these differences.
\begin{enumerate}
    \item[(a) ] \textit{Global symplectic condition}: The encoding matrix $\begin{bmatrix}
      \vec A \, | \, \vec B  
    \end{bmatrix}$ requires that columns of $\vec A$ are pairwise symplectically orthogonal, as are the columns of $\vec B$.
    \item[(b) ] \textit{Junk vector}: while the secret in $\lpn$ is the entire vector which the encoding matrix acts on, here it is only a small piece of the vector that the encoding piece acts on. 
    The remaining junk is re-generated across samples.
    \item[(c) ] \textit{Error model}: $\vec e$ is depolarizing instead of Bernoulli, implying that for $i \in [n]$, $e_i$ and $e_{i+n}$ are strongly correlated.
\end{enumerate}
As shown in Fig.~\ref{fig:classical_representation}, $\vec A$ corresponds to the space of operators which fix the logical computational basis state in $\lsn$---that is, the stabilizers and logical $\vec Z$'s of the code.
Since in $\lsn$ there is no distinction between stabilizers and logical $\vec Z$, we make no distinction of them here either, and instead consider a general space spanned by $n$ independent symplectically orthogonal vectors.
$\vec B$ corresponds to the space of logical $\vec X$.
We observe that the presence of the junk is a direct consequence of quantum degeneracy.

\begin{remark}[$\lsn$ vs $\slsn$, revisited]
    While the definition of $\lsn$, i.e. using logical computational basis states, was originally motivated to produce useful applications of the complexity of stabilizer decoding which require classical output, we observe here that the classical representation also crucially relies on the classical nature of the output.
    Since the classical representation is a key element in our reduction from $\lpn$, it sheds new light on the motivation behind the formulation of $\lsn$, namely that it admits a fully classical form which is more closely aligned with the structure of $\lpn$.
\end{remark}

\subsection{Equivalence with the Quantum Notion}
We now show that the two formulations of $\lsn$ are quantumly equivalent. 
Note that while the equivalence is quantum, there is no truly quantum phenomena in the reduction in the sense that the only quantum operations are applying certain Clifford operations to a computational basis state.
The only reason the reduction is required to be quantum is because the standard definition of $\lsn$ involves a quantum state.

When necessary, we will distinguish between the earlier and later definitions for $\lsn$ by referring to the purely classical problems as Decision $\lsn$ (classical) and Search $\lsn$ (classical).
We say that two tasks are \emph{strongly quantum-equivalent} if given an oracle which solves one task with probability $q$, there exists an efficient quantum algorithm which solves the other with probability at least $q$, using a single call to the oracle.

\begin{theorem}[Classical representation is equivalent to $\lsn$]
Search $\lsn[k, n, p][m]$ is strongly quantum-equivalent to Search $\lsn[k, n, p][m]$ (classical), and Decision $\lsn[k, n, p][m]$ is strongly quantum-equivalent to Decision $\lsn[k, n, p][m]$ (classical). 
\end{theorem}
\begin{proof}
We prove this theorem below in four pieces, via Lemmas \ref{lem:classical_reduction_1_search}, 
\ref{lem:classical_reduction_1_dec}, \ref{lem:classical_reduction_2_search}, and \ref{lem:classical_reduction_2_dec}.
\end{proof}

\begin{lemma}[Reduction to classical representation, search] \label{lem:classical_reduction_1_search}
Suppose that there is an oracle $\CO$ which solves Search $\lsn[k, n, p][m]$ (classical) with success probability $q$. Then there is an efficient quantum algorithm which solves Search $\lsn[k, n, p][m]$ with success probability $q$ using a single call to $\CO$.
\end{lemma}
\begin{proof} 
The input for the reduction is 
\begin{equation}
(\vec{C}_i, \;\vec{E}_i\vec{C}_i\ket{0^{n-k}, \mathbf{x}})\text{ for } i\in [m]. 
\end{equation}
We first show that we can efficiently sample Paulis $\vec P_i$ which are uniformly random, subject to the following conditions.
\begin{enumerate}
\item $\vec{P}_i$ has the same syndrome as the error $\vec{E}_i$ for any $i \in [m]$. 
\item $\vec{C}_i^{-1}\vec{P}_i\vec{E}_i\vec{C}_i\ket{0^{n-k}, \mathbf{x}} = \vec{C}_j^{-1}\vec{P}_j\vec{E}_j\vec{C}_j\ket{0^{n-k}, \mathbf{x}}$, for any $i, j \in [m]$. 
\end{enumerate}
If we had only the first condition, the sampling efficiency would hold by the same argument given in the proof of Proposition~\ref{prop:slsn_to_lsn}.
Namely, sampling a uniformly random Pauli with the same syndrome as $\vec E$ is achieved by randomly sampling a solution to a linear system of equations.
Moreover, each $\vec P_i$ differs from $\vec E$ by some product of stabilizers and logical Paulis.
As for the second condition, for any $\vec P_i$ such that
\begin{equation}
\ket{0^{n-k}, \vec y_i} = \vec{C}_i^{-1}\vec{P}_i\vec{E}_i\vec{C}_i\ket{0^{n-k},\mathbf{x}} \neq \vec{C}_1^{-1}\vec{P}_1\vec{E}_1\vec{C}_1\ket{0^{n-k},\mathbf{x}} = \ket{0^{n-k}, \vec y_1},
\end{equation}
we can always map $\vec P_i \mapsto \overline{\vec X}_{\vec y_i - \vec y_1} \vec P_i$, which does not change the syndrome but removes the logical $\vec X$ difference between $\vec P_i$ and $\vec P_1$.
Since the logical $\vec X$ part of $\vec P_1$ is still uniformly random, we have sampled a random element in the set of Paulis satisfying the above two conditions.
Now, let $\mathbf{t}_i = \Symp(\vec{P}_i)$ for each $i$. 
Furthermore, define
\begin{equation}
\vec{A}_i' = 
\begin{bmatrix}
\mathbf{z}_1^{(i)} & \dots & \mathbf{z}_n^{(i)}
\end{bmatrix} \quad \text{and} \quad \vec{B}_i =
\begin{bmatrix}\mathbf{x}_1^{(i)} & \dots & \mathbf{x}_k^{(i)} ,
\end{bmatrix}
\end{equation}
where $\mathbf{z}_j^{(i)} = \Symp(\vec{C}_i \vec{Z}_{j} \vec{C}_i^\dag)$ and $\mathbf{x}_j^{(i)} =  \Symp(\vec{C}_i \vec{X}_{j} \vec{C}_i^\dag)$.
For each $i$, the first $n-k$ of the $\mathbf{z}_j^{(i)}$ are the stabilizers of the code encoded by $\vec C_i$, and the last $k$ are the logical $\vec Z$ operators.
Likewise, the $\mathbf{x}_j^{(i)}$ are the logical $\vec X$ operators of the code $\vec C_i$.

We will show that with
\begin{equation}
\vec{A}_i = \vec{A}_i'\vec{R}_i,
\end{equation}
where $\vec{R}_i \sim \operatorname{GL}_n(\Z_2)$ is a random invertible $n \times n$ matrix,
\begin{equation}
\left(\begin{bmatrix}
      \vec A_i \, | \, \vec B_i  
    \end{bmatrix}, \;\mathbf{t}_i\right) \text{ for } i \in [m]
\end{equation}
is a valid sample of $\lsn[k, n, p][m]$ under Definition \ref{def:classical-search-LSN}. 
In particular, we will show that
\begin{equation}\label{eq:stabilizer_logical_error_sum}
\mathbf{t}_i =\begin{bmatrix}
      \vec A_i \, | \, \vec B_i  
    \end{bmatrix}\begin{bmatrix}\mathbf{r}_i\\\mathbf{y}\end{bmatrix}+ \mathbf{e}_i, 
\end{equation}
where $\mathbf{y}$ and $\mathbf{r}_{i}$ are uniformly random, while $\mathbf{e}_{i}$ is the symplectic representation of a depolarization error. 
Intuitively, the definitions of $\vec A_i'$ and $\vec B_i$ satisfy the symplectic orthogonality condition required of the classical representation, but presently $\vec A_i'$ has an undesirable correlation.
Namely, we know that the $(n-k+j)$th column of $\vec A_i'$ is not symplectically orthogonal with the $j$th column of $\vec B_i$, and is symplectically orthogonal with all other columns of $\vec A_i'$ and $\vec B_i$.
But $\vec A_i$ and $\vec B_i$ are required to be completely random, other then the fact that they each commute among themselves and are linearly disjoint.
This discrepancy arises that we have currently chosen a special basis for $\operatorname{im}(\vec A_i')$ where the first $n-k$ columns are stabilizers and the last $k$ columns are logical $\vec Z$ operators.
By randomly choosing a new basis via application of $\vec R_i$, we obtain a truly random basis choice, thereby erasing the correlation structure between $\vec A_i'$ and $\vec B_i$ and satisfying the distribution requirements of the classical representation.

First, $\vec{B}_i$ has columns that are the symplectic representation of logical $\vec X$ operators, and they are therefore random symplectically orthogonal columns.
Meanwhile, $\operatorname{im}(\vec{A}_i')$ is a uniformly random $n$-dimensional symplectically orthogonal subspace that is linearly disjoint from $\operatorname{im}(\vec{B}_i)$. 
The conjugation of $\vec Z_1, \;\dots, \;\vec Z_n$ by a uniformly random Clifford $\vec C_i$ produces a uniformly random set of $n$ independent commuting Paulis.
This distribution in symplectic form is supported on subspaces $\operatorname{im}(\vec{A}_i')$, which are each linearly disjoint from $\operatorname{im}(\vec{B}_i)$.
Since $\vec{A}_i$ is a random matrix with the same image as $\vec{A}_i'$, and the image is uniformly random, it is indeed a random matrix with symplectically orthogonal columns for which $\begin{bmatrix} \vec{A}_i \;|\; \vec{B}_i \end{bmatrix}$ is full-rank. 
The distributions of $\vec{A}_i$ and $\vec{B}_i$ are therefore correct. 

Now, the Pauli operator $\vec{P}_i$ has the same syndrome as $\vec{E}_i$, implying that $\vec{P}_i\vec{E}_i$ is the product of a uniformly random stabilizer, a uniformly random $\vec{Z}$-logical operator, and a uniformly random $\vec{X}$-logical operator. 
By ensuring that 
\begin{equation}
\vec{C}_i^{-1}\vec{P}_i\vec{E}_i\vec{C}_i\ket{0^{n-k}, \mathbf{x}} = \vec{C}_j^{-1}\vec{P}_j\vec{E}_j\vec{C}_j\ket{0^{n-k}, \mathbf{x}}
\end{equation}
for all $i, j \in [m]$, the logical $\vec{X}$ operators must be the same, and therefore have symplectic representations $\vec{B}_i \mathbf{y}$ for a uniformly random bitstring $\mathbf{y}$.
Meanwhile, the remaining stabilizer and logical $\vec{Z}$ operator yield the term $\vec{A}_i \mathbf{r}_i$, while $\Symp(\vec{E}_i)=\mathbf{e}_i$. 
Thus, everything follows exactly the desired distribution, so Eqn.~(\ref{eq:stabilizer_logical_error_sum}) is a valid sample for the classical representation of $\lsn^m(k, n, p)$.

We now call $\CO$, which returns $\vec y$ with probability $q$.
Observe that $\vec{C}_i^{-1}\vec{P}_i\vec{E}_i\vec{C}_i\ket{0^{n-k}, \mathbf{x}} = \ket{0^{n-k}, (\mathbf{x} + \mathbf{y})}$ for any $i$ because $\vec{P}_i\vec{E}_i$ implements a logical $\vec{X}$ operator which exactly corresponds to adding the $\mathbf{y}$ obtained from the classical representation. 
We recover $\vec x$ from $\vec y$ and $\vec x + \vec y$ to complete the reduction.
\end{proof}

\begin{lemma}[Reduction to classical representation, decision] \label{lem:classical_reduction_1_dec}
Suppose that there is an oracle $\CO$ which solves Decision $\lsn[k, n, p][m]$ (classical) with success probability $q$. Then there is an efficient quantum algorithm which solves Decision $\lsn[k, n, p][m]$ with success probability $q$ using a single call to $\CO$.
\end{lemma}
\begin{proof}
The reduction takes as input the samples 
\begin{equation}
(\vec{C}_i, \rho_i)\text{ for } i\in [m]. 
\end{equation}
In this case, $\rho_i$ is either the structured sample $\vec{E}_i\vec{C}_i\ket{0^{n-k}, \mathbf{x}} \bra{0^{n-k}, \mathbf{x}} \vec C_i^\dag \vec E_i^\dag$ mixed over random $\vec C_i \sim \CC_n$ and $\vec E_i \sim \CD_p^{\otimes n}$, or it is the maximally mixed state. 
We perform the same transformation as in Lemma~\ref{lem:classical_reduction_1_search}, giving classical samples of the form
\begin{equation}
\left(\begin{bmatrix}
      \vec A_i \, | \, \vec B_i  
    \end{bmatrix}, \;\mathbf{t}_i\right) \text{ for } i \in [m],
\end{equation}
where 
\begin{equation}
\mathbf{t}_i = \begin{bmatrix}
      \vec A_i \, | \, \vec B_i  
    \end{bmatrix}\begin{bmatrix}\mathbf{r}_i\\\mathbf{y}\end{bmatrix} + \mathbf{e}_{i}.
\end{equation}
By the same argument as in the proof of Lemma~\ref{lem:classical_reduction_1_search}, a structured sample of Decision $\lsn[k, n, p][m]$ indeed produces a structured sample as in Decision $\lsn[k, n, p][m]$ (classical). 
On the other hand, say that this same reduction was performed on input samples for which $\rho_i$ are maximally mixed states. 
The state $\rho_i$ may then be equivalently considered a random state $\vec{R}_i\vec{C}_i\ket{0^{n}}$, where $\vec{R}_i$ is a random Pauli operator, since this produces a maximally mixed state.
Then the Paulis $\vec{P}_i$ would be chosen so that they have the same syndrome as $\vec{R}_i$, and random so that
\begin{equation}
\vec{C}_i^{-1}\vec{P}_i\vec{R}_i\vec{C}_i\ket{0^n} = \vec{C}_1^{-1}\vec{P}_1\vec{R}_1\vec{C}_1\ket{0^{n}}. 
\end{equation}
The resulting $\vec{P}_i$ are independent and uniformly random, since we start by sampling a uniformly random Pauli with a uniformly random syndrome (that of $\vec R_i$), and then apply a choice of logical operator to enforce the above condition.
However, multiplying a uniformly random Pauli by another Pauli still gives a uniformly random Pauli.
Hence, in symplectic representation, $\vec e_i$ is a uniformly random vector.
When we add it to $\vec A_i \vec r_i + \vec B_i \vec y$, we still get a uniformly random vector.
Then $\vec t_i \sim \Z_2^{2n}$ and $\vec A_i, \vec B_i$ are distributed correctly by the same argument as in the proof of Lemma~\ref{lem:classical_reduction_1_search}.
Therefore, in both the unstructured and structured cases, the distributions are mapped correctly, so we can at this point call the decision oracle $\CO$ for the classical representation and output its answer to complete the reduction.
\end{proof}

We next show the reductions in the reverse direction, i.e. from the classical representation to the standard formulation of $\lsn$.

\begin{lemma}[Reduction from classical representation, search] \label{lem:classical_reduction_2_search}
Suppose that there is an oracle $\CO$ which solves Search $\lsn[k, n, p][m]$ with success probability $q$. Then there is an efficient quantum algorithm which solves Search $\lsn[k, n, p][m]$ (classical) with success probability $q$ using a single call to $\CO$.
\end{lemma}
\begin{proof}
The reduction takes as input the samples
\begin{equation}
\left(\begin{bmatrix}
      \vec A_i \, | \, \vec B_i  
    \end{bmatrix}, \;\mathbf{t}_i\right) \text{ for } i\in [m],
\end{equation}
for $\mathbf{t}_i =  \begin{bmatrix}
      \vec A_i \, | \, \vec B_i  
\end{bmatrix}\begin{bmatrix}\mathbf{r}_i \\ \mathbf{y}\end{bmatrix} + \mathbf{e}_i$. 
We need to generate samples of the form $(\vec C_i, \;\vec E_i \vec C_i \ket{0^{n-k}, \vec x})$ such that $\vec C_i \sim \CC_n$, $\vec E_i \sim \CD_p^{\otimes n}$, and $\vec x \sim \Z_2^k$.
The key step to doing so is the efficient sampling of a uniformly random Clifford $\vec C_i \in \CC_n$ such that \begin{align}
    \operatorname{span}(\Symp(\vec C_i \vec Z_1 \vec C_i^\dagger), \;\dots, \;\Symp(\vec C_i \vec Z_n \vec C_i^\dagger)) = \operatorname{im}(\vec A_i) 
\end{align}
and $\Symp(\vec C_i \vec X_j \vec C_i^\dagger) = (\vec B_i)_j$, where $(\vec B_i)_j$ is the $j$th column of $\vec B_i$, for $j\in [k]$.
We do so by way of observing that up to phases which we do not consider, a Clifford is determined by its action on $\vec Z_1, \;\dots, \;\vec Z_n, \;\vec X_1, \;\dots, \;\vec X_n$.
Presently, $\vec A_i$ encodes how $\vec C_i$ is to act on $\vec Z_1, \;\dots, \;\vec Z_n$, while $\vec B_i$ encodes how $\vec C_i$ is to act on $\vec X_1, \;\dots, \;\vec X_k$.
There are, however, two problems.
First, $\vec A_i$ and $\vec B_i$ are completely independent whereas we would like an explicit basis correspondence where the $j$th logical $\vec Z$ anticommutes exactly with the $j$th logical $\vec X$.
Second, the Clifford is underdetermined because we do not know how $\vec C_i$ is to act on $\vec X_{k+1}, \;\dots, \;\vec X_{n}$; these missing draws are known as \emph{destabilizers}.
We address these problems by way of an iterative sampling procedure.
Specifically, first draw randomly $n-k$ independent basis vectors from $\operatorname{im} (\vec A_i) \subseteq \Z_2^{2n}$ which are symplectically orthogonal to the columns fo $\vec B_i$; these are the stabilizers.
Draw the remaining $k$ independent basis vectors from $\operatorname{im} (\vec A_i) \subseteq \Z_2^{2n}$ such that the $j$th draw, in addition to being independent and symplectically orthogonal to the stabilizers and previous draws, is symplectically orthogonal to all but exactly the $j$th column of $\vec B_i$; this $j$th draw is a logical $\vec Z_j$.
The logical $\vec X$'s are given exactly by the columns of $\vec B_i$.
Finally, we draw a random completion of the $\vec X$ part of the Paulis, i.e. the destabilizers.
To sample the $j$th destabilizers, we sample a random vector in $\Z_2^{2n}$ which is symplectically orthogonal to the logical $\vec Z$'s, logical $\vec X$'s, previously drawn destabilizers, and all but exactly the $j$th stabilizer.
All of these samplings are random solutions to an explicit set of linear equations, so they can be done efficiently.
Now that we have specified the action of the Clifford $\vec C_i$ on an explicit complete basis set of Paulis, we can efficiently compute an exact description of $\vec C_i$, e.g. as a circuit~\cite{rengaswamy2018synthesis}.
We have generated a uniformly random Clifford $\vec C_i$ which upon acting on $\vec Z_1, \;\dots, \;\vec Z_n$ produces in symplectic form $\vec A_i$, and upon acting on $\vec X_1, \;\dots, \;\vec X_k$ produces in symplectic form $\vec B_i$.
To complete the argument that $\vec C_i$ is uniformly random, we must show that the number of such $\vec C_i$ is the same for each choice of $(\vec A_i, \vec B_i)$. 
This fact follows directly from the fact that the Cliffords act transitively on Paulis (i.e. for any pair of Paulis there is a Clifford which maps one to the other by conjugation), so the set of valid Cliffords which produce each $(\vec A_i, \vec B_i)$ are the same.
Since $(\vec A_i, \vec B_i)$ is itself chosen uniformly at random, this argument completes the proof that $\vec C_i$ is a uniformly random Clifford.

Next, let $\vec P_i = \Symp^{-1}(\vec t_i)$ be the Pauli whose symplectic representation is $\vec t_i$.
We prepare the states 
\begin{equation} \label{eq:search_red_2_state}
(\vec{C}_i, \;\vec{P}_i \vec{C}_i\ket{0^n}) \text{ for } i \in [m]. 
\end{equation}
By definition, $\vec P_i$ is the product of three operators, whose symplectic parts are $\vec{A}_i\mathbf{r}_i$, $\vec{B}_i\mathbf{y}$, and $\mathbf{e}_i$. 
The $\vec{A}_i\mathbf{r}_i$ component does nothing to the logical state.
The $\vec B_{i} \mathbf{y}$ component maps the logical state from $\ket{0^k}$ to $\ket{\vec y}$.
Hence, the samples are exactly equal to 
\begin{equation}
(\vec{C}_i, \;\vec{E}_i \vec{C}_i\ket{0^{n-k}, \mathbf{y}}) \text{ for } i \in [m],
\end{equation}
where $\vec C_i \sim \CC_n$, $\vec E_i \sim \CD_p^{\otimes n}$, and $\vec y \sim \Z_2^k$,
as required for Definition \ref{def:search-LSN}. We can then call the oracle for $\lsn^m(k, n, p)$, which will return $\vec y$ with probability $q$, to complete the reduction. 
\end{proof}

\begin{lemma}[Reduction from classical representation, decision] \label{lem:classical_reduction_2_dec}
Suppose that there is an oracle $\CO$ which solves Decision $\lsn[k, n, p][m]$ with success probability $q$. Then there is an efficient quantum algorithm which solves Decision $\lsn[k, n, p][m]$ (classical) with success probability $q$ using a single call to $\CO$.
\end{lemma}
\begin{proof}
The reduction takes as input the samples
\begin{equation}
\left(\begin{bmatrix}
      \vec A_i \, | \, \vec B_i  
    \end{bmatrix}, \;\mathbf{t}_i\right) \text{ for } i\in [m],
\end{equation}
where either
\begin{equation}
\mathbf{t}_i= \begin{bmatrix}
      \vec A_i \, | \, \vec B_i  
    \end{bmatrix}\begin{bmatrix}\mathbf{r}_i \\ \mathbf{y}\end{bmatrix} + \mathbf{e}_i,
\end{equation}
or $\mathbf{t}_i$ is uniformly random.
In the former case, the procedure in the proof of Lemma \ref{lem:classical_reduction_2_search} produces a valid structured sample $(\vec{C}_i, \; \vec{E}_i\vec{C}_i\ket{0^{n-k}, \mathbf{y}})$.
If this same procedure is applied to a uniformly random $\mathbf{t}_i$, then a uniformly random Pauli $\vec{P}_i$ is applied to each state $\vec{C}_i\ket{0^{n}}$, creating maximally mixed states. 
Hence, the resulting samples are $(\vec{C}_i, \rho_i)$, for $\rho_i$ maximally mixed. 
It follows that unstructured samples of $\lsn[k, n, p][m]$ (classical) get mapped to unstructured samples of $\lsn[k, n, p][m]$. 
This completes the proof. 
\end{proof}

\section{Self-Reducibility of Quantum Stabilizer Decoding}\label{sec:self_reducibility}

Generally, the classical landscape of hardness results admit two non-trivial forms of self-reductions. 
The first is a search-to-decision reduction, in which we show that an oracle solving the decision variant of a task can be called as the subroutine of an otherwise efficient algorithm to solve the search variant.
The second is a random self-reduction, wherein we show that an oracle solving a random instance of a task can be used to solve any instance of that task.
By definition, the worst case is always at least as hard as the average case, so the converse notion of an average-to-worst case reduction is vacuous.
Typically, the same can be said about a decision-to-search reduction.
Such a reduction is almost always trivial for classical problems because we can usually efficiently certify the validity of a proposed solution, and thus only accept correct answers.
However, in the quantum formulation it is unclear whether a search oracle is even useful in solving a decision problem, as we will show, because if is no longer easy to verify purported solutions.
Whether or not any of these three non-trivial reductions exist is an open problem for stabilizer decoding.
In this section, we resolve some of these open questions by either giving a reduction or providing evidence that no such reduction exists.

\subsection{Worst-Case Reductions Between Decision and Search} \label{sec:self_reducibility:subsec:worst_s2d}

We give reductions between search and decision. 
In the worst case, we must assume a decision oracle which is capable not only of solving $\qncp(k, n, w)$ for a specific $w$, but also of all $w' \leq w$.
It is an open question as to whether this assumption can be relaxed.

\begin{theorem}[Search-to-decision reduction, $\qncp$] \label{thm:search2dec_sqncp}
    Let $\CT = \set{\vec{P}_1, \;\dots, \;\vec{P}_{n-k}}$ be a stabilizer tableau with code distance $d$ such that $w \leq \lfloor \frac{d-1}{2} \rfloor$.
    Let $\vec{C} \in \CC_n$ be the associated Clifford circuit and let $\CS$ be the stabilizer subgroup generated by $\CT$.
    Suppose there exists an algorithm $\CA(\CT)$ which solves Decision $\qncp(k, n, w')$ for all $w' \leq w$ with probability at least $\frac{2}{3}$. 
    Then there exists an algorithm $\CB(\CT)$ with input $(\vec{C}, \;\ket{\phi} = \vec{E C} \ket{0^{n-k}, \psi })$ such that $\wt(\vec{E}) \leq w$, finds $\widetilde{\vec{E}} \in \CP_n$ such that $\widetilde{\vec{E}} \vec{E} \in \CS$ with probability $\frac{2}{3}$.
\end{theorem}
\begin{proof}
First, we apply the equivalence between $\qncp$ and $\qsdp$ from Theorem~\ref{thm:Decision_sqncp-qsdp_equivalence}.
It is therefore enough to find a search-to-decision reduction for syndrome decoding.
In a slight abuse of notation, we assume instead that $\CA$ and $\CB$ solve the respective syndrome decoding versions.
Here, $\CA$ outputs $\text{\texttt{YES}}$ if there is a solution, and $\text{\texttt{NO}}$ otherwise.
Let $\vec{H} \in \Z_2^{(n-k) \times 2n}$ be the symplectic check matrix for $\CT$, with rows $\mathbf{h}_1^X, \;\dots, \;\mathbf{h}_n^X, \;\mathbf{h}_1^Z, \;\dots, \;\mathbf{h}_n^Z$.
Let $\mathbf{v} \in \Z_2^{n-k}$ be the measured symplectic syndrome and $\mathbf{e} \in \Z_2^n$ be $\Symp(\vec{E})$.
Consider the following algorithm $\CB(\vec{H}, \mathbf{v}, w)$, where \texttt{CONTINUE} means to skip the rest of the loop iteration and immediately begin the next one.
\begin{enumerate}
    \item Using binary search and $\CA$, find $w_* = \wt(\Symp^{-1}(\mathbf{e}))$.
    \item Initialize vectors $\mathbf{\hat{e}}^X\leftarrow 0^n$, $\mathbf{\hat{e}}^Z \leftarrow 0^n$, $\mathbf{\hat{v}} \leftarrow \mathbf{v}$, and an integer $\hat{w} \leftarrow w_*$.
    \item For $i$ from 1 to $n$ do:
    \begin{enumerate}
        \item If $\CA(\vec{H}, \;\mathbf{\hat{v}} + \mathbf{h}_i^Z, \;\hat{w} - 1)= \text{\texttt{YES}}$, set $\mathbf{\hat{v}} \leftarrow \mathbf{\hat{v}} + \mathbf{h}_i^Z$, $\hat{w} \leftarrow \hat{w} - 1$, and $\hat{e}_i^X \leftarrow 1$, then \texttt{CONTINUE}.
        \item If $\CA(\vec{H}, \;\mathbf{\hat{v}} + \mathbf{h}_i^X, \;\hat{w} - 1)= \text{\texttt{YES}}$, set $\mathbf{\hat{v}} \leftarrow \mathbf{\hat{v}} + \mathbf{h}_i^X$, $\hat{w} \leftarrow \hat{w} - 1$, and $\hat{e}_i^Z \leftarrow 1$, then \texttt{CONTINUE}.
        \item If $\CA(\vec{H}, \;\mathbf{\hat{v}} + \mathbf{h}_i^X + \mathbf{h}_i^Z, \;\hat{w} - 1)= \text{\texttt{YES}}$, set $\mathbf{\hat{v}} \leftarrow \mathbf{\hat{v}} + \mathbf{h}_i^X + \mathbf{h}_i^Z$, $\hat{w} \leftarrow \hat{w} - 1$, $\hat{e}_i^X \leftarrow 1$, and $\hat{e}_i^Z \leftarrow 1$.
    \end{enumerate}
    \item Let $\mathbf{\hat{e}} = (\mathbf{\hat{e}}^X, \mathbf{\hat{e}}^Z)$. Output $\mathbf{\hat{e}}$.
\end{enumerate}
All of the outputs of $\mathcal{A}$ are correct with probability at least $\frac{2}{3}$, so by amplification we can make its error probability at most $1/9n^2$.
Then, by union bound, $\CA$ is correct all $3n \lceil \log_2 n \rceil$ times it is called with probability at least $2/3$; we henceforth assume perfect correctness without loss of generality.
The Search variant guarantees that there is a vector $\mathbf{e} \in \Z_2^n$ such that $\wt(\mathbf{e}) \leq w$ and $\vec{H\W} \mathbf{e} = \mathbf{v}$.
Step 1 uses the Decision algorithm $\CA$'s perfect answers to find $w_* = \wt(\Symp^{-1}(\mathbf{e}))$.
This is the value for which $\CA$ outputs $\text{\texttt{YES}}$ but for $w_*+1$ $\CA$ outputs $\texttt{NO}$.

Now, suppose that the error $\vec{E} = \Symp^{-1}(\mathbf{e})$, with $\wt(\vec{E}) \leq w$ and $\vec{H\W} \mathbf{e} = \mathbf{v}$, is supported on the first qubit.
Then $\vec{E}$ must have either $\vec{X}$, $\vec{Y}$, or $\vec{Z}$ on the first qubit, which means that applying one of them must yield a solution with weight $\leq w_* - 1$. 
Hence, one of the three sub-calls in Step 3 is guaranteed to occur.
Suppose instead that $\vec{E}$ is not supported on the first qubit. 
If any one of the sub-calls in Step 3 occurs, then this would imply that increasing the weight of $\vec{E}$ gave a valid solution with weight $\hat{w} - 1 = w_* - 1$. But $w_* \leq w$, and since $w \leq \lfloor \frac{d-1}{2} \rfloor$ there cannot be another solution with weight $\leq w$ giving the same syndrome.
Thus, we have a contradiction.
Consequently, Step 3 in the first iteration has some oracle sub-call if and only if there exists a solution supported on the first qubit, and records the action on the first qubit in symplectic notation. 
Applying this logic iteratively proves correctness.
\end{proof}

A closer inspection of the above proof reveals that in fact the outputted error $\widetilde{\vec{E}}$ satisfies $\wt(\widetilde{\vec{E}}) \leq w$ in addition to $\widetilde{\vec{E}} \vec{E} \in \CS$. 
This is the extra requirement which we imposed in the formulation of $\errqncp$ in Definition~\ref{def:worst_case_error_QNCP}. 
Therefore, $\errqncp$ reduces to Decision $\qncp$.
At the same time, there is a simple decision-to-search reduction for $\errqncp$, because the validity of a claimed $\errqncp$ solution can be easily checked.
\begin{claim}[Decision-to-search reduction, $\errqncp$] \label{claim:dec2search_wSQNCP}
    Let $\CT, \CS, \vec{C}, w, d$ be as in Theorem~\ref{thm:search2dec_sqncp}. 
    Suppose there is an algorithm $\CB(\CT)$ which solves Search $\errqncp(k, n, w)$ on the code with Clifford $\vec{C}$, with probability at least $\frac{2}{3}$.
    Then there is an algorithm $\mathcal{A}$ which solves Decision $\qncp$, i.e. determines whether the instance of Decision $\qncp$ is a \texttt{YES} (there is a solution) instance or \texttt{NO} (no solution) instance, with probability at least $\frac{2}{3}$.
\end{claim}
\begin{proof}
    $\CA$ converts to syndrome decoding by measuring the syndrome and using the symplectic representation to obtain check matrix $\vec{H}$ and symplectic syndrome $\mathbf{v}$. 
    Then, $\CA$ runs Search $\errqncp(k, n, w)$ on the given instance, which will return some $\widetilde{\vec{E}}$.
    We compute $\mathbf{\widetilde{e}} = \Symp(\widetilde{\vec{E}})$ and check whether $\vec{H \W} \mathbf{\widetilde{e}} = \mathbf{v}$ and $\wt(\widetilde{\vec{E}}) \leq w$. If so, output \texttt{YES}; otherwise output \texttt{NO}.
    If the instance is \texttt{YES}, then $\widetilde{\vec{E}}$ will have the correct syndrome and be low-weight with probability at least $\frac{2}{3}$, so $\CA$ will output \texttt{YES} with at least this probability.
    If the instance is \texttt{NO}, then any $\widetilde{\vec{E}}$ which satisfies the syndrome must satisfy $\wt(\widetilde{\vec{E}}) > w$, in which case $\CA$ will certainly say \texttt{NO}.
    Thus, in either case $\CA$ succeeds with probability at least $\frac{2}{3}$.
\end{proof}
This reduction would not go through for Search $\qncp$ or Search $\recqncp$, because there does not seem to be an efficient technique to verify the validity of solutions proposed by the search oracle.
It is an open question as to the existence or non-existence of a decision-to-search reduction using Search $\qncp$ or $\recqncp$.

While tempting to conclude, it is not true that the above results imply an equivalence between Decision $\qncp$ and Search $\errqncp$.
This is because the search-to-decision reduction requires a decision oracle which is capable of distinguishing whether the error is of weight above or below $w'$ for all $w' \leq w$.
Such an oracle is exceptionally strong and is equivalent to solving $\qncp(k, n, w')$ for all $w' \leq w$.
The search-to-decision reduction can be interpreted as revealing just how strong this oracle is, as it enables solving the short vector problem encoded within $\errqncp$.

As we next explore, the relation between search and decision exhibits slightly different inherently quantum phenomena in the average case. 
More precisely, the promise that instances are drawn from the uniform distribution makes it possible to construct reductions with a completely different caveat---using multiple samples---that equate average-case search and decision variants.

\subsection{Average-Case Search-Decision Equivalence}
In this section, we demonstrate a search-decision equivalence for $\lsn[][\poly]$. 
This reduction relies crucially on the ability to obtain multiple samples; it is an open question as to whether the equivalence holds with a single sample.

For context, we recall the trivial decision-to-search procedure in classical decoding.
For $\lpn$, one may consider a sample of $\lpn$, $(\vec{A}, \mathbf{y})$, and submit it to Search $\lpn$. 
If $\mathbf{y} = \vec{A}\mathbf{x} + \mathbf{e}$, then with non-negligible probability, the search algorithm will return $\vec{\widetilde{x}}$ such that $\vec{\widetilde{x}} = \vec{x}$. 
The decision algorithm, after verifying $\mathbf{x}$ by checking that $\vec y - \vec{A\widetilde{x}}$ is low-weight, can accept this sample.
However, if $\mathbf{y}$ is uniformly random, a solution $\mathbf{x}$ exists with negligible probability, so the decision algorithm will accept with negligible probability as well.

For $\lsn$, there is no longer an efficient algorithm to verify that a potential solution $\widetilde{\mathbf{x}}$ is indeed the solution. 
A verification protocol would have to observe that $\vec{C}\ket{\widetilde{\mathbf{x}}}$ and $\vec{EC}\ket{\vec{x}}$ differ by a low-weight error. 
However, $\vec{E}$ is not the only operator that could have taken $\vec{C}\ket{\vec{x}}$ to $\vec{EC}\ket{\vec{x}}$. 
Indeed, there are exponentially many such operators, and determining whether one of them is low-weight is precisely a short vector problem. 
Quantum degeneracy is therefore once again responsible for this qualitative difference between the classical and quantum settings. 
Nevertheless, we show that given more samples of the decision problem, there is still a reduction to the search problem. 

\begin{theorem}[Decision-to-search reduction, $\lsn$] \label{thm:decision_to_search_lsn}
Let $\CO$ be an oracle which solves Search $\lsn[k, n, p][m]$ with probability at least $\frac{1}{2^k} + \frac{1}{\poly(n)}$.
Then there exists an efficient quantum algorithm which solves Decision $\lsn[k, n, p][2m]$ with non-negligible advantage, using two calls to $\CO$.
\end{theorem}
\begin{proof}
We begin with samples $(\vec{C}_i, \rho_i)$ for $i\in[2m]$, where $\rho_i$ are either independent maximally mixed states, or are code states $\mathbb{E}_{\vec{C}_i, \vec{E}_i, \vec{x}}\left[(\vec{E}_i\vec{C}_i)\ket{0^{n-k}, \vec{x}}\bra{0^{n-k}, \vec{x}}(\vec{E}_i\vec{C}_i)^\dag\right]$, where $\vec{E}_i \sim \mathcal{D}_p^{\otimes n}$ are independent depolarization errors, $\vec{x} \sim \Z_2^k$, and $\vec{C}_i$ are independent random Clifford operators. 
The second option for $\rho_i$ may equivalently be considered as the state $\vec{E}_i\vec{C}_i\ket{0^{n-k}, \vec{x}}$, where $\vec{E}_i, \vec{C}_i, \vec{x}$ are all stochastically chosen from the same distributions. 

Split up the $2m$ samples into two halves. The first half of samples, $(\vec{C}_i, \rho_i)$ for $i \in[m]$, we run the search oracle on our samples as-is, and receive a proposed solution $\vec z_1$. 
For the second half of the samples, randomly choose a bitstring $\vec{y} \in \Z_2^k$, and apply $\vec{\overline{X}}_{\vec{y}}^{(i)}$ to each, i.e. apply the logical operation for $\vec{X}_{\vec y}$ on the $i$th code. 
If the $\rho_i$ are maximally mixed, applying $\vec{\overline{X}}_{\vec{y}}^{(i)}$ does nothing. 
If they are code states, then the states transform to 
\begin{equation}
\vec{\overline{X}}_{\vec{y}}^{(i)} \vec{E}_i\vec{C}_i\ket{0^{n-k}, \vec{x}} = \vec{E}_i \vec{C}_i \ket{ 0^{n-k}, \vec{x}+ \vec{y}} .
\end{equation}
We then run the search oracle on these transformed samples, obtaining a new output $\vec{z}_2$.
We accept if $\vec{z}_1 + \vec{z_2} = \vec{y}$, and reject otherwise.

Suppose first that all $2m$ of our states were genuine noisy code states. 
Then, $\vec{z}_1$ and $\vec{z}_2$ are independent, since they are the outcomes of the search oracle on completely independent inputs—the circuits and errors were independently sampled, and $\vec{x} + \vec{y}$ for uniformly random $\vec{y}$ is independent from $\vec{x}$.
Hence,
\begin{align}
    \Pr[\vec{z}_1 + \vec z_2 = \vec y] & = \sum_{\vec v \in \Z_2^k} \Pr[\vec z_1 = \vec v, \;\vec z_2 = \vec v + \vec y] \\
    & = \sum_{\vec v \in \Z_2^k} \Pr[\vec z_1 = \vec v] \cdot \Pr[\vec z_2 = \vec v + \vec y] .
\end{align}
Let $F(\vec w)$ be the probability that the Search $\lsn^m$ solver outputs an answer which differs from the correct answer by $\vec w$.
By Lemma~\ref{lemma:secret_rerandomize}, $F$ has no dependence on the input itself.
Then \begin{align}
    \Pr[\vec z_1 + \vec z_2 = \vec y] = \sum_{\vec v \in \Z_2^k} F(\vec v + \vec x) \cdot F(\vec v + \vec x) = \sum_{\vec v \in \Z_2^k} F(\vec v)^2 .
\end{align}
We may write $F(\vec v) = \frac{1}{2^k} + \delta(\vec v)$.
By normalization, $\sum_{\vec v \in \Z_2^k} \delta(\vec v) = 0$; by assumption, $\delta(0) = \frac{1}{\poly(n)}$.
Thus, \begin{align}
    \sum_{\vec v \in \Z_2^k} F(\vec v)^2 = \sum_{\vec v \in \Z_2^k} \frac{1}{2^{2k}} + \frac{1}{2^k} \delta(\vec v) + \delta(\vec v)^2 = \frac{1}{2^k} + \sum_{\vec v \in \Z_2^k} \delta(\vec v)^2 \geq \frac{1}{2^k} + \frac{1}{\poly(n)} .
\end{align}

Suppose instead that all $2m$ of the states were independent maximally mixed states.
Then the application of $\overline{\vec X}_{\vec y}$ has no effect, as the state remains maximally mixed.
Therefore, the distribution of $\vec{z}_1 + \vec z_2$ is independent from $\vec{y}$. 
But $\vec y$ is uniformly random, so the probability that $\vec y = \vec z_1 + \vec z_2$ is exactly $1/2^k$.
There is therefore a negligible chance that the distinguisher will accept in this case. The two cases differ by a non-negligible term. Therefore, the distinguisher constructed has non-negligible advantage, which completes the reduction.
\end{proof}
This proof takes independent samples and shifts half of them by a random string $\vec{y}$ known only to the distinguisher. 
If the samples were not structured, then the shift by $\vec{y}$ is information-theoretically impossible to recover from the samples that were submitted to the search oracle. 
Meanwhile, if they were structured, then a search oracle is capable of returning the two samples correctly, and the shift between them yields $\vec{y}$. 
The above argument also works for $\lpn$, but unlike the standard technique with verification, it generalizes to the quantum setting as well.

For the search to decision reduction, our argument is inspired by the analogous search to decision reduction for $\lpn$.
It will require a lemma that shows that low-weight errors are capable of scrambling one qubit. 
\begin{lemma}[Low-weight errors scramble at least one qubit completely] \label{lem:error_randomize_qubit}
Let $\rho = \ket{y}\bra{y} \otimes \frac{\vec{I}^{\otimes(n-1)}}{2^{n-1}}$ be an $n$ qubit state for some $y \in \Z_2$, let $\vec{C} \sim \CC_n$ be a random Clifford, and let $\vec{E} \sim \mathcal{D}_p^{\otimes n}$ be an error. 
Then for any $p = \omega(\frac{\log n}{n})$, $\ket{\vec{C}}\bra{\vec{C}} \otimes (\vec{EC})\rho(\vec{EC})^\dag$ is negligibly close in trace distance to $\ket{\vec{C}}\bra{\vec{C}} \otimes \frac{\vec{I}^{\otimes n}}{2^n}$.
\end{lemma}

\begin{proof}
The state $(\vec{EC})\rho(\vec{EC})^\dag$ may be rewritten as $(\vec{CE}')\rho(\vec{CE}')^\dag$, where $\vec{E}' = \vec{C}^\dag \vec{EC}$. 
The action of $\vec{E}'$ on the state depends on the Pauli it applies to the first qubit.
If it is an $\vec{X}$ or $\vec{XZ}$ Pauli, then it flips the bit $y$, while if it is $\vec{I}$ or $\vec{Z}$ it does nothing.
$\vec{E}'$ may always be expressed as $\vec{X}_1\vec{P}$ or $\vec{P}$, where $\vec{P}$ is a Pauli on the last $n-1$ qubits and potentially a $\vec{Z}_1$. 
$\vec{P}$ fixes $\rho$, since the maximally mixed state will not be altered by the Pauli on the last $n-1$ qubits, and $\vec{Z}_1$ fixes $\ket{y}\bra{y}$. 
If we can show that, knowing $\vec C$, $\vec{E}'$ can be expressed as $\vec{P}$ with probability $\frac{1}{2}+\epsilon$ and as $\vec{X}_1\vec{P}$ with probability $\frac{1}{2}-\epsilon$, then
\begin{align} \label{eq:lsn_s2d_goal}
\underset{\vec E \sim \CD_p^{\otimes n}}{\mathbb{E}}[\vec{E}'\rho(\vec{E}')^\dag \,|\, \vec C] & = \left(\frac{1}{2} + \e\right) \frac{\vec{I}^{\otimes(n-1)}}{2^{n-1}} \otimes \ketbra{y}{y} + \left(\frac{1}{2} - \e\right) \frac{\vec{I}^{\otimes(n-1)}}{2^{n-1}} \otimes \ketbra{1-y}{1-y} \\
& = \frac{\vec{I}^{\otimes n}}{2^{n}} + \epsilon(\ket{y}\bra{y} - \ket{1- y}\bra{1- y}) \otimes \frac{\vec{I}^{\otimes(n-1)}}{2^{n-1}} .
\end{align}
We here allow $\e$ to be negative, and instead aim to show that $|\e| = \negl(n)$ for all but a negligible fraction of $\vec C$.
If so, then the trace distance of $\ketbra{\vec C}{\vec C} \otimes \vec{E}'\rho(\vec{E}')^\dag$ with the maximally mixed state is also negligible. 
Applying the unitary $\vec{C}$ and the fact that the trace distance is preserved by unitary transformations, we would obtain the same statement for $(\vec{CE}')\rho(\vec{CE}')^\dag$, completing the proof.

To bound $|\e|$, we apply the Bernoulli representation of depolarizing noise as from Corollary~\ref{corollary:bernoulli-depolarizing_duality} in order to decompose the error into a sum of independent Bernoulli errors.
In the symplectic representation, $\vec{C}^\dag$ is a $2n \times 2n$ matrix $\vec{S}$ that preserves the symplectic inner product~\cite{rengaswamy2018synthesis}, such that $\vec S$ maps $\mathbf{e} = \Symp(\vec E)$ to $\mathbf{e}' = \Symp(\vec E')$.
The bit $e'_1$ determines whether $\vec E'$ applies a $\vec X$ Pauli to the first qubit, and is the sole variable determining $\e$.
Furthermore, $e'_1 = \mathbf{r}^\intercal \mathbf{e}$, where $\mathbf{r}^\intercal$ is the first row of $\vec{S}$.
For a random $\vec C$, $\mathbf{r}^\intercal$ follows the uniform distribution over $\Z_2^{2n}$ because the marginal distribution of the first row of a random matrix that preserves the symplectic inner product is uniformly random. 

By Corollary~\ref{corollary:bernoulli-depolarizing_duality}, we may write $\vec e = (\vec e_{\vec X}, \vec e_{\vec Z}) + (\vec e_{\vec Y}, \vec e_{\vec Y})$, where $e_{\vec X}, e_{\vec Y}, e_{\vec Z}$ are i.i.d. $\Ber(q)$ and $p = 3[q^2(1-q) + q(1-q)^2]$.
Thus, $e_1' = \vec r^\intercal \widetilde{\vec e}_{\vec X} + \vec r^\intercal \widetilde{\vec e}_{\vec Z} + \vec r^\intercal (\vec e_{\vec Y}, \vec e_{\vec Y})$, where $\widetilde{\vec e}_{\vec X} = (\vec e_{\vec X}, 0)$ and $\widetilde{\vec e}_{\vec Z} = (0, \vec e_{\vec Z})$.
Each of these three terms are independent.
Consider the distribution of the first term, $\vec r^\intercal \widetilde{\vec e}_{\vec X}$.
Let $m$ be the weight of the first $n$ entries of $\vec r$ and assume that $m \geq \frac{n}3$.
Then $\vec r^\intercal \widetilde{\vec e}_{\vec X} \sim \Bin(m, q)$.
Using the fact that the probability that a Binomial random variable is even is $\frac{1}{2} + \frac{1}{2} (1 - 2q)^m$ and the exponential bound $1 + x \leq e^x$, $\forall x \in \R$, \begin{align}
    \Pr[\vec r^\intercal \widetilde{\vec e}_{\vec X} = 0] & = \frac{1}{2} + \frac{1}{2} (1 - 2q)^m \leq \frac{1}{2} + \frac{1}{2} (1 - 2q)^{n/3} \\
    & \leq \frac{1}{2} + \frac{1}{2} \exp{- \frac{2}{3} q n} .
\end{align}
So long as $q = \omega(\frac{\log n}{n})$, the above probability is $\frac{1}{2} + \eta$ where $\eta = \negl(n)$.
Note that $\frac{3}{2} q \leq 3[q^2(1-q) + q(1-q)^2] \leq 3q$ for all $q \in [0, \frac{1}{2}]$, so $q = \omega(\frac{\log n}{n})$ if and only if $p = \omega(\frac{\log n}{n})$.

Now, given $b \sim \Ber(\frac{1}{2} + \d)$, for any independent random variable $b' \sim \Ber(p')$, we can check by direct computation that $b + b' \sim \Ber(p'')$ where $p'' \in [\frac{1}{2} - \delta, \;\frac{1}{2} + \delta]$.
Thus, if $m \geq \frac{n}3$, $e_1' = \vec r^\intercal \vec e \sim \Ber(\e)$ where $|\e| \leq \frac{1}{2} \exp{\frac{- 2qn}3} = \negl(n)$.
Therefore, we have nearly accomplished our goal set out in Eqn.~(\ref{eq:lsn_s2d_goal}) for any $\vec C$ such that $m \geq 3$.
Since $\vec r$ is uniformly random over $\Z_2^{2n}$, the Chernoff bound implies that $m \geq \frac{n}3$ with probability $1 - \exp{-\W(n)}$, which fully accomplishes our goal.
\end{proof}

\begin{theorem}[Search-to-decision reduction, $\lsn$] \label{thm:search_to_decision_lsn}
Let $p = \omega(\frac{\log n}{n})$.
Suppose that there exists an oracle $\CO$ which solves Decision $\lsn[k, n, p][\poly][]$ with non-negligible advantage.
Then there exists an efficient quantum algorithm which solves Search $\lsn[k, n, p][\poly][]$ with probability $\Omega(1)$, using $\poly(n)$ calls to $\CO$. 
\end{theorem}
\begin{proof}
The input for the reduction is polynomially many samples 
\begin{equation}
(\vec{C}_i, \;\vec{E}_i\vec{C}_i\ket{0^{n-k}, \mathbf{x}}).
\end{equation}
We will construct a procedure to learn the $j$th bit of $\mathbf{x}$ with non-negligible probability.
Repeating it polynomially many times with independent samples amplifies that probability to $1 - o(\frac1k)$. 
Then, applying this procedure to each bit of $\mathbf{x}$, with independent samples of the search problem, we will obtain $\mathbf{x}$ with a probability of at least a constant.

Without loss of generality, we will describe the procedure to obtain the first bit $x_1$. 
Sample a random $\vec y \sim \Z_2^k$ and let $ \overline{\vec X}_{\vec y}$ be a Pauli $\vec{X}$ operator acting on the last $k$ qubits according to $\vec y$. 
For $i \in[m]$ samples necessary to call the decision oracle once, sample a random Pauli on all qubits except the $(n-k+1)$th qubit (i.e. the first qubit corresponding to $\vec{x}$), and let their tensor product be $\vec{P}_i$. 
Define $\vec{P}'_i$ as the Pauli so that $\vec{P}'_i\vec{C}_i = \vec{C}_i\vec{P}_i$.
Prepare the samples
\begin{equation}
(\vec{C}_i(C_{n-k+1}\vec{P}_i) \overline{\vec X}_{\vec y} , \;\vec{P}_i'\vec{E}_i\vec{C}_i\ket{0^{n-k}, \vec{x}}). 
\end{equation}
Here, $C_{n-k+1} \vec U$ denotes a controlled operation $\vec U$ with the $(n-k+1)$th qubit as the control.
Since $\vec{C}_i$ is uniformly random, we may write $\vec{C}_i' \coloneqq \vec{C}_i(C_{n-k+1}\vec{P}_i) \overline{\vec X}_{\vec y}$, and rewrite the samples as
\begin{equation}
(\vec{C}_i', \;\vec{P}_i'\vec{E}_i\vec{C}_i' \overline{\vec X}_{\vec y} (C_{n-k+1}\vec{P}_i)\ket{0^{n-k}, \vec{x}}) .
\end{equation}
Suppose first that $x_1 = 1$. Then the distribution is \begin{align}
(\vec{C}_i', \;\vec{P}'_i\vec{E}_i\vec{C}_i' \overline{\vec X}_{\vec y} \vec{P}_i\ket{0^{n-k}, \mathbf{x}}) & = (\vec{C}_i', \;\vec{E}_i \vec{C}_i' \overline{\vec X}_{\vec y} \ket{0^{n-k}, \mathbf{x}}) \\
& = (\vec{C}_i', \;\vec{E}_i \vec{C}_i' \ket{0^{n-k}, \mathbf{x} + \vec y}) .
\end{align}
These are precisely structured $\lsn$ samples.
Suppose instead that $x_1 = 0$. Then the controlled operation is equivalently the identity, so the distribution of a given sample is 
\begin{equation}
(\vec{C}_i', \vec{E}_i\vec{C}_i'\vec{P}_i\ket{0^{n-k}, \mathbf{x} + \mathbf{y} }).
\end{equation}
The application of the random $\vec{P}_i$ fully randomizes the state and yields the maximally mixed state for each sample on every qubit except the $(n-k+1)$th, which remains $x_1 + y_1$ with certainty. 
However, by Lemma \ref{lem:error_randomize_qubit}, the error $\vec{E}_i$ scrambles this final bit so that the state is negligibly close to the maximally mixed state, since $p = \omega\left( \frac{\log(n)}{n}\right)$. 

We conclude that the two options, $x_1 = 0$ and $x_1 = 1$, are precisely the types of samples that the decision problem can distinguish with non-negligible advantage.
This probability may be amplified by repeated trials (polynomially many times) such that the failure probability is $o(\frac1k)$, so that, iterating over each $x_j$, the correct secret $\vec x$ is outputted with probability $\W(1)$.
\end{proof}

Using the exact same proof strategy, we may obtain a reduction from Search $\lsn^m$ to Decision $\lsn^m$ for any \emph{fixed} $m$, and in particular for the case of greatest interest, $m=1$. 
However, the reduction holds only when there are very few logical qubits.

\begin{corollary}[Search-to-decision reduction, $\lsn$ with very few logical qubits]
\label{cor:search_to_decision_lsn_m}
Let $p = \omega(\frac{\log n}{n})$ and $k= O(\log n)$.
Suppose that there exists an oracle $\CO$ which solves Decision $\lsn[k, n, p][m][]$ with non-negligible advantage.
Then there exists an efficient quantum algorithm which solves Search $\lsn[k, n, p][m][]$ with probability $\frac{1}{2^k} + \frac{1}{\poly(n)}$, using one call to $\CO$. 
\end{corollary}

\begin{proof}
Using one step of the approach from Theorem $\ref{thm:search_to_decision_lsn}$, we may obtain the first bit of $\vec{x}$ with probability $\frac{1}{2} + \frac{1}{\poly(n)}$. Guessing the remaining bits, 
we obtain $\vec{x}$ with probability $\geq \frac{1}{2^{k-1}}\left(\frac{1}{2}+\frac{1}{\poly(n)}\right) = \frac{1}{2^k} + \frac{1}{\poly(n)}$, where the equality holds because $k = O(\log(n))$.
\end{proof}

\subsection{Barriers Against Random Self-Reductions}
One of the most direct ways to show that a problem is hard on average is to first show that it is hard in the worst case, and then construct a \emph{random self-reduction}.
This reduction takes the form of an efficient algorithm that scrambles an arbitrary instance of the problem into a random one, such that (a) an average-case solver correctly solves the transformed task with high probability and (b) a solution of the scrambled instance can be transformed efficiently into a solution of the original given instance.
In general, random self-reductions for decoding tasks have proved difficult even in the classical case.
Until recently, no random self-reduction into $\lpn$ was known for any parameter regime.
The work of \cite{brakerski2019worst} showed that for a relatively easy regime of worst-case decoding (in the sense that there exists a quasipolynomial time decoder in this regime) a random self-reduction to $\lpn$ can be constructed. 
However, the reduction only holds for a restricted subset of classical linear codes.
In the setting of quantum stabilizer codes, it is an open question as to whether any semblance of a random self-reduction can be achieved.
To that end, we first prove that no random self-reduction from an arbitrary worst-case $\qncp$ instance into $\lsn$ can exist.
The culprit preventing the existence of such a reduction is exchange symmetry between the code and the error---both are Pauli strings which are randomly scrambled during the reduction, but one must be maximally random while the other cannot be maximally random (less we lose decodability), which is impossible.
However, we will show that an analogous no-go theorem also holds classically. \cite{brakerski2019worst} circumvent this barrier by finding a clever restriction of possible worst-case instances---codes whose codewords all have nearly the same weight---wherein the reduction is possible at least in some non-trivial parameter regime.
This reduction proceeds by choosing a sufficiently local/sparse scrambling process, which carefully balances scrambling a code with the above restricted structure, without completely scrambling the error.
Quantumly, there are two natural analogs of a local scrambling model: a random sparse Clifford operator, and a random local Clifford circuit.
We prove that these two natural analogs of a local scrambling process cannot produce a random self-reduction for \emph{any} choice of restriction on worst-case instances.
Since locality/sparsity of a scrambling process is the only simple way to control the weight of an arbitrary error---which when too large destroys decodability even information theoretically---this result provides evidence that no random self-reduction exists at all, for any restriction.
We formally conjecture the nonexistence of random self-reductions into $\lsn$ (and its variants) at the end of this section.

We begin with a formal definition of a random self-reduction in the stabilizer code setting.
\begin{definition}[Random self-reduction distribution] \label{def:random_self_reduction}
    Fix parameters $n, k$.
    Let $\CD_{n, k, p} = \Unif(\Stab(n, k)) \otimes \CD_{p}^{\otimes n}$, where $\Stab(n, k)$ is set of $\llbracket n, k \rrbracket$ stabilizer codes and $\CD_{p}^{\otimes n}$ is the $n$-fold depolarizing channel with probability $p$.
    A $(\CV, \e, w, p)$ \emph{random self-reduction distribution} is an efficiently sampleable distribution $\CR_n$ over Cliffords $\CC_n$ which satisfies \begin{align}
        \mathop{\TV}\limits_{\vec{R} \sim \CR_n}((\vec{R C}, \;\vec{R E R}^\dag), \;\CD_{n, k, p}) \leq \e ,
    \end{align}
    for all $\vec C$ whose corresponding stabilizer code $\vec S$ is in $\CV \subseteq \CC_n$, and for all $\vec{E}$ such that $\wt(\vec{E}) \leq w$.
\end{definition}
Here, $\e$ quantifies the quality of the reduction.
We say that the reduction is \emph{strong} if $\e = \negl(n)$.
$\CV \subseteq \CC_n$ represents the restriction to a set of worst-case Clifford operators, analogous to the restrictions placed on the generator matrix in the classical reduction.
Lastly, $p$ represents the error rate in the random case. 
A reduction with $p = \frac34$ is always possible, as this error is equivalent to replacing the encoded state with a maximally mixed state.
We say that the reduction is \emph{non-trivial} if $p = \frac{3}{4} - \frac{1}{\poly(n)}$.
Otherwise, the reduction is trivial, because the encoded state of the output of the reduction is negligibly close to the maximally mixed state.
We also consider \emph{tableau-randomizing} random self-reduction distributions, where instead of randomizing $\llbracket n, k \rrbracket$ stabilizer codes, we randomize a stabilizer tableau, an ordered list $\vec{S}_1, \;\dots, \;\vec{S}_{n-k}$ that forms a basis of a stabilizer subspace.
Applying $\vec{R}$ yields $\vec{RS}_1\vec{R}^{-1}, \;\dots, \;\vec{RS}_{n-k}\vec{R}^{-1}$, and a random self-reduction distribution of this type must bring this stabilizer tableau close to the uniform distribution. While it is not strictly necessary for a random self-reduction to randomize a tableau, a distribution $\mathcal{R}_n$ that randomizes stabilizer subspaces but not tableaus would be rather exotic. One of our barriers to a reduction will apply only to tableau-randomizing reductions, but still poses a significant obstacle. 

The above definition establishes a \textit{necessary} condition to achieve the type of worst-case to average-case reduction achieved in \cite{brakerski2019worst} and sought out in \cite{poremba2025learningstabilizersnoiseproblem}. 
A desired reduction will take a worst-case instance state $\vec{EC}\ket{0^{n-k},\mathbf{x}}$, and apply the random distribution $\mathcal{R}_n$, yielding
\begin{equation}
(\vec{RER}^\dag) (\vec{RC}) \ket{0^{n-k},\mathbf{x}},
\end{equation}
for $\vec{R} \sim \mathcal{R}_n$. 
In order for this outputted state to be part of a valid $\lsn$ instance, the distribution of $(\vec{RC})\ket{0^{n-k},\mathbf{x}}$ must be close to the distribution of $\vec{C}'\ket{0^{n-k},\mathbf{x}}$ for $\vec{C}' \sim \mathcal{C}_n$.
We have not mentioned the distribution over the secret $\vec x$ because it is always easy to re-randomize the secret by way of Lemma~\ref{lemma:secret_rerandomize}.

\begin{remark}[Reduction for $\slsn$]
    We do not consider random self-reductions for $\slsn$. 
    In the one-sample variant, $\lsn$ and $\slsn$ are equivalent by Lemma~\ref{lem:equiv-slsn-lsn}. 
    In the multi-sample variant, the scrambling operator would have to map the secret to at least a higher-order $t$-design while also mapping Cliffords to Cliffords by multiplication, which is manifestly impossible except perhaps for $t \leq 3$, in which case the barriers are similar to those we will explore for $\lsn$ itself.
\end{remark}

We require $\vec{R}$ to be a Clifford for the following reasons: If we are to be able to apply $\vec{R}$ to a noisy state, it must be a unitary (we do not consider a potential more general notion of a self-reduction for which the dimension is allowed to change).
Furthermore, we require that $\vec{R EC} = \vec{E' RC} = (\vec{RER}^\dagger) \vec{RC}$.
The latter condition requires $\vec{E}' = \vec{RER}^\dag \in \CP_n$, implying that $\vec{R}$ is necessarily a Clifford.
Therefore, the given definition is essentially as general as possible.
The only potential other operation one may consider is the application of an additional Pauli error $\widetilde{\vec{E}}$ at the end to $\vec{RER}^\dag$, a technique analogous to noise flooding~\cite{hsieh2024general,bergamaschi2025revisiting}.
While we, for simplicity, do not consider this extra operation in our definition, our results below hold even if this extra operation is allowed, because our arguments are centered around the claim that in order for $\vec{R}$ to scramble $\vec{C}$, it already necessarily scrambles the error $\vec{E} \mapsto \vec{RER}^\dag$ to an information-theoretically un-decodable error.
Hence, adding further randomness via additional noise cannot help.
We will in this work consider only random self-reductions which are strong.

\subsubsection{No Tableau-Randomizing Reduction Without Instance Structure}
In principle, it is possible to construct a uniform distribution over low-dimensional subspaces of a space such that the basis for each is not uniformly random.
However, a scrambling operator which completely scrambles the set of stabilizer codes without scrambling tableaus appears difficult to construct, and natural scrambling operations typically produce randomized tableaus.
We here prove, however, that no general tableau-randomizing reduction is possible.
That is, if $\CV = \Stab(n, k)$, then no strong non-trivial tableau-randomizing reduction exists.

\begin{lemma}[Scrambling stabilizer tableaus scrambles the error]
\label{lemma:scramble_collateral_damage}
    If a distribution $\CR_n$ over $\CC_n$ takes every stabilizer tableau $\vec{S}_1, \;\dots, \;\vec{S}_{n-k}$ corresponding to a stabilizer code $S \in \Stab(n, k)$ to a distribution $\vec{R S}_1\vec{R}^{-1}, \;\dots, \;\vec{R S}_{n-k}\vec{R}^{-1}$ for $\vec{R} \sim \mathcal{R}_n$ which has negligible distance to the uniform distribution over stabilizer tableaus (in TV distance), then $\CR_n$ must also take every error $\vec{E} \in \CP_n$ by conjugation to have negligible distance from $\CD_{3/4}^{\otimes n}$. 
\end{lemma}
\begin{proof}
There exists a stabilizer tableau with the error $\vec{E}$ in it, and after being randomized by $\vec{R} \sim \mathcal{R}_n$ it is $\epsilon$-close in total variation distance to uniform over all stabilizer tableaus, for $\epsilon = \negl(n)$.
Hence, $\vec{RER}^{-1}$ has a marginal distribution $\epsilon$-close to the marginal distribution of a single Pauli in a stabilizer tableau, i.e. a uniformly random Pauli in $\mathcal{P}_n \setminus \set{\vec{I}}$. 
This distribution has negligible total variation distance from $\mathcal{D}_{3/4}^{\otimes n}$, and therefore $\vec{RER}^{-1}$ does as well. 
\end{proof}

\begin{corollary}[No strong tableau-randomizing random-self-reduction for $\qncp$] \label{thm:no_w2a_general}
    There is no tableau-randomizing reduction $(\mathcal{V}, \e, w, p)$ random self-reduction distribution, where $\mathcal{V}$ is the set of all stabilized tableaus, from $\qncp(k, n, w)$ with computational basis logical states to $\lsn(k, n, p)$ such that $p =  \frac34 - \frac{1}{\poly(n)}$ and $\e = \negl(n)$.
\end{corollary}
\begin{proof}
    Suppose that such a reduction distribution exists, $\mathcal{R}_n$, over Cliffords. 
    Then the conditions of Lemma \ref{lemma:scramble_collateral_damage} are met, because for any stabilizer tableau, mixing by the distribution $\mathcal{R}_n$ yields a marginal distribution on $\vec{R}\vec{S}_1\vec{R}^{-1}, \;\dots, \;\vec{R}\vec{S}_{n-k}\vec{R}^{-1}$ that is $\epsilon$-close to the uniform distribution. 
    Hence, $\vec{RER}^{-1}$ is also $\epsilon$-close to the uniform distribution $\mathcal{D}_{3/4}^{\otimes n}$, and it therefore cannot have negligible total variation distance from $\mathcal{D}_{3/4 - 1/\poly(n)}^{\otimes n}$. 
    Indeed, there is an efficient algorithm that distinguishes these distributions, namely by sampling polynomially many times and estimating the expected error weight, which is either $\frac{3}{4}n$ or $\left(\frac{3}{4}-\frac{1}{\poly(n)}\right)n$. 
    But the existence of an efficient algorithm proves that the two distributions have non-negligible total variation distance. 
    
\end{proof}

\subsubsection{No Local Reduction for any Instance Structure}
While Corollary~\ref{thm:no_w2a_general} is a relatively strong barrier to random self-reductions, it is not an inherently quantum phenomenon.
To emphasize this fact, we prove in Appendix~\ref{sec:app-classical_w2a_no_go} that an analogous impossibility theorem holds classically.
However, classically there exists an explicit random self-reduction under a significant but nevertheless non-trivial set of restrictions in the worst-case instance.
Specifically, \cite{brakerski2019worst} shows that if the rate vanishes as $\frac{k}n = O(\frac1{n^{\epsilon}})$ for a fixed $\epsilon \in (0, 1)$; when the maximum error weight is $w = O\left(\frac{n^{\e}}{(1-\e)^2} \log^2 n\right)$; and when the worst-case instance's classical linear code is sufficiently \emph{balanced}; there exists a random self-reduction into $\lpn(k, n, p)$ where $p = \frac{1}{2} - \frac{1}{\poly(n)}$.
Here, a $\b$-\emph{balanced} code is one for which all nonzero elements of the code space have weight between $(\frac{1}{2} - \b)n$ and $(\frac{1}{2} + \b)n$.
Despite the fact that the worst-case instance has relatively low weight $w$ and the $\lpn$ instance has error probability so close to $\frac{1}{2}$ that it is difficult to imagine any decoding algorithm other than brute-force enumerated list decoding, the self-reduction is nonetheless remarkable because it achieves \emph{some} notion of a random self-reduction into $\lpn$ for the first time.

Corollary~\ref{thm:no_w2a_general} shows that no general reduction works because there is an inherent symmetry between the code and the error. 
A tableau-randomizing random self-reduction wishes to re-randomize the stabilizer tableau completely into a uniformly random set of Paulis, whilst simultaneously re-randomizing the error into a distribution of Paulis which are not close to maximal in weight, as such errors are not decodable even information theoretically.
However, both the tableau and the error are Paulis, so without some structure to break their symmetry, there is no way to act differently on tableaus and errors.
The reduction of \cite{brakerski2019worst} implicitly incorporate this insight by using a \emph{locally generated} distribution for their random self-reduction distribution.
Specifically, a reduction operator generated in a local manner (e.g. as the sum of random weight-1 vectors) guarantees a bounded norm which we can control.
Further, if we restrict the worst-case instance to have sufficiently low error weight and sufficiently high code weight, it becomes plausible that a weight-bounded random self-reduction operator could approximately randomize the code without completely randomizing the error.

Quantumly, the story appears substantially different. 
Enforcing a specific locality of the reduction operator appears to be the only reasonable technique for preventing an over-scrambled error.
Yet, as we show, for natural notions of a local reduction operator in the quantum setting, there is no random self-reduction for \emph{any} choice of restriction on stabilizer codes, and \emph{any} restriction on the error weight.
This inherently quantum phenomenon arises from the qualitatively distinct aspects of scrambling over the Clifford group, as compared to scrambling on a vector space over a binary field.

Any reasonable generation of the classical local self-reduction operator would have to be a sufficiently sparse Clifford.
In particular, this means a Clifford operator $\vec{C}$ for which its symplectic representation $\Symp(\vec{C})$ has $o(n)$ 1s in each column, i.e., $\max_i (|\vec{CX}_i\vec{C}^\dag|, \;|\vec{CZ}_i\vec{C}^\dag|) = o(n)$. 
Otherwise, even errors of constant weight could potentially be amplified to weight larger than $\frac{3}{4}n$, and therefore the randomized distribution of the error far from $\mathcal{D}_{3/4}^{\otimes n}$.

We prove, however, that no distribution over such non-trivially sparse Cliffords are viable self-reduction operators because there are simply not enough of them to reach the space of stabilizer codes.
While we otherwise consider strong reductions ($\e = \negl(n)$) in this work, the below theorem is a no-go for any $\e$ which is non-negligibly away from $1$.

\begin{theorem}[No random self-reduction from sparse Cliffords] \label{thm:w2a_nogo_sparse_Clifford}
If $k = n - \Omega(n)$, then number of stabilizer codes $N = |\Stab(n, k)|$ and the number of Cliffords $M$ with column sparsity $d = o(n)$ satisfy $\log(N) - \log(M) = \Omega(n^2)$. 
Consequently, no sparse random self-reduction distribution exists for any $\epsilon < 1 - 2^{-o(n^2)}$, any error weight $w \geq 1$, and any probability parameter $p$. 
\end{theorem}
\begin{proof}
The number of $n-k$ dimensional stabilizer subspaces is lower bounded by the number of phase-free stabilizer subspaces (those for which no Pauli operators have no phase). 
The number of such subspaces is equal to 
\begin{equation}
\frac{C_{n-k}}{|\mathrm{GL}_{n-k}(\mathbb{Z}_2)|},
\end{equation}
where $C_{n-k}$ is the number of ordered, symplectically orthogonal linearly independent lists in $\Z_2^{2n}$ of length $n-k$, and $\operatorname{GL}_{n-k}(\Z_2)$ is the set of invertible $(n-k) \times (n-k)$ matrices over $\Z_2$.
The number of bases of a given $(n-k)$-dimensional symplectic subspace with trivial symplectic inner product is $|\mathrm{GL}_{n-k}(\mathbb{Z}_2)|$.

The order of $\mathrm{GL}_{n-k}(\mathbb{Z}_2)$ can be upper bounded by the number of $(n-k)\times (n-k)$ matrices, and therefore $\log_2(|\mathrm{GL}_{n-k}(\mathbb{Z}_2)|) \leq (n-k)^2$. 
Defining the complement rate $\overline{R} = 1 - \frac{k}{n}$, it follows that $\log_2(|\mathrm{GL}_{n-k}(\mathbb{Z}_2)|) \leq \overline{R}^2 n^2$. 
Meanwhile, the number of ordered symplectically orthogonal bases can be estimated iteratively as follows. In the $i$th step, $2^{2n-i+1}$ of the elements commute with the previous $i-1$ elements of the basis. 
Of these elements, $2^{i-1}$ are in the span of the previous $i-1$. 
Taking the product of $2^{2n-i+1} - 2^{i-1}$ for $i\in[n-k]$, and noting that $2^{2n-i+1} - 2^{i-1} \geq 2^{2n - i}$, 
\begin{align}
\log_2(C_{n-k}) & \geq \sum_{i = 1}^{n-k} (2n - i)   = 2n(n-k) - \frac{(n-k)(n-k+1)}{2} \\
& = 2n(n-k) - \frac{1}{2} n (n-k) + \frac{1}{2} (k-1) (n-k) \\
& = \frac{3}{2} \overline{R} n^2 + \frac{\overline{R}}{2} (k-1) n \geq \frac{3}{2} \overline{R} n^2 .
\end{align}
By assumption, $k = n - \W(n)$, so equivalently $\overline{R} = \frac{n-k}{k} = \Omega(1)$.
Hence, noting that $\overline{R} \in [0, 1]$, \begin{align}
    \log_2(N) = \log_2(C_{n-k}) - \log_2(|\operatorname{GL}_n(\Z_2)|) \geq \left(\frac{3}{2} \overline{R} - \overline{R}^2 \right) n^2 = \W(n^2) .
\end{align}

Next, we upper bound the number of sparse Cliffords whose symplectic representation has column-sparsity $d \leq n$. 
Here, we are measuring sparseness in the symplectic representation, so the Cliffords are $2n \times 2n$ matrices over $\Z_2$ which preserve the symplectic inner product.
Certainly the number of $d$-sparse Cliffords is at most the number of $d$-sparse $2n \times 2n$ binary matrices.
The number of column vectors with a sparsity of $d$ is $\sum_{i=0}^d {2n \choose i}$.
By binomial-entropy bounds (see, e.g., Chapter 1 of \cite{richardson2008modern}), $\sum_{i=0}^d {2n \choose i} \leq 2^{n H_2(d/2n)}$, where $H_2(x) = - x \log_2(x) - (1-x)\log_2(1-x)$ is the binary entropy.
For $d = o(n)$, $H_2(\frac{d}{2n}) = o(1)$, so \begin{align}
    \log_2 \left( \sum_{i=0}^d {2n \choose i} \right) \leq n H_2\left(\frac{d}{2n}\right) = o(n) .
\end{align}
Summing up over columns, we conclude that $\log_2(M) = o(n^2)$, and therefore $\log_2(N) - \log_2(M) = \Omega(n^2)$. 
This completes the proof of the technical claim in the theorem.

To deduce that the above implies the non-existence of a sparse random self-reduction distribution for any $\mathcal{V} \subseteq \Stab(n, k)$, note that for any starting Clifford $\vec C$, randomizing over sparse Cliffords yields the distribution $\vec{R C}$, which has support of size $M = 2^{o(n^2)}$
Meanwhile, $\Unif(\Stab(n, k))$ is a uniform distribution over a support of size $N = 2^{\Omega(n^2)}$. 
As a consequence, the event that these two distributions coincide has probability $p = 2^{-\Omega(n^2) + o(n^2)} = 2^{-\W(n^2)} = \negl(n)$, implying
\begin{equation}
\TV(\vec{R C}, \;\Unif(\Stab(n, k))) = 1-2^{-\Omega(n^2)}. 
\end{equation}
In particular, for any $\vec E \in \CP_n$ and $p \in (0, 1)$,
\begin{equation}
\TV((\vec{R C}, \;\vec{RER}^\dag), \;\CD_{n, k, p}) = 1-2^{-\Omega(n^2)} .
\end{equation}
This statement concerns only about the number of sparse Cliffords versus the number of stabilizer codes, so it is independent of the error.
\end{proof}

As a consequence of the above result, a random self-reduction distribution cannot have bounded sparsity.
This is in contrast to the classical case, in which the reduction of \cite{brakerski2019worst} does indeed have bounded sparsity.
Without completely bounded sparsity, one of the only remaining natural approaches is to generate a product of $m$ random sparse operators in hopes that for a certain choice of $m$, a random self-reduction can be obtained.
We next show that for the simplest model of such an approach, no tableau-randomizing reduction exists.

A reasonable approach to the construction of a reduction operator is to take the product of Cliffords drawn from a small base distribution of sparse matrices.
We develop a formalism for the analysis of such constructions via a connection into the mixing of Markov processes.
A natural construction for the base distribution is a random $t$-local Clifford gate, where $t$ is a constant.
That is, sample $t$ random qubit indices, and apply a uniformly random Clifford on the $t$ qubits.
Such a distribution need not be the only plausible base distribution, but it is most natural in the sense that it is completely symmetric in how it mixes various Paulis and is local not only in sparsity as a symplectic matrix but also in qubit support.
We will prove that no such construction can yield a strong and non-trivial tableau-randomizing reduction, beginning with a map into Markov theory.
A finite Markov process (also known as a Markov chain) is a sequence of random variables $(X_t)_t$ such that \begin{align}
    \Pr[(X_t)_t = (x_t)_t] = \Pr[X_t = x_t \;|\; X_{t-1} = x_{t-1}] \cdots \Pr[X_1 = x_1 \;|\; X_0 = x_0] \Pr[X_0 = x_0] .
\end{align}
In other words, conditioned on the previous value, the current value is independent of all other previous values.
Such a process over a finite state space $\W$ can always be described by a transition matrix $Q \in \R^{|\W| \times |\W|}_{\geq 0}$, so that given an initial distribution $\mb{v} \in \R^{|\W|}_{\geq 0}$, the law of $X_t$ is given by $Q^t \mb{v}$.
A process is irreducible if every outcome in $\W$ has positive probability of being reached in finitely many steps.
The chain is aperiodic if it does not return to any particular outcome at regular intervals.
A basic result of Markov theory is that irreducible and aperiodic chains converge to a limiting distribution, regardless of the starting state.
This limiting distribution is known as the \emph{stationary distribution} $\pi$, and satisfies $Q \pi = \pi$.
That is, $\pi$ is the unique $+1$ eigenvector of $Q$.
Convergence is required for a random self-reduction, as the ultimate reduction distribution must be uniformly random in order to scramble to code.
As such, we formulate our analytical strategy only for irreducible and aperiodic chains.

\begin{definition}[Optimistic and pessimistic mixing times]
    Let $M$ be an irreducible and aperiodic Markov chain on state space $\W$ with transition matrix $Q$ and stationary distribution $\pi$. The \textit{optimistic mixing time} is given by \begin{align}
        \t_*(\e) := \min_{s \in \CS} \min \set{t \in \Z_{\geq 0} \,|\, \TV(Q^t \mathbf{1}_s, \;\pi) \leq \e} ,
    \end{align}
    and the \textit{pessimistic mixing time} is given by \begin{align}
        \t^*(\e) := \max_{s \in \CS} \min \set{t \in \Z_{\geq 0} \,|\, \TV(Q^t \mathbf{1}_s, \;\pi) \leq \e} ,
    \end{align}
    where $\mathbf{1}_s$ is the indicator vector on $s$. Note that $\t^*(\e)$ coincides with the standard definition of the mixing time of an irreducible, aperiodic Markov chain and represents the minimum time at which, regardless of starting point, the distribution is $\e$-close to the stationary distribution. On the other hand, $\t_*$ represents the minimum time at which \textit{some} choice of starting state will have fully mixed.
\end{definition}

Suppose now that we have a reduction for $\qncp(k, n, w)$, given by a distribution $\CR_n$ over $\CC_n$. 
If we study how the Clifford acts on a Pauli string $\vec{P} \in \CP^{\otimes n}$ by conjugation, then we can interpret the action as transitions on a Markov chain. 
The state space is $\W = \CP_n \setminus \set{\vec{I}}$ (i.e. all non-identity Paulis) and the transition matrix is 
\begin{align}
    Q_{\vec{P}_1, \vec{P}_2} = \Pr_{\vec{U} \sim \CD}[\vec{U P U}^\dag = \vec{P}_2 \;|\; \vec{P} = \vec{P}_1] .
\end{align}
We prove a much stronger statement below than necessary, which may be of interest in its own right.
Specifically, we show that the $t$-qubit random Clifford circuit model scrambles Paulis at an approximately uniform rate.
That is, whether the initial Pauli has weight 1 or weight $\frac{3}4n$, the process brings both Paulis to be $\e$-close to a uniformly random non-identity Pauli at approximately the same rate.
This result implies that no restriction on the code (via a restriction on the structure of the stabilizer tableau Paulis) and no restriction on the error can prevent the code from being maximally scrambled before the error is also maximally scrambled.
In fact, there is not even a way to prevent the error from being over-scrambled before even the first Pauli in the tableau becomes maximally scrambled.
As such, no random self-reduction can exist for any subset of codes, and at any regime of parameters using this local model.

\begin{lemma}[Small scrambling gap implies no tableau-randomizing reduction] \label{lemma:scrambling_gap}
    Let $\CR_n$ be a distribution over $\CC_n$ which defines a Markov chain $M$ on non-identity Paulis by repeatedly conjugating an initial non-identity Pauli by independent samples $\vec{R} \sim \CR_n$. 
    Assume this chain to be irreducible and aperiodic, such that the stationary distribution is a uniformly random non-identity Pauli. 
    Let $\eta(\e) > 0$ be the pessimistic error when scrambling optimistically, which we call the \emph{scrambling gap}.
    That is, $\eta(\e)$ is defined implicitly via \begin{align}
        \t^*(\eta(\e)) = \t_*(\e) .
    \end{align}
    By definition, $\eta(\e) \geq \e$.
    If $\eta(\e) = O(\poly(n) \cdot \e^b)$ for some constant $b > 0$, then for any $k, w$ and subset $\CV \subseteq \Stab(n, k)$, there does not exist a positive integer $d$ for which the distribution induced by applying $d$ samples from $\CR_n$ yields a $(\CV, \e, w, p)$ random self-reduction operator as defined in Definition~\ref{def:random_self_reduction}, such that $\e = \negl(n)$ and $p = \frac34 - \frac{1}{\poly(n)}$.
\end{lemma}
\begin{proof}
    Let $\CR_n^d$ be the distribution of $\vec{R}_d \cdots \vec{R}_1$, where $\vec{R}_i$ are i.i.d. samples from $\CR_n$.
    Suppose for contradiction that there exists a subset $\CV \subseteq \CC_n$, a weight $w$, and an integer $d > 0$ for which the unitary operation defined by applying $d$ samples from $\CR_n$ in sequence yields a strong and non-trivial self-reduction. 
    Let $\vec{C} \in \CV$ be a valid worst-case code and $\vec{E} \in \CP_n$ such that $1 \leq \wt(\vec{E}) \leq w$ be a valid worst-case error. 
    By assumption, \begin{align}
        \mathop{\TV}\limits_{\vec{R} \sim \CR_n^d}((\vec{RC}, \;\vec{RER}^\dag), \;\Unif(\CC_n) \otimes \CD_p^{\otimes n}) \leq \e = \negl(n) ,
    \end{align}
    where $p = \frac34 - \frac{1}{\poly(n)}$. Monotonicity of TV distance implies that \begin{align}
    \label{eq:Dp_vs_R_scrgap}
        \mathop{\TV}\limits_{\vec{R} \sim \CR_n^d}(\vec{RER}^\dag, \CD_{p}^{\otimes n}) \leq \e .
    \end{align}
    Let $\CT = \set{\vec{P}_1, \;\dots, \;\vec{P}_{n-k}}$ be a stabilizer tableau corresponding to $\vec{C}$. 
    Then the stabilizer tableau corresponding to $\vec{RC}$ is given by $\widetilde{\CT} = \set{\vec{RP}_1\vec{R}^\dag, \;\dots, \;\vec{RP}_{n-k} \vec{R}^\dag}$.
    The stabilizer tableau is computed deterministically from the corresponding Clifford, and the TV distance monotonically decreases with applying functions.
    This is because applying functions can only restrict the set of possible events, and the variational form of the TV distance is a maximization over all events.
    Consequently,
    \begin{align}
    \mathop{\TV}\limits_{\substack{\vec{R} \sim \CR_n^d, \widetilde{\vec{C}} \sim \CC_n}}(\vec{RP}_1\vec{R}^\dag, \; \widetilde{\vec{C}} \vec{P}_1 \widetilde{\vec{C}}^\dag) \leq \e.
    \end{align}
    Equivalently, with $\m_n$ the distribution of a uniformly random non-identity Pauli, \begin{align}
        \mathop{\TV}\limits_{\vec{R} \sim \CR_n^d}(\vec{RP}_1\vec{R}^\dag, \; \m_n) \leq \e .
    \end{align}
    By definition of the optimistic mixing time and the above, $d \geq \t_*(\e)$. 
    Then by assumption, $\t_*(\e) = \t^*(\eta(\e))$, so $d \geq \t^*(\eta(\e))$, where $\eta(\e) = \poly(n, \e) \leq n^a \e^b$ for some constants $a, b > 0$. 
    Using the definition of the pessimistic mixing time,
    \begin{align}
        \mathop{\TV}\limits_{\vec{R} \sim \CR_n^d}(\vec{RER}^\dag, \;\m_n) \leq \eta(\e) \leq n^a \e^b .
    \end{align}
    By the triangle inequality,
    \begin{align}
        \mathop{\TV}\limits_{\vec{R} \sim \CR_n^d}(\vec{RER}^\dag, \;\CD_{3/4}^{\otimes n}) \leq \mathop{\TV}\limits_{\vec{R} \sim \CR_n^d}(\vec{RER}^\dag, \;\m_n) + \TV(\m_n, \;\CD_{3/4}^{\otimes n}) \leq n^a \e^b + \exp{-\W(n)},
    \end{align}
    where the second equality holds because the total variation distance between $\mu_n$, a uniform distribution over Paulis, and $\mathcal{D}_{3/4}^{\otimes n}$, is $\exp{-\Omega(n)}$.
    At the same time, as previously discussed, \begin{align}
        \TV(\CD_{3/4}^{\otimes n}, \;\CD_{p}^{\otimes n}) \geq \frac{1}{\poly(n)} .
    \end{align}
    By the triangle inequality and Eqn.~(\ref{eq:Dp_vs_R_scrgap}), \begin{align}
        \frac{1}{\poly(n)} & \leq \TV(\CD_{3/4}^{\otimes n}, \;\CD_{p}^{\otimes n}) \\
        & \leq \mathop{\TV}\limits_{\vec{R} \sim \CR_n^d}(\vec{RER}^\dag, \;\CD_{p}^{\otimes n}) + \mathop{\TV}\limits_{\vec{R} \sim \CR_n^d}(\vec{RER}^\dag, \;\CD_{3/4}^{\otimes n}) \\
        & \leq \e + n^a \e^b + \exp{-\W(n)} \\
        & = \negl(n) ,
    \end{align}
    which is a contradiction.
\end{proof}

We now return to our particular model of study: random $t$-qubit Clifford operators applied in sequence, for a fixed constant $t \geq 2$.
It is convenient to re-express this Markov chain into a simpler form by way of the symplectic representation.
A $t$-qubit random Clifford acts on a Pauli by first choosing a subset of $t$ qubits randomly, then doing nothing if the Pauli on those $t$ qubits are all $I$.
Otherwise, it transforms the Pauli on those $t$ qubits into a uniformly random non-identity Pauli.
The corresponding transformation in symplectic form is described below.

\begin{definition}[$t$-Clifford chain] \label{def:Clifford_random_chain}
    Let $M$ be a Markov chain whose state space is $\W := \set{0,1,2,3}^n \setminus \set{0^n}$ and whose transitions are given by the following rule. Let the state be $s \in \W$. Pick a random $t$-tuple $(i_1, \;\dots \;i_t)$ where $i_j \in [n]$ and all $i_j$ are distinct. If $s_{i_j} = 0$ for all $j \in [t]$, do nothing. Otherwise, pick a random $t$-tuple $(a_1, \;\dots, \;a_t) \in \set{0,1,2,3}^t \setminus \set{0^t}$ and replace $(s_{i_1}, \;\dots, \;s_{i_t})$ with $(a_1, \;\dots, \;a_t)$.
\end{definition}
Note that this chain is irreducible and aperiodic~\cite{harrow2009random}. 
We now upper bound the pessimistic scrambling time and lower bound the optimistic scrambling time, and thereby show that the scrambling gap is sufficiently small to apply the nonexistence lemma (Lemma~\ref{lemma:scrambling_gap}).

\begin{lemma}[Random $t$-qubit Cliffords, upper bound] \label{lemma:2-qubitClifford_upper_bound}
    The $t$-Clifford chain from Definition~\ref{def:Clifford_random_chain} with  $\t^*(\e) = O(n \log \frac{n}{\e})$.
\end{lemma}
\begin{proof}
    This is proven by a combination of Theorem~5.1 and Corollary 5.2 in \cite{harrow2009random}. 
    While these theorems consider the case when $t = 2$, it holds equivalently for $t \geq 2$ a constant (it will not hold for $t = 1$ since no entanglement can be generated).
    Note that the proof of Corollary 5.2 is not correct in \cite{harrow2009random}, but a corrected proof is given by \cite{diniz2011comment}.
\end{proof}

In our lower bound, it does not matter if $t$ is not constant, so for completeness we include it in the asymptotic scaling statement.
While such a generalization is not necessary for this work, we believe that minor modifications to the upper bound proofs in \cite{harrow2009random} will give a matching scaling behavior of $O(\frac{n}{t} \log \frac{n}{\e})$.
We prove the lower bound by a mapping into the coupon collector problem from probability theory, which we first recall.

\begin{problem}[Coupon collector] \label{problem:coupon_collector}
    There are $n$ distinct coupons in a bag, and we repeatedly draw $t$ distinct coupons and then put them back. 
    How many draws do we need for the probability that you collected coupon 1 to be at least $1 - \e$?
\end{problem}

Let us give a lower bound for Problem~\ref{problem:coupon_collector}. Define the event \begin{align}
    C_i^m = \set{\text{coupon } i \text{ not drawn after } m \text{ draws}} .
\end{align}
We are interested in finding the minimum $m$ such that $\Pr[C_1^m] \leq \e$. 
\begin{lemma} \label{lemma:coupon_collector_lower_bound}
Let $n > t$.
    If $m = o(\frac{n}{t} \log(\frac1\e))$ then $\Pr[C_1^m] > \e$. Thus at least $m = \W(\frac{n}{t}\log (\frac1\e))$ are required for $\Pr[C_1^m] \leq \e$.
\end{lemma}
\begin{proof}
First, note that for all $n > 1$, $1 - \frac{1}{n} \geq \exp{\frac{1}{1-n}}$. 
This follows from a simple transformation of the standard bound $1 + x \leq e^x$, $\forall x \in \R$. 
Equivalently, $e^{-x} \leq \frac{1}{1 + x}$. 
Let $x = \frac{1}{n-1}$ so that $x > 0$. 
Then \begin{align} \label{eq:exponential_bound_coupon}
    \exp{\frac{1}{1-n}} & = \exp{-\frac{1}{n-1}} \leq \frac{1}{1 + \frac{1}{n-1}} = \frac{n-1}{n - 1 + 1} = \frac{n-1}{n} = 1 - \frac{1}{n} .
\end{align}
After $m$ steps, the probability of $C_1^m$ is \begin{align}
    \Pr[C_1^m] = \br{1 - \frac{t}{n}}^m \geq \exp{-\frac{m}{\frac nt-1}}
\end{align}
by Eqn.~(\ref{eq:exponential_bound_coupon}).
Solving, we obtain \begin{align}
    m < \left( \frac{n}{t} - 1 \right) \log \left(\frac1\e\right) = O\left(\frac{n}{t} \log \frac{1}{\e}\right).
\end{align}
Therefore, if $m = o(\frac{n}{t} \log(\frac1\e))$ then then $\Pr[C_1^m] > \e$.
\end{proof}

\begin{lemma}[Random $t$-qubit Cliffords, lower bound] \label{lemma:2-qubitClifford_lower_bound}
    The $t$-Clifford chain from Definition~\ref{def:Clifford_random_chain} has $\t_*(\e) = \W(\frac{n}{t} \log \frac{1}{\e})$, for all $t < n)$.
\end{lemma}
\begin{proof}
    Let $d = o(\frac{n}{t} \log \frac{1}{\e})$ be the number of steps for which we run the Clifford chain. 
    Let $\vec{V}$ be the unitary produced by this process. 
    We interpret the Clifford chain as a coupon collector problem. 
    In particular, without loss of generality assume that the first qubit supports the starting Pauli $\vec{P}_0 \neq I$---we may always relabel the qubits until this is the case. 
    The probability that the first qubit is not in the support of $\vec{R}$ (i.e. $\vec{R} = \vec{I}_1 \otimes \vec{R}_{\text{rest}}$) is at least the probability that after $d$ steps the first qubit was never chosen in any of the $t$-tuples. 
    This is precisely the coupon collector problem as stated in Problem~\ref{problem:coupon_collector}. Let $A_1$ be the event that $\vec{R} = \vec{I}_1 \otimes \vec{R}_{\text{rest}}$. 
    By Lemma~\ref{lemma:coupon_collector_lower_bound}, if $d = o(\frac{n}{t}) \log \frac{1}{\e}$ then $\Pr[A_1] > \e$. 
    Let $\vec{P}_1 := \vec{R P}_0 \vec{R}^\dag$, and consider the event $B_1 := \set{\vec{P}_1 \vec{P}_0 \rvert_1 = \vec{I}}$, where the notation $\vec{P} \rvert_1$ indicates the restriction of the Pauli $\vec{P}$ to the first qubit. 
    Let $B_2$ denote the event where $\vec{P}_1\rvert_1 \neq \vec{I}$.
    Now let $\CR_n^d$ be the distribution of $\vec{R}$.
    Note that $\Pr_{\vec{R} \sim \CR_n^d}[A_1 | B_2] \geq \Pr_{\vec{V} \sim \CR_n^d}[A_1 , B_2] = \Pr_{\vec{R} \sim \CR_n^d}[A_1]$.
    The last equality holds because if $A_1$ occurs then $B_2$ does as well, since if the first qubit is not touched then the Pauli remains non-identity. Now,
    \begin{align}
        \Pr_{\vec{R} \sim \CR_n^d}[B_1 | B_2] & = \Pr_{\vec{R} \sim \CR_n^d}[B_1 \;|\; A_1 , B_2] \Pr_{\vec{R} \sim \CR_n^d}[A_1 | B_2] \\
        & \phantom{hiii} + \Pr_{\vec{R} \sim \CR_n^d}[B_1 \;|\; \neg A_1 , B_2] \left(1 - \Pr_{\vec{R} \sim \CR_n^d}[A_1 | B_2]\right) \\
        &\geq (1) \Pr_{\vec{R} \sim \CR_n^d}[A_1] + \left(\frac{1}{3}\right) \left(1 - \Pr_{\vec{R} \sim \CR_n^d}[A_1]\right) > \frac{1}{3} + \frac{2\epsilon}{3}
    \end{align}
    since if $A_1$ occurs, $\vec{R}$ does not affect the first qubit so $P_0 \rvert_1 = P_1 \rvert_1$ and thus $P_1 P_0 \rvert_1 = I$. 
    But if $A_1$ does not occur, then $\vec{R}$ touches the first qubit. 
    The first time $\vec{R}$ touches the first qubit, it does so by scrambling it and $t-1$ other qubits into a uniformly random non-identity Pauli on $t$ terms. 
    In doing so, it removes any bias of the first qubit to any non-identity Pauli, and every subsequent Clifford leaves the marginal distribution unbiased between non-identity Paulis. 
    The probability of $B_1$ conditioned on both $\neg A_1$ and $B_2$ is therefore $\frac{1}{3}$, since it is equally likely to be any of the three non-identity Paulis. 
    
    On the other hand, under the stationary distribution $\pi$, a uniformly random non-identity Pauli, the first qubit is exponentially close to uniformly random. 
    Since the excluded outcome is the identity, it is unbiased between non-identity Paulis on the first qubit, so the probability of $B_1$ is $\frac{1}{3}$ conditioned on $B_2$.
    We therefore find that
    \begin{align}
        \left|\frac{\Pr_{\vec{R} \sim \CR_n^d}[B_1 , B_2]}{\Pr_{\vec{R} \sim \CR_n^d}[B_2]} - \frac{\Pr_{\pi}[B_1 , B_2]}{\Pr_{\pi}[B_2]} \right| &> \br{\frac{1}{3} + \frac{2\e}{3} \pm \exp{-\W(n)}} - \frac{1}{3} = \left(\frac{2}{3} \e \pm \exp{-\W(n)} \right) .
    \end{align}
    Now, note that $|\Pr_{\pi}[B_2] - \frac{3}{4}| \leq  \exp{-\W(n)}$, being the probability that the first qubit has a non-identity Pauli under the distribution of a uniformly random non-identity Pauli. 
    We may therefore assume that $\Pr_{\vec{R} \sim \CR_n^d}[B_2] \geq \frac{1}{2} - \exp{-\W(n)}$, or else the total variation distance of the two distributions would be at least $\frac{1}{4}$. 
    We therefore have a crude lower bound on each of the denominators.
    Under this assumption,
    \begin{align}
    \bigg|\Pr_{\vec{R} \sim \CR_n^d}[B_1 , B_2]\Pr_{\pi}[B_2] & - \Pr_{\pi}[B_1 , B_2]\Pr_{\vec{R} \sim \CR_n^d}[B_2] \bigg| \geq \Pr_{\pi}[ B_2]\Pr_{\vec{R} \sim \CR_n^d}[ B_2]\cdot \left( \frac{2}{3} \e \pm \exp{-\W(n)}\right) \\
    & \geq \br{\frac{3}{4} - \exp{-\W(n)}} \br{\frac{1}{2}- \exp{-\W(n)}} \left( \frac{2}{3} \e  \pm \exp{-\W(n)} \right) \\
    & = \W(\e) .
    \end{align}
    Now, assume for the sake of contradiction that both $\left|\Pr_{\vec{R} \sim \CR_n^d}[B_1 , B_2] - \Pr_{\pi}[B_1 , B_2]\right|=o(\epsilon)$ and $\left|\Pr_{\vec{R} \sim \CR_n^d}[B_2] - \Pr_{\pi}[B_2]\right|  = o(\epsilon)$. 
    Then, the difference above could be rewritten as
    \begin{align}
    \bigg|\Pr_{\vec{R} \sim \CR_n^d}[B_1 , B_2]\Pr_{\pi}[B_2] & - \Pr_{\pi}[B_1 , B_2]\Pr_{\vec{V} \sim \CR_n^d}[B_2] \bigg| \\
    & = \left|(\Pr_{\pi}[B_1 , B_2] + o(\epsilon))\Pr_{\pi}[B_2] - \Pr_{\pi}[B_1 , B_2](\Pr_{\pi}[B_2]+o(\epsilon)) \right| \\
    &= \left|o(\epsilon)(\Pr_{\pi}[B_2] + \Pr_{\pi}[B_1 , B_2]) + o(\epsilon)^2\right| \\
    &= o(\epsilon).
    \end{align}
    This is a contradiction, and therefore
    \begin{align}
    \TV_{\vec{R} \sim \CR_n^d}(\vec{R P}_0 \vec{R}^\dag , \;\pi) & \geq \max\left(\left|\Pr_{\vec{V} \sim \CR_n^d}[B_1 , B_2] - \Pr_{\pi}[B_1 , B_2]\right|, \;\left|\Pr_{\vec{V} \sim \CR_n^d}[ B_2] - \Pr_{\pi}[ B_2]\right|\right)\\
    & = \Omega(\epsilon) . 
    \end{align}
    Consequently, if $d = o(\frac{n}{t} \log \frac{1}{\e})$ then regardless of starting Pauli, the distribution is $\W(\e)$ away from the stationary distribution.
\end{proof}


\begin{theorem}[No local tableau-randomizing reduction for any parameter regime] \label{thm:no_w2a_even_if_restricted}
    Fix any logical qubit function $k(n)$ and weight function $w(n)$, and any subset of codes $\CV(n) \subseteq \CC_n$.
    Then there is no strong, non-trivial random self-reduction from $\qncp(k, n, w)$ (with computational basis logical states) to $\lsn(k, n, p)$.
    More precisely, there is no $(\CV(n), \;\e, \;w(n), \;p(n))$ random self-reduction distribution $\CR_n$ on $\llbracket n, k \rrbracket$ stabilizer codes, such that $p = \frac34 - \frac{1}{\poly(n)}$ and $\e = \negl(n)$, and $\CR_n$ is a sequence of random $t$-local Cliffords.
\end{theorem}
\begin{proof}
    The pessimistic mixing time of the chain is given by $\t^*(\e) = O(n \log \frac{n}{\e}) \leq \a n \log \frac{n}{\e}$ for sufficiently large $n$, by Lemma~\ref{lemma:2-qubitClifford_upper_bound}. The optimistic mixing time is given by $\t_*(\e) = \W(n \log \frac{1}{\e}) \geq \beta n \log \frac{1}{\e}$ for sufficiently large $n$, by Lemma~\ref{lemma:2-qubitClifford_lower_bound}. Here $\a, \b$ are constants. Then \begin{align}
        \b n \log \frac{1}{\e} \leq \t_*(\e) = \t^*(\eta(\e)) \leq \a n \log \frac{n}{\eta(\e)} .
    \end{align}
    Solving, we find that $\eta(\e) \leq n \e^{\b/\a}$. Hence, $\eta(\e)$ satisfies the assumptions of Lemma~\ref{lemma:scrambling_gap}, which completes the proof of the theorem.
\end{proof}

These combined results show that the landscape of random self-reducibility for stabilizer codes qualitatively differs from the classical analog.
Both classically and quantumly, restriction on the set of allowed codes in the worst case, as well as the weight of the worst-case error, is necessary.
However, with such restrictions, some notion of a random self-reduction is achievable, because local scrambling distributions there mix certain high-weight bitstrings substantially more quickly than low-weight bitstrings.
Therefore, by assigning codes to be comprised of the certain structured high-weight bitstrings and the errors to be comprised of low-weight bitstrings, a local model suffices as a reduction distribution.
Quantumly, the natural analogs of local scrambling---sparse Cliffords or products of sparse Cliffords---have no clear gap.
Sparse Cliffords unconditionally cannot give a reduction, and products of sparse Cliffords cannot randomize a tableau without randomizing the error.
Therefore, quantumly it appears that the primary barrier is not necessarily finding the right restriction of stabilizer codes, but rather finding how to (a) scramble the code space without scrambling the tableau and (b) scramble using non-sparse Cliffords while also not blowing up the error.
Satisfying just one of these conditions already appears challenging, so the two constraints together pose a substantial barrier to a random self-reduction.
It remains an open question as to whether these barriers can be surpassed, or can be strengthened to a complete impossibility theorem for random self-reductions.
We conjecture that, at least in the case of tableau-randomizing reductions, analysis techniques generalizing the scrambling gap formalism will reveal unconditionally that a reduction cannot scramble the tableau.
\begin{conjecture}
    There is no $(\CV, \e, w, p)$ tableau-randomizing random self-reduction operator (as defined in Definition~\ref{def:random_self_reduction}) for $\lsn$, for any $\CV \subseteq \CC_n$, $w \geq 1$, such that $\e = \negl(n)$ and $p = \frac34 - \frac{1}{\poly(n)}$.
\end{conjecture}

At the same time, these barriers do not definitely prove the non-existence of a random self-reduction.
Moreover, we have not considered an even more general notion of a random self-reduction, e.g. when the dimensions $n, k$ can increase by a polynomial factor.
It is unclear whether allowing for a more broad definition of a random self-reduction can circumvent the strength of the barrier discussed above.
More precisely, can a random self-reduction be achieved by overcoming the above barriers, by increasing the dimension (i.e. by letting the reduction operator $\vec{R}$ be an isometry), or by letting $\e = \frac{1}{\poly(n)}$ (in which case we would require a smaller $p$ to be non-trivial, e.g. $p < \frac34 - \W(1)$)?

\section{Average-Case  Hardness of Quantum Stabilizer Decoding} \label{sec:hardness_lsn}
Despite the apparent challenges in giving a random self-reduction for $\lsn$, we establish in this section the algorithmic hardness of $\lsn$ by directly reducing $\lpn$ to $\lsn$.
More precisely, we show that $\lsn[k, n, p][m]$ for any $k \geq 1$, $p = \Omega(n^{-(1-\epsilon)})$ with $\e \in (0, 1)$, and any $m \geq 1$ is harder than an instance of $\lpn[k', \;2n, \;p']$ for $k' =p'n$, $p' = \frac{p}{6}$. 
These parameters include any problem instance with a constant error probability $p$. 
In this regime, our result states that $\lsn$ with \emph{any} number of samples $m$ and \emph{any} number of logical qubits $k \geq 1$, is as hard as $\lpn$ with a constant rate. 
In spite of much effort, no sub-exponential time algorithms for $\lpn$ at constant rate and constant error probability are known.
Thus, the existence of a sub-exponential time algorithm which can decode a typical quantum stabilizer code---subject to a constant noise rate---with non-negligible probability implies by way of this reduction a cryptographic breakthrough in the form of a sub-exponential algorithm for the hardest regime of $\lsn$.

It is illustrative to observe how standard classical attacks fail on $\lsn$.
We discuss the algorithms of BKW~\cite{BKW03} and Lyubashevsky~\cite{10.1007/11538462_32} which apply respectively in the regimes of negligible rate and inverse-polynomial rate.
The BKW algorithm finds a sparse vector $\mathbf{s}^\intercal$ that is in the left-kernel of the code $\vec{A}$. 
Then, calculating $\mathbf{s}^\intercal (\vec{A}\vec{x}+\vec{e}) = \mathbf{s}^\intercal \vec{e}$ solves Decision $\lpn$, because this product is 1 with low probability if $\mathbf{s}$ is sufficiently small.
However, in order for a sparse vector to exist at all, there must be superpolynomially (in $k$) many rows. 
The runtime of the algorithm is $2^{O(k/\log k)}$, and in particular sub-exponential. 
$\lsn$ corresponds to a problem of $\lpn$ where the rate is artificially restricted (in the sense that the encoding matrix in the classical representation has dimensions $2n \times (n+k)$, so the ``rate'' is at least $1/2$), and therefore BKW cannot apply at all.
Lyubashevsy's algorithm is a technique for amplifying polynomially many rows ($k^{1+\epsilon}$ for any $\epsilon>0$) to superpolynomially many, and then applying the BKW algorithm to achieve a runtime of $2^{O(k/\log\log k)}$. 
Once again, $\lsn$ with constant probability and any rate corresponds to an $\lpn$ problem with a constant rate, and therefore Lyubashevsky's algorithm fails as well. 

This explicit relationship between $\lpn$ and $\lsn$ establishes a qualitative difference between $\lpn[k, n, p]$ and $\lsn[k, n, p]$. 
The former problem has a trivial brute force algorithm with runtime $O(2^k n^3)$: guess the solution $\vec{x}'$, and check if $\vec{Ax}'$ is close to the sample $\vec{Ax}+\vec{e}$. 
On the other hand, even when $k = 1$, $\lsn[k, n, p]$ is as hard as $\lpn[k', \;2n, \;p']$ for $k' =p'n$, $p' = \frac{p}{6}$. 
Therefore, there cannot be an algorithm solving $\lsn$ with any runtime of the form $O(f(k) \cdot \poly(n))$ for any $f(k)$ (e.g. $f(k) = 2^k$ or $2^{2^k}$), as such an algorithm would imply an \emph{efficient} algorithm for constant-rate $\lpn$.
Previously, we noted this barrier by observing that there is no obvious way to certify the validity of a proposed answer $\vec x'$ to a $\lsn$ instance $\vec{EC}\ket{0^{n-k},\vec{x}}$.
To do so would be to solve a certain short vector problem: determining if $\vec{EC}\ket{0^{n-k},\vec{x}}$ and $\vec{C}\ket{0^{n-k},\vec{x}'}$ differ by a low-weight Pauli.
We can now establish a much stronger statement, namely that the existence of an efficient certification protocol, or equivalently the existence of a solution to this short vector problem, would imply a polynomial time algorithm solving $\lpn$ in the hardest regime.
One interpretation of these results is that while in $\lpn$ the natural security parameter is $k$, in $\lsn$ the natural security parameter is $n$, making $\lsn$ substantially harder when $n$ and $k$ differ significantly.

Our proof of this reduction ties together several components from previous sections.
Because $\lsn$ and $\lpn$ differ greatly even at a qualitative level in their definitions, we introduce several hybrids to bridge the gap.
On the quantum side, we will use exclusively the equivalent classical representation of $\lsn$ developed in Section~\ref{sec:classical_lsn}.
This representation appears much more qualitatively similar to $\lpn$ in the sense that the samples are of the form $(\vec B, \;\vec B \vec z + \vec e)$.
However, it differs in that (a) $\vec z$ is mostly junk (the secret constitutes a small sliver of $\vec z$), (b) the columns of $\vec B$ have certain symplectic orthogonality conditions, and (c) the error is depolarizing instead of Bernoulli.
Therefore, on the classical side, we will introduce a problem which we call \emph{symplectic learning parity with noise} ($\symplpn$), which differs from $\lpn$ in that the encoding matrix has symplectically orthogonal columns and that the error is depolarizing.
We reduce $\lpn$ to $\symplpn$ by modifying the code and error, then from $\symplpn$ to the classical representation of $\lsn$.
The latter reduction proceeds by a strange step wherein we embed directly into the junk, not into the secret of $\lsn$.
As a consequence, we can reduce only into the decision variant of $\lsn$ instead of the search variant.
Since we wish ultimately to reduce to the search variant, we will give the reduction from $\lpn$ to the multi-sample variant of Decision $\lsn$ and then apply our multi-sample decision-to-search reduction from Section~\ref{sec:self_reducibility}. 
This together will complete the reductive steps from Fig.~\ref{fig:LSN_main_figure}.
This final step is particularly interesting, because the multi-sample variant plays a crucial role despite being a non-standard notion of quantum error correction.

\subsection{Symplectic Learning Parity with Noise}

To begin, we formally introduce the $\symplpn$ problem. We introduce both \emph{search} and \emph{decision} variants of the problem; although we will mainly consider the latter variant.
Informally, $\symplpn$ is the task of decoding a random classical linear code with a symplectically self-orthogonal encoding matrix $\vec A$, rate $\frac{k}n = \frac{1}{2}$, and error drawn from a depolarizing distribution (in symplectic representation).

\begin{definition}[$\symplpn$]
The Symplectic Learning Parity with Noise problem, denoted by $\symplpn[n, p]$, is characterized by an integer $n \in \mathbb{N}$ and $p \in (0,1)$. We distinguish between two variants of the problem:
\begin{description}
    \item $\bullet$ \textbf{Search} $\symplpn[n, p]$ is the task of finding $\vec x \sim \Z_2^{n}$ given as input a sample of the form
    \begin{align}
        (\vec{A} \in \Z_2^{2n \times n}, \;\vec{A}\mathbf{x} + \mathbf{e} \Mod{2})
    \end{align}
    where $\vec A$ is random subject to the constraint that its $n$ columns are linearly independent and symplectically orthogonal, and where $\vec e$ is a depolarizing error with probability parameter $p$, i.e. $\Symp^{-1}(\vec e) \sim \CD_p^{\otimes n}$.

    \item $\bullet$ \textbf{Decision} $\symplpn[n, p]$ is the task of distinguishing between the samples
    \begin{align}
        (\vec{A} \in \Z_2^{2n \times n}, \;\vec{A}\mathbf{x} + \mathbf{e} \Mod{2}) \quad \text{ or } \quad (\vec{A} \in \Z_2^{2n \times n}, \;\mathbf{u} \sim \Z_2^{2n})
    \end{align} 
    where $\vec A, \vec x, \vec e$ are as before and $\vec u$ is a uniformly random bitstring.
\end{description}
\end{definition}
This definition is motivated by bridging part but not all of the differences between $\lpn$ and the classical representation of $\lsn$.
We will first prove that Decision $\lpn$ reduces to Decision $\symplpn$, and then construct this latter reduction.
To reduce $\lpn$ to $\symplpn$, we first construct an algorithm that takes a matrix $\vec{A} \in \Z_2^{(2n-\ell(1+\epsilon)) \times l}$, and adds $\ell(1+\epsilon)$ rows to create a matrix $\vec{B} \in \Z_2^{2n \times \ell}$ with symplectically orthogonal columns.
This algorithm has the property that it takes the uniform distribution over codes $\vec{A}$ to a distribution close to uniform over symplectically orthogonal codes.
We then correct the error distribution by applying the depolarizing decomposition into Bernoullis from Corollary~\ref{corollary:bernoulli-depolarizing_duality}.

\begin{lemma}[Symplectic matrix extension] \label{lem:symplectic_extension} Let $n,\ell \in \mathbb{N}$ and $\e > 0$ be parameters such that $\ell(1+\e) \leq n$.
There exists a $\poly(n)$ time classical algorithm which on input a uniformly random $(2n - \ell (1+\epsilon)) \times \ell$ matrix $\vec{A}$, outputs 
\begin{align}
    \vec{B} = \begin{bmatrix} \vec{A} \\ \vec{A'} \end{bmatrix} \in \Z_2^{2n \times \ell}
\end{align} 
for a matrix $\vec{A'}$ with dimensions $\ell(1+\epsilon) \times \ell$, such that \begin{align}
    \mathop{\TV}\limits_{\vec B \sim \nu}(\vec{B}, \mu_{2n, \ell}) = \negl(\ell),
\end{align}
where $\mu_{2n, \ell}$ is the uniform distribution over $2n \times \ell$ full-rank matrices with symplectically orthogonal columns and $\nu$ is the output distribution of the algorithm. 
\end{lemma}

\begin{proof}
Decompose $\boldsymbol{A}$ as 
\begin{equation}
\vec{A} = \begin{bmatrix} \vec{N}_1\\\vec{M} \\ \vec{N}_2 \end{bmatrix}
\end{equation}
where $\vec{N}_1, \vec{N}_2 \in \Z_2^{(n- \ell(1+\epsilon)) \times \ell}$, and $\vec{M} \in \Z_2^{\ell(1+\epsilon) \times \ell}$. 
The matrix $\vec{B}$ to has symplectically orthogonal columns if and only if the matrix
\begin{equation}
\vec{S}(\vec A') = (\vec{N}_1)^\intercal \cdot \vec{N}_2 + \vec{M}^\intercal \cdot \vec{A'}
\end{equation}
is symmetric. 
We can check this fact by observing that $\vec{S}_{ij} + \vec{S}_{ji}$ is precisely the symplectic inner product of the $i$th and $j$th columns of $\vec{B}$. 
Hence, $\vec{S}_{ij} = \vec{S}_{ji}$ corresponds to symplectic orthogonality. 

We sample a random symmetric matrix $\vec{S'} \in \Z_2^{\ell \times \ell}$, and let $\vec{T} = \vec{S'} - (\vec{N}_1)^\intercal \cdot \vec{N}_2$. 
Sample a random matrix $\vec{A'}$ for which $\vec{M}^\intercal \cdot \vec{A'} = \vec{T}$; this can be done efficiently being that it is a random solution to a system of linear equations. 
(If no such matrix exists, then output the zero matrix.)
Append $\vec A'$ to $\vec A$ to form $\vec B$.
For any such choice of $\vec{A'}$, 
\begin{equation}
\vec{S}(\vec A') = (\vec{N}_1)^\intercal \cdot \vec{N}_2 + \vec{M}^\intercal \cdot \vec{A'} = \vec{S'},
\end{equation}
and therefore $\vec{B}$ is symplectically orthogonal as desired. 

We next show that if $\vec{A}$ is uniformly random, then $\vec{B}$ is close to a uniformly random full-rank matrix with symplectically orthogonal columns. 
We do so by proving that for almost all $\vec A$, the number of choices $n_{\vec B}(\vec A)$ for a full-rank symplectically orthogonal $\vec B$ are the same.
Suppose first that $\vec M^\intercal \in \Z_2^{\ell \times \ell (1+\e)}$ is full-rank; this occurs with probability $r = 1 - \negl(\ell)$ since $\vec A$ is uniformly random.
Note that by dimensionality, $\vec M^\intercal$ being full-rank implies that $\vec B$ and $\vec A$ are full-rank.
Then we claim that 
\begin{equation}
n_{\vec{B}}(\vec{A}) = 2^{\ell(\ell+1)/2} \cdot 2^{\epsilon \ell^2}.
\end{equation}
There are $2^{\ell(\ell+1)/2}$ choices for $\vec S'$.
For each $\vec S'$, there are $2^{\dim \vec A' - \dim \vec S'} = 2^{(1+\e)\ell^2 - \ell^2} = 2^{\e \ell^2}$ choices of $\vec A'$ by the assumption that $\vec M^\intercal$ is full rank.
For a fixed $\vec A$, the choice of $\vec A'$ uniquely determines $\vec B$ and $\vec S'$, so every pair $(\vec A', \vec S')$ gives a distinct solution $\vec B$, proving the claim.
Let $n_{\nu}$ be the number of $\vec B$ that the algorithm can produce, conditioned on a full-rank $\vec M^\intercal$.
$n_{\vec B}(\vec A)$ does not depend on $\vec A$, so \begin{align}
    n_{\nu} = 2^{(2n - \ell(1+\e))\ell} \cdot 2^{\ell(\ell+1)/2} \cdot 2^{\epsilon \ell^2} = 2^{2n\ell - \ell^2(1+\e) + \ell(\ell+1)/2 + \e \ell^2} = 2^{2 n \ell - \ell(\ell-1)/2} .
\end{align}
On the other hand, we claim that there are exactly \begin{align}
    n_{\mu} = \prod_{j=1}^{\ell} (2^{2n - (j-1)} - 2^{j-1})
\end{align}
$2n \times \ell$ binary matrices which are full-rank and symplectically orthogonal.
This can be seen from an iterative construction procedure: at the $j$th step, out of $2^{2n}$ possible strings, the condition of symplectic orthogonality with $j-1$ previous independent strings decreases the dimensionality by $j-1$.
Furthermore, among the $2^{2n - (j-1)}$ strings which are symplectically orthogonal to the previous strings, $2^{j-1}$ lie in the span of the previous strings, so we subtract off those strings.
Asymptotically, however, this subtraction bears no relevance.
This is because our assumption of $\ell(1+\e) \leq n$ implies that $n - \ell \geq \e \ell$, and
\begin{align}
    \frac{2^{j-1}}{2^{2n - (j-1)}} = 2^{-2(n-(j-1))} \leq 2^{-2(n-(\ell-1))} \leq 2^{-2(\e \ell + 1)} = 2^{- \e \ell} = \negl(\ell) .
\end{align}
Hence, $n_{\mu} = r' \cdot 2^{2 n \ell - \ell(\ell - 1)/2} = r' \cdot n_{\nu}$, where $r' = 1 - \negl(\ell)$.
To complete the proof, note that according to the algorithmic distribution $\nu$, conditioned on an event ($M^\intercal$ is full rank) which has probability $r = 1 - \negl(\ell)$, the distribution is uniform over $n_{\nu}$ elements.
Meanwhile, according to the desired distribution $\mu_{2n, \ell}$, every full-rank symplectically orthogonal matrix has equal probability of being drawn, of which there are $n_{\mu} = r' n_{\nu} = (1 - \negl(\ell)) n_{\nu}$ of them.
These probabilities differ by a negligible multiplicative factor, with $\nu$ conditioning on an event with $1 - \negl(\ell)$ probability.
Therefore, \begin{align}
    \mathop{\TV}\limits_{\vec B \sim \nu}(\vec{B}, \mu_{2n, \ell}) = \negl(\ell)
\end{align}
as desired.
\end{proof}

Before demonstrating the reduction, we show that Bernoulli errors on a fixed, known set of indices can be randomized to produce independent Bernoulli errors with lower probability on all indices. 
\begin{lemma}[Bernoulli noise extension] \label{lem:error_extension}
Let $n,m \in \mathbb{N}$ be parameters such that $m = \omega(\log n)$, and let $\mathbf{e} \in \Z_2^n$ be a vector such that $e_i = 0$ for $i \leq n-m$, and $e_i \sim \Ber(\frac{1}{2})$ for $i > n-m$.
Then for any constant $\d \in (0, 1)$, there is an efficiently sampleable distribution $\m_{n, m}$ over $\Z_2^n$ such that
\begin{equation}
\mathop{\TV}\limits_{\substack{\pi \sim S^n\\ \vec e' \sim \m_{n, m}}}(\pi(\mathbf{e} + \mathbf{e'}), \;\Ber(p)^{\otimes n}) = \negl(n) .
\end{equation}
Here, $p = \frac{m}{2n}\left(1 + \delta\right)$ and $\pi$ is a uniformly random permutation drawn from the symmetric group $S^n$ on $n$ elements.
\end{lemma}
\begin{proof}
We will construct a procedure for sampling $\mathbf{e'}$.
First, sample $T \sim \Bin(n, \;2p)$, so that $\mathbb{E}[T] = 2np = m(1+\d)$.
By a Chernoff bound,
\begin{equation}
\Pr\left[M \geq \left(1 - \frac\d2\right) \mathbb{E}[T]\right]
 \leq \exp{- \frac{\d^2}{2 + \d} \mathbb{E}[T]} = \exp{- \frac{\d^2 (1 + \d)}{2 + \d} m} = \negl(n) ,
\end{equation}
since $m = \w(\log n)$.
Therefore, with probability $1 - \negl(n)$, 
\begin{equation}
T \geq m(1+\delta) \left(1 - \frac\d2\right) = m \left(1 + \frac\d2 - \frac{\d^2}2 \right) > m. 
\end{equation}
Choose $\mathbf{e'}$ to be a vector where the first $n - T$ entries are zero, and the rest are sampled uniformly randomly. 
Conditioned on the event that $T > m$, $\vec e + \vec e'$ is equidistributed as $\vec e'$ itself, since $\vec e$ has zeros in the first $n-m < n-T$ entries, and the remaining entries are uniformly random.

Because $\pi$ erases any correlation with index, the distribution of $\pi(\mathbf{e} + \mathbf{e'})$, conditioned on $T > m$, can be equivalently represented by first sampling $T \sim \Bin(n, \;2p)$ random indices, and then choosing the bit at each index to be 1 with probability $\frac{1}{2}$ independently.
The random choice of indices can instead be generated by independently for each index selecting it with probability $2p$.
Thus, each bit is $1$ independently with probability $2p \cdot \frac{1}{2} = p$.
Consequently, except under a negligibly improbable event, $\pi(\vec e + \vec e')$ is distributed as $\Ber(p)^{\otimes n}$, which completes the proof.
\end{proof}

We now use the above lemmas to give a reduction from $\lpn$ to $\symplpn$.
\begin{theorem}[$\lpn$ to $\symplpn$] \label{thm:lpn_to_symplpn}
Let $n, \ell \in \mathbb{N}$ be parameters and $\e \in (0, 1)$ be a constant, such that $\ell(1+\e) \leq n$ and $\ell = \Omega(n^{\epsilon})$.
Define $r = \frac{\ell}{2n}(1+\epsilon)^2 \leq p + \frac{\ell}{2n}(1+3\epsilon)$, where $p \in (0, 1)$, and $q = p + r - 2pr \leq p + \frac{\ell}{2n}(1+3\epsilon)$.
Let $\CO$ be an oracle that solves Decision $\symplpn[n, \;3q]$ with advantage $\eta$.
Then there exists an algorithm, running in time $\poly(n)$, which solves Decision $\lpn[\ell, \; \lceil 2n - \ell(1+\epsilon) \rceil, p]$ to advantage $\eta - \negl(n)$, using a single call to $\CO$.
\end{theorem}

\begin{proof}
Let $(\vec{A},\mathbf{y})$ be a sample of Decision $\lpn[\ell, \;2n - \ell (1 + \epsilon), p]$, where either $\mathbf{y} = \vec{A}\mathbf{x} + \mathbf{e}$ or $\vec y \sim \Z_2^{\ell}$.
First, suppose that $\mathbf{y} = \vec{A}\mathbf{x} + \mathbf{e}$. 
By Lemma \ref{lem:symplectic_extension}, we may obtain a full-rank, symplectically orthogonal matrix $\vec{B} = \begin{bmatrix} \vec{A} \\ \vec{A'} \end{bmatrix}$ such that the distribution over $\vec B$ is negligibly close to uniform over full-rank symplectically orthogonal matrices in $\Z_2^{2n \times \ell}$.
Sample $\vec r \sim \Z_2^{\ell(1 + \e)}$ and define
\begin{equation}
\mathbf{y'} = \begin{bmatrix} \mathbf{y}\\ \mathbf{r} \end{bmatrix} = \begin{bmatrix} \mathbf{y}_1\\ \mathbf{y}_2 \end{bmatrix} ,
\end{equation}
where $\vec y_1, \vec y_2 \in \Z_2^{n}$.
Likewise, we re-write $\vec{B}$ as 
\begin{equation}
\vec{B} = \begin{bmatrix}\vec{A}_1 \\\vec{A}_2 \end{bmatrix},
\end{equation}
where $\vec A_1, \vec A_2 \in \Z_2^{n \times \ell}$. 
By definition, $\mathbf{y}_1 = \vec{A}_1\mathbf{x} + \mathbf{e}_1$ for a Bernoulli error $\mathbf{e}_1 \sim \Ber(p)^{\otimes n}$. 
Because $\vec r$ is uniformly random, we may also equivalently express $\mathbf{y}_2$ as 
\begin{equation}
\mathbf{y}_2 = \vec{A}_2\mathbf{x} + \mathbf{e}_2 + \mathbf{f},
\end{equation}
where $\vec e_2$ is $\Ber(p)$ on the first $n - \ell(1+\e)$ bits and zero on the rest, and $\mathbf{f}$ is zero on the first $n - \ell (1+\epsilon)$ bits and independently $\Ber(\frac{1}{2})$ on the rest.

Next, sample a random permutation $\pi \sim S^n$. Denote by $\pi \oplus \pi$ a permutation operator on $\Z_2^{2n}$ whereby the first $n$ and last $n$ bits are permuted according to the same permutation $\pi \in S^n$.
Denote $(\pi \oplus \pi) (\vec M)$ for a matrix $\vec M \in \Z_2^{2n \times \ell}$ the action of $\pi \oplus \pi$ column-wise on $\vec M$.
By definition, $\pi \oplus \pi$ preserves the symplectic inner product, since for each $i \in [n]$ indices $(i, i+n)$ are relabeled jointly to $(j, j+n)$ for some $j \in [n]$.
Thus, the distribution of a full-rank symplectically orthogonal matrix $\vec B$ is invariant under action of $\pi \oplus \pi$.
We also sample $\vec e' \in \Z_2^n$ using Lemma~\ref{lem:error_extension}, with parameters $m = \ell(1+\e)$ and $\d = \e$.
Now, define 
\begin{equation}
    \mathbf{z} = \begin{bmatrix} \mathbf{z_1}\\ \mathbf{z_2}\end{bmatrix} = \begin{bmatrix} \pi (\mathbf{y_1}) \\ \pi (\mathbf{y_2} + \mathbf{e'})\end{bmatrix}
\end{equation}
and
\begin{equation}
    \vec{K} = \begin{bmatrix} \vec{K}_1 \\ \vec{K}_2\end{bmatrix} = (\pi \oplus \pi) (\vec B) = \begin{bmatrix} \pi (\vec{A}_1) \\ \pi (\vec{A}_2) \end{bmatrix}.
\end{equation}
Under these definitions,
\begin{equation}
    \mathbf{z}_1 = \vec K_1 \vec x + \vec e_1 ,\quad \vec z_2 = \vec{K}_2 \vec{x} + \pi(\mathbf{e}_2) + \pi(\mathbf{f} + \mathbf{e'}) .
\end{equation}
Marginally, $\pi(\vec f + \vec e') \sim \Ber(q)^{\otimes n}$ up to $\negl(n)$ corrections in probabilities by Lemma~\ref{lem:error_extension}.
Meanwhile, independently, $\pi(\vec e_2) \sim \Ber(p)^{\otimes n}$.
Hence, $\pi(\vec e_2) + \pi(\vec f + \vec e')$ is distributed as an $n$-fold Bernoulli with probability $r = \frac{l}{2n}(1+\epsilon)^2 \leq \frac{l}{2n}(1+3\epsilon)$. 
Summing the independent Bernoulli errors $\vec{P}(\mathbf{e}_2)$ and $\vec{P}(\mathbf{f} + \mathbf{e'})$ gives an $n$-fold Bernoulli distribution with probability parameter $q = p(1-r) + r(1-p) = p + r - 2pr \leq p + r$.
This error parameter is larger than that of $\vec e_1$ by $(1-2p)r$.
By way of Lemma~\ref{lemma:Bern_convolution}, we may sample some $\vec g \in \Z_2^{n}$ such that $\vec e_1 + \vec g \sim \Ber(q)$ as well.

Presently, our transformed sample 
\begin{align}
    \left( \vec K, \; \begin{bmatrix}
        \vec z_1 + \vec g \\ \vec z_2
    \end{bmatrix} \right)
\end{align}
is nearly distributed as a valid $\symplpn$ structured sample, except (a) $\vec C$ has $\ell$ columns instead of $n$, and is not necessarily full rank, and (b) the error is Bernoulli instead of depolarizing.
The first is remedied by additional sampling: we construct a random completion $\vec L \in \Z_2^{2n \times (n-\ell)}$ of $\vec K \in \Z_2^{2n \times \ell}$.
More precisely, the $\ell \times \ell$ submatrix consisting of the first $\ell$ columns of $\vec K$ is marginally uniformly random, and so full rank with probability $1 - \negl(\ell) = 1 - \negl(n)$.
Thus, $\vec K$ is also full rank with probability $1 - \negl(n)$.
We can then iteratively sample the next column of $\vec L$ subject to being independent from and symplectically orthogonal to all previous columns as well as columns of $\vec K$.
We also extend the secret by sampling $\vec w \sim \Z_2^{n - \ell}$.
Now, our updated sample is of the form \begin{align}
    \left( \begin{bmatrix}
        \vec K \; \vec L
    \end{bmatrix} , \; \begin{bmatrix}
        \vec z_1 + \vec{Kw} + \vec g \\ \vec z_2 + \vec{Kw} 
    \end{bmatrix} \right) .
\end{align}
The last remaining issue is the distribution of the error, which must be transformed from Bernoulli to depolarizing.
By Corollary~\ref{corollary:bernoulli-depolarizing_duality}, this transformation can be done by sampling a random $\vec h \sim \Ber(q)^{\otimes n}$ and adding $\vec h$ to each of the two halves of the output.
The cost of such an operation is that the error probability increases by at most a factor of 3.
Since we can always increase the error to exactly a factor of 3 larger by Lemma~\ref{lemma:depolarizing_convolution}, we can assume that the probability increases to $3q$.
The final sample
\begin{align}
    \left( \begin{bmatrix}
        \vec K \; \vec L
    \end{bmatrix} , \; \begin{bmatrix}
        \vec z_1 + \vec{Kw} + \vec g + \vec h \\ \vec z_2 + \vec{Kw} + \vec h
    \end{bmatrix} \right) 
\end{align}
has, up to $\negl(n)$ corrections discussed above, the distribution of a $\symplpn(n, \;3q)$ instance.

Finally, suppose instead that the given $\lpn$ sample $(\vec A, \vec y)$ was unstructured, so that $\vec y \sim \Z_2^{\ell}$.
Then the above transformations would still map $\vec A$ to a valid $\symplpn$ encoding matrix (up to negligible corrections), but all the transformations performed on $\vec y$ leave its distribution invariant because $\vec y$ is uniformly random and we extend it by padding a uniformly random vector $\vec r$.
Thus, the unstructured $\lpn$ sample also maps to an unstructured $\symplpn$ sample.
We now run $\CO$ on the transformed sample and output its decision.
By the above, we will correctly decide the $\lpn$ instance if and only if we correctly decide the $\symplpn$ instance, unless a negligibly rare event occurs.
Consequently, our probability of success can decrease by at most $\negl(n)$ from $\CO$'s probability of success.
This completes the reduction.
\end{proof}

\subsection{From Symplectic $\lpn$ to $\lsn$}

Initially, the reasonable next step may appear to be a reduction from Decision $\symplpn$ to Decision $\lsn$.
By dimension counting alone, the natural way to perform this reduction is to put the $2n \times n$ $\symplpn$ encoding matrix $\vec B$ into the stabilizer/logical $\vec Z$ part of the classical representation of $\lsn$, and then generate the secret/logical $\vec X$ part ourselves.
While this reduction does go through exactly as stated, it is not helpful in showing that \emph{Search} $\lsn$ is hard, because we have no way of reducing Decision $\lsn$ to Search $\lsn$.
There is also no natural way to reduce Decision $\symplpn$ directly to Search $\lsn$.
This is because we generate the secret for $\lsn$ ourselves in the above reduction sketch, so a search oracle finding the $\lsn$ secret has no utility at all.
To circumvent this barrier, we give a stronger decision reduction---namely, into multi-sample $\lsn$---and then apply our decision-to-search reduction from Section~\ref{sec:self_reducibility} so as to finally reduce to Search $\lsn$ (and thereby also Decision and Search $\slsn$).

\begin{lemma}[Decision $\symplpn$ to Decision $\lsn^{\poly}$] \label{lem:symplpn_to_lsnpoly}
Let $k \geq 1$ be any positive integer, which may depend on $n$.
Suppose that there is an oracle $\CO$ which solves the classical representation of Decision $\lsn[k, n, p][\poly]$ with advantage $\eta$.
Then there exists an efficient classical algorithm which solves Decision $\symplpn[n, p]$ to advantage $\frac\eta{\poly(n)}$ using a single call to $\CO$.
\end{lemma}
\begin{proof}
A sample of Decision $\symplpn[n, p]$ consists of $(\vec{A}, \mathbf{z})$, where $\vec{A} \in \Z_2^{2n \times n}$ is a uniformly random full-rank symplectically orthogonal matrix.
In the structured case, $\mathbf{z} = \vec{A}\mathbf{x} + \mathbf{e}$, for uniformly random $\mathbf{x} \in \Z_2^n$, and $\Symp^{-1}(\vec e) \sim \CD_{p}^{\otimes n}$.
In the unstructured case, $\vec z \sim \Z_2^{2n}$.
Suppose $\CO$ uses $m = \poly(n)$ samples.
Recall that each sample is of the form $(\vec A_i, \;\vec B_i, \;\vec z_i)$, where in the structured case $\vec z_i = \vec A_i \vec r_i + \vec B_i \vec y + \vec e_i$ and in the unstructured case $\vec z_i \sim \Z_2^{2n}$.
The matrices $\vec A_i$ and vectors $\vec r_i$ are the \emph{junk} matrices and vectors, respectively.
They are re-generated for each sample and contain no information about the $\lsn$ secret $\vec y$.
Our basic strategy is to embed the $\symplpn$ data $(\vec A, \vec z)$ into the junk matrix and vector for one of the samples.
This is a strict departure from the typical forms of coding reductions, wherein the secret from one problem becomes the secret in the other.
This usual strategy is infeasible in our case, as the dimensions are completely mismatched, e.g. it is impossible to embed a secret of length $n$ into the secret of $\lsn$ if $k = 1$.
If $m = 1$, there is no ambiguity as to which sample we would embed into, and we would thus be essentially done.
However, since we have an arbitrary polynomial number of samples, we must thoughtfully choose which sample to embed the $\symplpn$ data into, so as to maintain a reduction.
To do so, we will proceed by a randomized version of the hybrid argument.
That is, we will choose a \emph{random} sample out of the $m$ possibilities in which to embed the $\symplpn$ data.

Sample a random index $j \in [m]$ and a random $\lsn$ secret $\mathbf{y} \in \Z_2^k$, and output
\begin{equation}
\left(\begin{bmatrix}
      \vec A_i \, | \, \vec B_i  
    \end{bmatrix}, \mathbf{z}_i\right) \text{ for } i\in[l],
\end{equation}
where 
\begin{enumerate}
\item If $i> j$, the $i$th sample is structured with secret $\vec y$. 
That is, $\vec{A}_{i} \in \Z_2^{2n \times n}$, $\vec{B}_{i} \in \Z_2^{2n \times k}$ are uniformly random symplectically orthogonal matrices such that $\begin{bmatrix}
      \vec A_i \, | \, \vec B_i  
    \end{bmatrix}$ is full rank, and $\mathbf{z}_i = \vec{A}_i\mathbf{r}_i+ \vec{B}_{i}\mathbf{y} + \mathbf{e}_i$ where $\mathbf{r}_i \in \Z_2^{2n}$ is uniformly random and $\mathbf{e}_i\sim \Symp(\mathcal{D}_{p}^{\otimes n})$.
\item If $i=j$, we embed the $\symplpn$ data into the sample.
Formally, $\vec{A}_{i} = \vec{A}$, $\vec{B}_{i} \in \Z_2^{2n \times k}$ is uniformly random and symplectically orthogonal such that $\begin{bmatrix}
      \vec A_i \, | \, \vec B_i  
    \end{bmatrix}$ is full rank, and $\mathbf{z}_{i} =\vec{B}_{i}\mathbf{y} + \mathbf{z}$,
\item If $i < j$, the $i$th sample is unstructured. 
That is, $\vec{A}_{i} \in \Z_2^{2n \times n}$, $\vec{B}_{i} \in \Z_2^{2n \times k}$ are uniformly random symplectically orthogonal matrices such that $\begin{bmatrix}
      \vec A_i \, | \, \vec B_i  
    \end{bmatrix}$ is full rank, and $\mathbf{z}_i \in \Z_2^n$ is uniformly random.
\end{enumerate}
We then apply the oracle $\CO$ to the resulting sample, and output its decision. 

Depending on whether the original sample $(\vec{A}, \vec{z})$ was structured or unstructured, the set of $m$ samples submitted to $\mathcal{O}$ has a different number of structured and unstructured samples. 
The first $j-1$ samples are always unstructured, by construction, and samples $j+1, \;\dots, \;m$ are always structured.
Meanwhile, the $j$th sample is structured if and only if the $\symplpn$ instance $(\vec{A}, \vec{z})$ is structured.

Let $p_k$ be the probability that $\mathcal{O}$ outputs \texttt{STRUCTURED} when the first $k$ of samples are structured, and the remaining $m-k$ are unstructured.
Then, the probability that $\mathcal{O}$ outputs \texttt{STRUCTURED} by our reduction when $(\vec{A}, \vec{z})$ is indeed structured is
\begin{equation}
\Pr[\CO = \text{\texttt{STRUCTURED}} \,|\, \text{\texttt{STRUCTURED}}] = \frac{1}{m} \sum_{k = 0}^{m-1} p_k,
\end{equation}
since $j$ is chosen randomly from 1 to $m$. Likewise, when $(\vec{A}, \vec{z})$ is unstructured,
\begin{equation}
\Pr[\CO = \text{\texttt{UNSTRUCTURED}} \,|\, \text{\texttt{UNSTRUCTURED}}] = \frac{1}{m} \sum_{k = 1}^{m} p_k,
\end{equation}
The difference in these two probabilities is
\begin{equation}\label{eq:symplpn_to_decision_distinguisher}
\left| \frac{1}{m}\sum_{k = 0}^{m-1} p_k - \frac{1}{m}\sum_{k = 1}^m p_k \right| = \frac{1}{m} | p_0 - p_m | .
\end{equation}
However, having $0$ unstructured samples and $m$ unstructured samples are precisely the two cases of Decision $\lsn[k, n, p][\poly]$, and by assumption $\CO$ succeeds in this problem with advantage $\eta$.
Therefore, $|p_0 - p_m| = \eta$, and as a result the advantage our reduction has on Decision $\symplpn(n, p)$, given by Eqn.~(\ref{eq:symplpn_to_decision_distinguisher}), is $\frac{\eta}m$.
\end{proof}

With these lemmas, we may at last establish the hardness of Search $\lsn[k, n, p]$.

\begin{theorem}[Average-case hardness of stabilizer decoding with any logical dimension]
Let $p = \Omega\left(n^{-(1-\epsilon)}\right)$, for $\e \in (0, 1)$ a constant, and let $k \geq 1$ be an arbitrary number of logical qubits (which may depend on $n$).
Suppose that there is an oracle $\CO$ which solves the classical representation of Search $\lsn[k, n, p]$ with probability $\frac{1}{2^k} + \frac{1}{\poly(n)}$.
Then there is a classical algorithm running in $\poly(n)$ time which solves Decision $\lpn[\lfloor \frac{np}6 \rfloor, \;2n, \;\frac{p}6]$ with advantage $\frac{1}{\poly(n)}$, calling $\CO$ as a subroutine.
By the equivalence of $\lsn$ and its classical representation, if $\CO$ instead solves the conventional formulation of Search $\lsn[k, n, p]$ there is also a quantum algorithm running in $\poly(n)$ time which solves Decision $\lpn[\lfloor \frac{np}6 \rfloor, \;2n, \;\frac{p}6]$ with advantage $\frac{1}{\poly(n)}$, calling $\CO$ as a subroutine.
\end{theorem}

\begin{proof}
Let $p' = \frac{p}6$ and $k' = \lfloor n p' \rfloor$; let $\e \in (0, 1)$ be a constant.
We have an immediate reduction from $\lpn[k', \;2n, \;p']$ to $\lpn[k', \;\lfloor 2n - (1+\e) n p' \rfloor , \;p']$, since we are always free to discard rows.
By Theorem \ref{thm:lpn_to_symplpn}, this latter problem reduces to Decision $\symplpn[n, \;3q]$, with $q = p' + \frac{k'}{2n}(1 + 3\epsilon) \leq \left(1 + \frac{1 + 3\e}2\right) p'$. 
Using Lemma~\ref{lemma:depolarizing_convolution}, we can assume that $q$ is exactly this upper bound by adding extra noise.
Then, by Lemma \ref{lem:symplpn_to_lsnpoly}, Decision $\symplpn[n, \;3q]$ reduces to Decision $\lsn[k, n, 3q][\poly]$. 
Finally, Decision $\lsn[k, \;n, \;3q][\poly]$ reduces to Search $\lsn[k, \;n, \;3q][\poly]$ by Theorem~\ref{thm:decision_to_search_lsn}, which reduces easily to Search $\lsn[k, \;n, \;3q]$ by ignoring all but the first sample. 
Choosing $\e = \frac13$, we have $q \leq 2p'$, so $3 q \leq 6 p' = p$.
The composition of these reductions proves the theorem.
\end{proof}

\section{Summary and Outlook} \label{sec:discussion}

This work explored the hardness of decoding quantum stabilizer codes from several perspectives.
We showed that quantum degeneracy and the inability to directly verify solutions caused the very definition of quantum decoding to lose its robustness to slight perturbations, e.g. finding the error is likely not equivalent (in an instance-by-instance manner) to recovering the logical state itself.
Nonetheless, all definitions are \textbf{NP}-complete.
We next studied the relation between $\lpn$ and $\lsn$; that is, between decoding random classical and random quantum codes.
The gap between these two problems is significant even conceptually.
We bridge this gap from both directions.
Quantumly, we introduced a completely classical representation of quantum decoding, which shares the same qualitative structure as $\lpn$ while also being equivalent to $\lsn$.
Classically, we introduced a symplectic variant of $\lpn$ which guarantees that the generator matrix satisfies a symplectic orthogonality condition naturally baked into quantum codes.
We then reduced $\lpn$ to its symplectic variant, and reduced that to $\lsn$.

At the same time, we examined the hardness of $\lsn$ from the perspective of self-reducibility.
We give search-to-decision reductions (and, in the average case, decision-to-search), but prove at the same time that several reasonable notions of a random self-reduction from the worst-case decoding problem to the average case $\lsn$ do not exist.
One interpretation of this distinction is that search-to-decision reductions can be accomplished with only local operators, whereas a random self-reduction necessitates a great deal of entanglement.
The entangling operations required, however, act in a highly uniform way (namely, they mix all Paulis at approximately the same rate), which ultimately bars a reduction.

A natural question stemming from this work is whether the converse reduction from $\lsn$ to $\lpn$ holds; namely, is there a reduction from $\lsn(k, n, p)$ to $\lpn(\Theta(np), \;\Theta(n), \;\Theta(p))$?
The existence of such a reduction would help shed light on the the exact relationship between quantum and classical decoding more generally.

There are also several open questions centered around strengthening both the positive and negative results we give about reductions among quantum decoding.
For example, we give a reduction from Search $\slsn$ to Search $\lsn$ only when $k = O(\log(n))$ or the solver for $\lsn$ succeeds with a very large probability, namely substantially more than $\frac{1}{2}$ (as compared to the usual $\frac{1}{2^k} + \frac{1}{\poly(n)}$). 
Whether or not this reduction can be improved remains open.
Likewise, we give search-to-decision reductions in the worst case with the requirement of a $\qncp(k, n, w')$ solver that works for each choice of $w' \leq w$.
It is unclear whether a reduction exists with a weaker assumption, e.g. if one requires only a solver with weight exactly $w$.
From the random self-reduction front, an important question is whether one can prove with no assumptions at all that no restricted notion of a random self-reduction exists, or whether such a reduction does exist but requires an exotic algorithmic approach.

Finally, from the point of view of applications, there is a long history of $\lpn$ as the basis for a great deal of cryptographic protocols~\cite{10.1007/978-3-642-27660-6_9}. 
A natural question concerns whether we can also base interesting cryptography on hard problems in quantum, rather than classical error correction. 
If cryptographic schemes based on $\lpn$ can be adapted to $\symplpn$ or $\lsn$, then new classical cryptography could be built from natural quantum-inspired problems. 
Finally---and perhaps, more interestingly---cryptography based
on these assumptions may turn out be harder to break than any cryptography based on $\lpn$.

\subsection*{Acknowledgments}
\addcontentsline{toc}{section}{Acknowledgments}
We thank Thiago Bergamaschi, Alexandru Gheorghiu, Zhiyang (Sunny) He, Yihui Quek, Justin Raizes, Peter Shor, Kabir Tomer, Adam Wills, and Henry Yuen for illuminating discussions.
ABK is supported by the Engineering and Physical Sciences Research Council grant number EP/Z002230/1: (De)constructing quantum software (DeQS).
JZL is supported in part by a National Defense Science and Engineering Graduate (NDSEG) fellowship.
AR is supported by a Caltech Summer Undergraduate Research Fellowship (SURF).
JZL and AP are supported in part by the U.S. Department of Energy,
Office of Science, National Quantum Information Science Research Centers, Co-design
Center for Quantum Advantage (C2QA) under contract number DE-SC0012704.
VV is supported in part by NSF CNS-2154149 and a Simons Investigator Award.

\printbibliography
\appendix


\section{Nonexistence of unstructured classical random self-reductions} \label{sec:app-classical_w2a_no_go}
We here sketch a proof that there is no random self-reduction into $\lpn(k, n, p)$ from the worst-case version, known as $\ncp(k, n, w)$, which holds for all classical linear codes.
In this case, the reduction distribution is given by a random vector $\mb{r} \in \Z_2^n$, drawn from a distribution $\boldsymbol{R}_n$, which multiplies the code on the left.

\begin{lemma} \label{lemma:no_NCP_to_NCP}
    Let $\boldsymbol{A} \in \Z_2^{n \times k}$ be a generator matrix and $\mb{e} \in \Z_2^n$ an error with $|\mb{e}| \leq w$, and $\mb{s} \in \Z_2^{k}$ be the logical bitstring.
    Let $\mb{r}$ be drawn from any distribution over $\F_{2}^{n}$, transforming $(\boldsymbol{A}, \mb{e}) \longrightarrow (\mb{r}^T \boldsymbol{A}, \;\mb{e} \cdot \mb{r})$. 
    Suppose for any $\boldsymbol{A}$, the law of $\mb{r}^T \boldsymbol{A}$ is within TV distance $\e$ of $\Unif(\Z_2^k)$. Then the law of $\mb{e} \cdot \mb{r}$ is also $\e$-close in TV distance to $\Ber(\frac{1}{2})$.
\end{lemma}

\begin{proof}
    Fix some choice of $\mb{e}$. Choose any $C$ whose first row is $\mb{e}$. By monotonicity of TV distance, $\TV(\mb{e} \cdot \mb{r}, \;\Unif(\Z_2)) \leq \TV(\mb{r}^T \boldsymbol{A}, \;\Unif(\Z_2^n)) \leq \e$. But $\Unif(\Z_2) = \Ber(\frac{1}{2})$ which completes the proof.
\end{proof}

\begin{theorem}
    There is no strong random self-reduction from $\ncp(k, n, w)$ to $\lpn(k, n, p)$ such that $p = \frac{1}{2} - \frac{1}{\poly(n)}$.
\end{theorem}
\begin{proof}
    Suppose that such a reduction distribution $\boldsymbol{R}_n$ exists, so that \begin{align}
        \mathop{\TV}\limits_{\mb{r} \sim \boldsymbol{R}_n}((\mb{r}^T \boldsymbol{A}, \mb{e} \cdot \mb{r}), \;\Unif(\Z_2^{n \times k}) \otimes \Ber(p)) \leq \e = \negl(n) .
    \end{align} 
    By monotonicity, the TV between just the codes is $\leq \e$. Hence, by Lemma~\ref{lemma:no_NCP_to_NCP}, the TV distance between the smoothed error and $\Ber(\frac{1}{2})$ is $\leq \e$. However, \begin{align}
        \TV\left(\Ber\left(\frac{1}{2} - \d\right), \;\Ber\left(\frac{1}{2}\right)\right) = \frac{1}{2} \br{\abs{p - (p-\d)} + \abs{(1-p) - (1-p+\d)}} = \d .
    \end{align}
    For $p = \frac{1}{2} - \frac{1}{\poly(n)}$, the above and the triangle inequality imply that \begin{align}
        \frac{1}{\poly(n)} & \leq \TV\left(\Ber(p), \Ber\left(\frac{1}{2}\right)\right) \\
        & \leq \mathop{\TV}\limits_{\mb{r} \sim \boldsymbol{R}_n}(\mb{e} \cdot \mb{r}, \;\Ber(p)) + \mathop{\TV}\limits_{\mb{r} \sim \boldsymbol{R}_n}\left(\mb{e} \cdot \mb{r}, \;\Ber\left(\frac{1}{2}\right)\right) \\
        & \leq \epsilon + \epsilon = 2 \epsilon = \negl(n) ,
    \end{align}
    a contradiction for sufficiently large $n$.
\end{proof}

\end{document}